\documentclass[aps,prb,superscriptaddress,twocolumn,preprintnumbers,nocompress,floatfix,nofootinbib,bibnotes]{revtex4-2}
\pdfoutput=1
\usepackage{mathtools}
\usepackage{amsthm}
\usepackage{amsfonts,amssymb }
\usepackage{graphicx} 
\usepackage{physics}
\usepackage{xcolor}
\usepackage{dsfont}
\usepackage{bm}
\usepackage{comment}
\usepackage[hidelinks]{hyperref}
\usepackage{soul}
\usepackage{xstring}

\let\RevTeXOriginalHref\href
\newcommand{\RevTeXNoTitleHref}[2]{%
  \edef\HrefText{\detokenize{#2}}%
  \edef\TitleA{\detokenize{\bibinfo {title}}}%
  \edef\TitleB{\detokenize{\bibfield {title}}}%
  \IfSubStr{\HrefText}{\TitleA}
    {#2}
    {%
      \IfSubStr{\HrefText}{\TitleB}
        {#2}
        {\RevTeXOriginalHref{#1}{#2}}%
    }%
}
\AtBeginEnvironment{thebibliography}{\let\href\RevTeXNoTitleHref}
\AtEndEnvironment{thebibliography}{\let\href\RevTeXOriginalHref}

\newtheorem{theorem}{Theorem}

    \newtheorem{corollary}{Corollary}
    \newtheorem{lemma}{Lemma}
    \newtheorem{prop}{Proposition}
    
\theoremstyle{definition}
  \newtheorem{definition}{Definition}

\begin{document}

\preprint{MIT-CTP/6117}

\title{Sustained growth of quantum circuit complexity in many-body Hamiltonian dynamics}
\author{Wonjun Lee}
\email{wonjun98@mit.edu}
\affiliation{College of Natural Sciences, Korea Advanced Institute of Science and Technology, Daejeon, 34141, Republic of Korea}
\affiliation{Research Laboratory of Electronics, Massachusetts Institute of Technology, Cambridge, Massachusetts, 02139, United States}
\affiliation{MIT Center for Theoretical Physics -- a Leinweber Institute, Massachusetts Institute of Technology, Cambridge, Massachusetts, 02139, United States
}
\author{Sa\'ul Pilatowsky-Cameo}
\email{saulpila@mit.edu}
\affiliation{MIT Center for Theoretical Physics -- a Leinweber Institute, Massachusetts Institute of Technology, Cambridge, Massachusetts, 02139, United States
}
\author{Soonwon Choi}
\email{soonwon@mit.edu}
\affiliation{MIT Center for Theoretical Physics -- a Leinweber Institute, Massachusetts Institute of Technology, Cambridge, Massachusetts, 02139, United States
}

\begin{abstract}
The quantum circuit complexity of an evolving many-body quantum system  is believed to exhibit a sustained growth, maintained for timescales much longer than the onset of thermalization. Most previous works have focused on models which violate energy conservation, such as random unitary circuits. Here we study generic, local time-independent Hamiltonian dynamics.
We unconditionally prove that for generic local Hamiltonians and typical initial product states at high effective temperature, the robust quantum circuit complexity must grow over a very long period of time,  attaining an exponentially large value at late times.
Our approach relies on two structural properties that we prove rigorously: (i) generic local Hamiltonians satisfy generalized spectral no-resonance conditions of arbitrary order, and (ii) typical high-temperature product states are effectively supported in exponentially many energy eigenstates. These properties have been widely assumed without proof in prior works. As corollaries of this result, we show the late-time state displays robust volume-law entanglement that is irremovable by polynomial-size circuits, and we establish a no fast-forwarding result for generic local Hamiltonians. We also lower bound the complexity  of preparing sufficiently large subsystems with local quantum channels, in contrast with sufficiently small thermalizing regions which we show retain low complexity at late times. 
\end{abstract}

\maketitle

\section*{Introduction}
Thermalization is conventionally thought to be the end of quantum dynamics. However, this fundamental understanding has been upended by the realization that a thermal wavefunction can retain internal structure that continues to evolve well after local observables equilibrate and entanglement saturates. This hidden structure may manifest in novel physical phenomena such as deep thermalization~\cite{Ho_2022,Mark2024} and the emergence of pseudorandomness in dynamics~\cite{Roberts_2017,hunterjones2019,Haferkamp2022randomquantum,Harrow2023,Pilatowsky_Cameo_2023}. If thermalization does not mark the end of dynamical evolution, what quantity continues to grow after entanglement entropy has saturated? 

\begin{figure}[t]
    \centering
    \includegraphics[width=\linewidth]{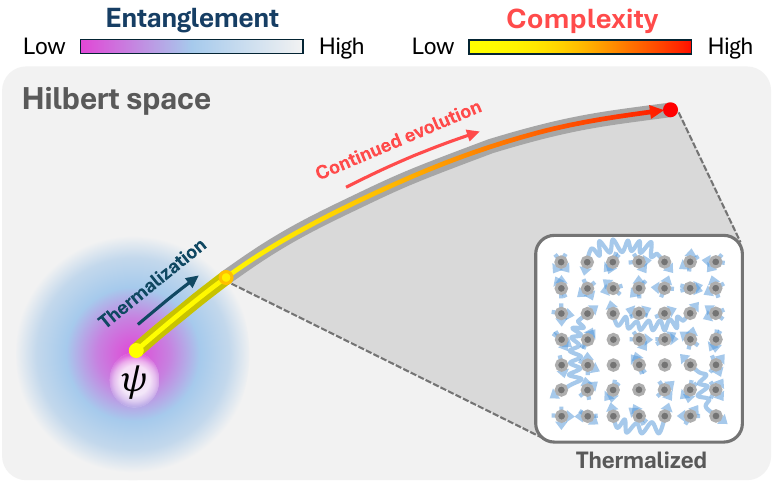}
    \caption{\textbf{Post-thermalization dynamics.}  An initial product state $\psi$ evolves under a local Hamiltonian. At thermalization, entanglement saturates and the state becomes locally indistinguishable from the Gibbs state, but nevertheless the state continues to evolve. We show that quantum circuit complexity continues to grow for a very long period of time, even after the onset of thermalization.} 
    \label{fig:1}
\end{figure}

Quantum circuit complexity has been proposed as a  measure to quantify post-thermalization evolution: it counts how many elementary quantum operations are needed to prepare the evolving many-body state. Motivated by holographic duality, in which the growth of a black-hole interior is interpreted as continued evolution of the dual state on the boundary~\cite{Susskind2016,Stanford2014,Brown2016}, Brown and Susskind conjectured that the circuit complexity of a chaotic quantum system grows linearly in time for an exponentially long time~\cite{Brown2018}, well after the system has reached thermal equilibrium (see Fig.~1). Beyond holography, this conjecture  has broad implications for quantum information science, including quantum simulation~\cite{Atia2017}, quantum pseudorandomness~\cite{Ji2018,cui2025hamiltonian,mao2025} and quantum computational complexity theory~\cite{aaronson2016}. More fundamentally, it allows one to define an “arrow of time” for a quantum many-body system, even after reaching thermal equilibrium.
  
Establishing the growth of quantum circuit complexity in dynamics requires ruling out a \textit{shortcut}, by which a late-time quantum state could be prepared by a circuit much shorter than the physical evolution that produced it. This is in general computationally intractable, as the number of possible circuits grows exponentially in both time and system size. A similar fundamental difficulty appears in classical circuit complexity: proving strong lower bounds for explicit Boolean functions remains a major open problem in theoretical computer science~\cite{Jukna2012}, where even super-polynomial lower bounds for a Boolean-function family in $\mathsf{NP}$ would imply $\mathsf{P}\neq\mathsf{NP}$. 

Despite the challenges, tremendous progress toward proving the Brown-Susskind conjecture has been made from several fronts, including Nielsen's geometric formulation of complexity~\cite{Nielsen2006,Balasubramanian2020}, and rigorous proofs under the evolution of random circuits~\cite{Brand_o_2016,Brandao2021,Haferkamp2022,chen2024} and other stochastic time-dependent Hamiltonian dynamics~\cite{Jian2023,Oszmaniec2024}. Such rigorous proofs represent important milestones, but remain unsatisfactory for a fundamental reason:  they focus on models which violate the conservation of energy. This conservation law is of fundamental importance, being the basic constraint obeyed by any system evolving under a fixed Hamiltonian.  Indeed, the original holographic motivation for the Brown-Susskind conjecture concerns energy-conserving evolution, with the boundary state evolving under a certain conformal-field-theory Hamiltonian~\cite{Stanford2014,Brown2016,Brown2016complexity}. In such a scenario, existing results can only establish lower bounds on late-time complexity for fine-tuned model Hamiltonians, under complexity-theoretic assumptions~\cite{Susskind2016,aaronson2016,susskind2018,bohdanowicz2017}.

In this work, we prove that for any typical initial product state above a constant threshold temperature evolving under a generic local extensive Hamiltonian, the quantum circuit complexity must display a sustained growth for an exponentially long time, eventually reaching an exponentially large value at late times, yet we do not show this growth is linear.
We establish our results unconditionally, without unproven assumptions on the Hamiltonian or initial state, other than physical considerations: locality, extensivity, and high temperature.

\section*{Growth of circuit complexity}
We consider a system of $n$ qubits placed in a Euclidean lattice of arbitrary dimension. We define the $\varepsilon$-robust circuit complexity $\mathcal{C}_{\varepsilon}(\rho)$ of any state $\rho$ as the minimum number of $2$-qubit channels necessary to prepare $\rho$ up to trace-distance error $\varepsilon$, with access to arbitrarily many unentangled ancilla qubits. 

The system is initialized in a product state $\ket{\psi}$ (trivial complexity). We ask how much complexity increases purely as a result of the dynamics, which are generated by a fixed Hamiltonian $H$ via unitary evolution
\begin{equation}
    \ket{\psi(t)}=e^{-iHt}\ket{\psi}.
\end{equation}
We impose two physical constraints on $H$: $H$ is {\it local}, meaning each term has bounded range in the lattice, and {\it extensive}, meaning it contains a term of nontrivial strength acting on each local region. 

Not all local extensive Hamiltonians lead to sustained growth of complexity: counterexamples include fast-forwardable Hamiltonians, such as all-commuting local models~\cite{Atia2017}. We show that such Hamiltonians form a measure zero set within the space of local Hamiltonians, so they are fine-tuned and not robust to arbitrarily small generic local perturbations. 

Our statement holds for all but an exponentially small fraction of initial product states drawn from a family of ensembles parametrized by temperature. The simplest is the ensemble of random product states, with each qubit distributed according to the local Haar measure,
\begin{equation}
\label{eq:Haarens}
  \mathcal{E}_0=\Big\{\bigotimes_{i=1}^n \ket{\psi_i} \, \Big | \ket{\psi_i}\sim \mathrm{Haar}(\mathbb C^2)\Big\}.
\end{equation}
 Most states in $\mathcal{E}_0$ have an effective infinite temperature, corresponding to an inverse temperature $\beta=0$. This effective temperature is defined by matching the energy expectation value $\expval{H}{\psi}$ to the thermal energy $E(\beta)=\tr(g_\beta H )$ of the Gibbs state $g_\beta=\exp(-\beta H)/\tr(\exp(-\beta H))$. To target a finite temperature, we also consider the {\it microcanonical product-state ensemble}~\cite{huang2024randomproduct} of product states whose energy expectation lies within a shell of width $\delta \sqrt{n}$ around $E(\beta)$:
\begin{equation}
    \label{eq:ensfintemp}
\mathcal{E}_{\beta,\delta}=\Big\{{\ket{\psi}\sim \mathcal{E}_0}\Big |\,
        |\expval{H}{\psi}-E(\beta)|<\delta\sqrt{n}\Big\}.
\end{equation}
The parameter $\delta$ can be chosen arbitrarily as long as it is above a constant $\delta_*$~\cite{Supple}. 
Since energy is extensive, all states in $\mathcal{E}_{\beta,\delta}$ have the same energy density $\lim_{n \to \infty} E(\beta)/n$
in the thermodynamic limit. Our results hold for $|\beta|$ below a constant threshold $\beta_\mathrm{c}$~\cite{Supple} depending on the lattice dimension and local energy scales, but not on system size.

We bound the {\it subsystem complexity}, the complexity of the reduced state $\psi_A(t)\coloneqq \tr_{A^c}(\dyad{\psi(t)})$, for any sufficiently large subsystem $A$; taking $A$ to be the full lattice yields a statement about global complexity. 
\begin{theorem}[Exponential late-time quantum circuit complexity]\label{thm:theorem1}
  For any generic local extensive Hamiltonian, all but an exponentially small fraction of initial states $\psi\sim\mathcal{E}_{0}$ or $\mathcal{E}_{\beta,\delta}$ ($|\beta|< \beta_\mathrm{c}$), and all but a double-exponentially small fraction of time $t\geq 0$, the $\varepsilon$-robust subsystem circuit complexity is exponentially large
\begin{equation}
\label{eq:theorem1}
\mathcal{C}_\varepsilon(\psi_A(t))\geq  \exp \Omega(n).
\end{equation}
This statement holds for any $\varepsilon< 1-\exp(-\Theta(n))$ and any subsystem $A$ satisfying $|A|\geq c n$ for a certain constant $c<1$. 
\end{theorem}
\noindent Here $O,\Omega,\Theta$ denote asymptotic upper, lower, and tight scalings up to constant factors; see the Supplementary Material (SM)~\cite{Supple} for a statement with explicit constants. {\it Generic} means the statement fails only for a zero measure set of local Hamiltonians, and the fraction of time is the $T\to\infty$ limit of the uniform probability over $[0,T]$, which exists by the quasiperiodicity of unitary evolution and the Kronecker-Weyl theorem. A finite-time version of Theorem~\ref{thm:theorem1} is given in Methods.

We emphasize that this result has been proven without any additional assumptions. The logical structure is illustrated in Fig.~\ref{fig:2}. The key idea, explained below, combines well-known counting arguments~\cite{haah2025}---improved to allow $\varepsilon$ exponentially close to unity---with new results on the spectral properties of generic local extensive Hamiltonians. We discuss three implications of Theorem~\ref{thm:theorem1}.

\begin{figure}
    \centering
    \includegraphics[width=\linewidth]{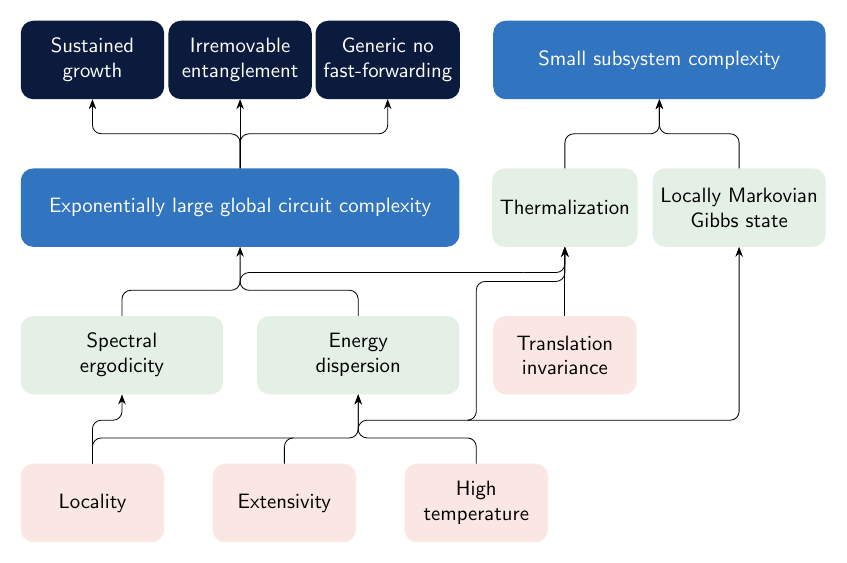}
    \caption{\textbf{Logical structure of the main results.} Orange boxes denote assumptions, green boxes denote derived properties, blue boxes denote main results in this work, and dark-blue boxes denote implications of our results. 
    }
    \label{fig:2}
\end{figure}
\textit{Sustained growth}---Theorem~\ref{thm:theorem1} guarantees that the complexity eventually grows to an exponentially large value at late times, even though it does not describe how fast. Nevertheless, since complexity cannot grow faster than linearly by recent results on efficient Hamiltonian simulation~\cite{Childs2019,Haah2023}, the growth must be sustained over an exponentially long period:

\begin{corollary}[Sustained growth of complexity over  exponentially long time]\label{thm:cor1}
    For $H$ and $\psi$ as in Theorem~\ref{thm:theorem1}, before the error-averaged complexity $\bar{\mathcal{C}}_{\varepsilon,\Delta\varepsilon}(\psi(t))\coloneqq\frac{1}{2\Delta\varepsilon}\int_{\varepsilon-\Delta \varepsilon}^{\varepsilon +\Delta \varepsilon}\dd{\varepsilon'}\mathcal{C}_{\varepsilon'}(\psi(t))$ saturates to an exponentially large value, its growth is sustained over an exponentially long time, in the sense that for any constant $\Delta t>0$, there are exponentially many times $t_1<t_2<\cdots <t_M$ separated by at least $\Delta t$ such that
    \begin{equation}
         \bar{\mathcal{C}}_{\varepsilon,\Delta\varepsilon}(\psi(t_{1}))<  \bar{\mathcal{C}}_{\varepsilon,\Delta\varepsilon}(\psi(t_{2}))< \cdots< \bar{\mathcal{C}}_{\varepsilon,\Delta\varepsilon}(\psi(t_{M})).
    \end{equation}
    Here $\Delta\varepsilon$ can be chosen exponentially small and $\varepsilon$  exponentially close to $1$.
\end{corollary}
\noindent 
Corollary~\ref{thm:cor1} shows that the wavefunction retains nontrivial evolution, quantified by growing complexity, for exponentially long times, far beyond thermalization timescales. 

\textit{Irremovable entanglement}---Theorem~\ref{thm:theorem1} also implies that entanglement must saturate to a volume law under the dynamics, and that this late-time entanglement is computationally hard to remove. We quantify entanglement by the average subsystem entropy $\bar S(\psi)=\frac{1}{M}\sum_{m=1}^M S(\psi_{A_m})$  over a partition of the lattice into $M$ regions $A_1,\dots,A_M$, with average size $\bar a=\tfrac{1}{M}\sum_{m=1}^M |A_m|=\frac{n}{M}$.

\begin{corollary}[Irremovable volume-law entanglement] 
For ${\psi(t)}$ as in  Theorem~\ref{thm:theorem1} and $\max_m |A_m|\leq O(\bar a)$, the average subsystem entropy of $\psi(t)$ follows a volume law, \textit{i.e.}, $\bar{S}(\psi(t))\geq \Omega(\bar{a})$. Furthermore, this volume law is irremovable with polynomial circuits: $\psi(t)$ cannot be mapped to a state $\phi$ satisfying $\bar S(\phi)\leq O(\bar a^{1-\eta})$ by any quantum circuit with $\leq\exp(O(n^{1-\eta/4}))$ two-qubit gates for any constant $\eta>0$ and $M=\Theta(n^{\kappa})$ with $\eta\leq \kappa\leq 1/4$.
    \label{cor:cor2}
\end{corollary}

\noindent
Corollary~\ref{cor:cor2} follows from two facts: any low-entangled state can be prepared by a polynomial circuit~\cite{Schon2005}, and if a polynomial circuit could disentangle the late-time state, the same circuit run in reverse would efficiently prepare it, contradicting Theorem~\ref{thm:theorem1}.

\textit{No fast-forwarding}---Finally, Theorem~\ref{thm:theorem1} implies an obstruction to the simulation of late-time dynamics. We say $H$ can be \textit{fast-forwarded} with complexity $G$ if for any time $t$ and initial state $\phi$, some circuit of at most $G$ two-qubit gates maps $\phi$ to within small constant trace distance of $\phi(t)$. This definition allows arbitrary classical preprocessing to compile the circuit, whereas the conventional algorithmic interpretation of fast-forwarding also requires efficient compilation~\cite{Atia2017}.
Even with such a lenient definition, we establish a no-go theorem for fast-forwarding.
\begin{corollary}[Generic no fast-forwarding]
    Any generic extensive local Hamiltonian cannot be fast-forwarded with complexity $G=\mathrm{poly}(n)$. Therefore, Hamiltonians allowing such fast-forwarding---including all-commuting local, 1D finite-range quadratic fermionic, or number-conserving quadratic bosonic Hamiltonians that are extensive---must form a measure zero set in the ensemble of all extensive local Hamiltonians.
\end{corollary}

Corollaries 1-3 are proven in the SM~\cite{Supple}.

\section*{Complexity of small subsystems}
We have shown that quantum circuit complexity grows over at least exponentially long times, far exceeding conventional thermalization times. This result, however, seems counterintuitive: in practice, we do not observe any interesting dynamical processes after a system thermalizes. Why is the sustained growth of complexity guaranteed by Theorem~\ref{thm:theorem1} effectively {\it invisible}? A possible explanation is that in any realistic setting, an observer has access only to local observables, and the large complexity is not manifested locally. For any small subsystem, the complexity only grows to a subexponential value, saturated once the system thermalizes. We place the above idea on rigorous mathematical footing by leveraging recent efficient algorithms for preparing Gibbs states~\cite{Brandao2019,chen2025}.

\begin{theorem}[Thermalizing systems must have small complexity for small subsystems~\cite{chen2025}]
    \label{thm:low-c-from-therm}
    Let $H$ be a Hamiltonian, $\psi$ be an initial state at sufficiently high temperature, and $A$ be a contiguous region. For $0<\varepsilon<1$, suppose $\psi$ thermalizes on $A$ with trace-distance error $\varepsilon_{\mathrm{therm}}=2\varepsilon/3$. Then, for all but an exponentially small fraction of time $t\geq 0$, the complexity is upper bounded as
    \begin{equation*}    
        \mathcal{C}_\varepsilon(\psi_A(t))\leq  \exp(\mathrm{polylog}(|A|)).
    \end{equation*}
\end{theorem}
\noindent Here, $\mathrm{polylog}$ means polylogarithmic growth, which upon exponentiation produces quasipolynomial (subexponential) scaling. 

The statement of Theorem~\ref{thm:low-c-from-therm} is conditional on thermalization: an initial state $\psi$ \textit{thermalizes} on a region $A$ if the reduced state $\psi_A(t)$ equals the reduced Gibbs state in $A$ up to small trace-distance error $\varepsilon_{\mathrm{therm}}$ except for an exponentially small fraction of time, which is widely expected for physical systems. Thermalization is rigorously proven for translation-invariant systems with non-degenerate spectral gaps at high temperature, for constant-size subsystems~\cite{huang2020,Pilatowsky-Cameo2025}. In Methods, we extend this result to subsystems of size $|A|\leq O(\log n)$, and, in the SM~\cite{Supple}, conditioned on a weak eigenstate thermalization hypothesis, to larger subsystems with $|A|<n/2$, even without translation invariance. Related bounds on the late-time complexity of small subsystems have already been established for random circuit dynamics~\cite{fan2025,haah2025}.

\section*{Proof of Theorem~\ref{thm:theorem1}}

\begin{figure}
    \centering
    \includegraphics[width=0.9\linewidth]{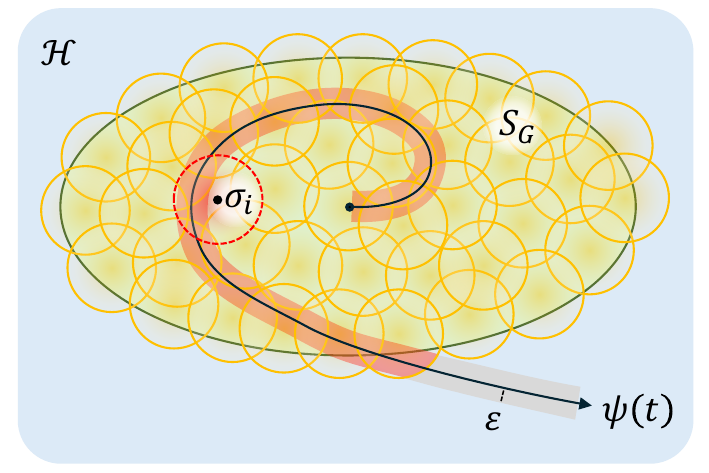}
    \caption{
    \textbf{Main counting argument.} Set  $S_G$ of states that can be prepared by a sequence of $G$ two-qubit channels (green) and its covering net with states $\sigma_i$ (orange circles). The system is initialized in a product state (black circle) which evolves in time (arrow). At time $t$, the circuit complexity is $\leq G$ only if $\ket{\psi(t)}$ is within the neighborhood of some $\sigma_i$. We bound the fraction of time that $\ket{\psi(t)}$  can spend near a particular $\sigma_i$ (red) to be double-exponentially small, so a union bound guarantees that the total time spent inside the green region is double-exponentially small, as long as $G$ is exponential.
    }
    \label{fig:3}
\end{figure}

The proof of Theorem~\ref{thm:theorem1} is a counting argument, illustrated in Fig.~\ref{fig:3}. Consider the set $S_G$ of states that can be prepared by a circuit containing at most $G$ two-qubit channels. How many states does $S_G$ contain? This set is continuously parametrized, so it is infinite, but by coarse graining --- identifying states that are close in trace distance --- we may distill a finite set.   Mathematically, we construct a finite {\it covering net} $\mathcal{N}$ such that any state in $S_G$ is close to some state in $\mathcal{N}$, which can be done with a total of $|\mathcal{N}|\lesssim G^{O(G)}$ states~\cite{haah2025}. Crucially, any state having robust complexity upper bounded by $G$ is close to some state in $\mathcal{N}$. 

Now, we ask, how much time can $\psi_A(t)$ spend near any particular state? If we could show that $\psi_A(t)$ spends at most a double-exponentially small fraction of time $q_\sigma=\exp(- e^{\alpha n})$ near any state $\sigma$,  then the total fraction of time it spends near any state in $\mathcal{N}$ (and hence close to any state with complexity less than $G$) would not exceed $|\mathcal{N}|\cdot q_\sigma\lesssim G^{O(G)}\times  \exp(- e^{\alpha n})$, which is still double-exponentially small for $G=\exp(\alpha'n)$ with $\alpha'<\alpha$. We thus aim to show that $q_\sigma$ is indeed double-exponentially small. 

To guarantee sufficient ergodicity of the trajectory, we identify two properties of the Hamiltonian and initial state: \textit{spectral ergodicity} and \textit{energy dispersion}. 
Spectral ergodicity is a property of the energy eigenvalues $E_\mu$ of the Hamiltonian. It prevents short-time state-to-state recurrences which would force the trajectory to spend much time near the initial state. For illustration, a harmonic oscillator evolves periodically because its energy levels are equally spaced, so the dynamical phases oscillate in resonance.
Spectral ergodicity forbids such resonances, in a strong sense:
\begin{definition}[Spectral ergodicity]
\label{def:spec-erg}
A Hamiltonian satisfies spectral ergodicity if for any integer $k\geq 1$, sums of $k$ energy levels are equal  $\sum_{i=1}^k E_{\mu_i} =\sum_{i=1}^k E_{\nu_i} $ only when $(\mu_1,\dots,\mu_k)$ is a permutation of $(\nu_1,\dots,\nu_k)$.
\end{definition}
 \noindent For $k=1$, this is the non-degeneracy condition ($E_\mu=E_\nu$ only if $\mu=\nu$); for $k=2$, it is the non-degenerate gap condition~\cite{Reimann2008,Linden2009}: $E_{\mu_1}-E_{{\nu_1}} = E_{\nu_2}-E_{{\mu_2}}$ only if $(\mu_1,\mu_2)=(\nu_1,\nu_2)$ or $(\mu_1,\mu_2)=(\nu_2,\nu_1)$. The general property for arbitrary $k$ is known as the {\it $k$-th no-resonance condition}~\cite{Kaneko2020,Mark2024,Riddell2023}.

The second condition, energy dispersion, is a property of the initial state $\ket{\psi}$ when expanded in the energy eigenbasis $\{\ket{\mu}\}_\mu$. To guarantee that the state explores a large fraction of the Hilbert space, we need it to be supported on many energy eigenstates with non-negligible population $p_\mu=\abs{\braket{\mu}{\psi}}^2$. \begin{definition}[Energy dispersion]
    We say $\ket{\psi}$ is energy dispersed if its {\it inverse participation ratio} or  {\it effective dimension}~\cite{Linden2009} $1/\sum_\mu p_\mu^2\geq 2^{\gamma n}$ is exponentially large with a constant $\gamma>0$. 
\end{definition} 

In Methods, we prove that under these two conditions, $q_\sigma$ is double-exponentially small, which implies the following result (see the SM~\cite{Supple} for a full proof and explicit constants).
\begin{theorem}[Exponential complexity from spectral ergodicity and energy dispersion]
   \label{thm:theorem2}
    For any Hamiltonian satisfying spectral ergodicity and any initial state $\psi$ which is energy dispersed,  the circuit complexity is exponentially large,
    \begin{equation}
        \mathcal{C}_\varepsilon(\psi_A(t))\geq  \exp \Omega(n),
    \end{equation}
    except for a double-exponentially small fraction of time $t\geq 0$, 
     for a subsystem $A$ and $\varepsilon$ as in Theorem~\ref{thm:theorem1}.
\end{theorem}
\noindent We note that Theorem~\ref{thm:theorem2} applies to arbitrary, potentially nonlocal Hamiltonians and entangled initial states. 
The proof techniques of Theorem~\ref{thm:theorem2} also have implications for recent circuit complexity lower bounds for the Scrooge ensemble~\cite{mcginley2025,Supple}. To complete the proof of Theorem~\ref{thm:theorem1}, we show that the two conditions of Theorem~\ref{thm:theorem2} are generically satisfied.

\textit{Spectral ergodicity is generic}--- The no-resonance conditions, predominantly the non-degenerate gap condition $(k=2)$, appear as assumptions in most works that rigorously study thermalization and equilibration in quantum many-body dynamics~\cite{Reimann2008,Linden2009,Mark2023,Mark2024,huang2020,Pilatowsky-Cameo2025,Riddell2023,Kaneko2020}. One notable exception is Ref.~\cite{Huang2021} which proved the non-degenerate gap condition for generic local Hamiltonians. We generalize it to the $k$-th no-resonance condition with arbitrary $k$. Specifically, we parametrize any local Hamiltonian with range $r$ in the Pauli-string basis 
\begin{equation}
\label{Eq:HamDef}H=\sum_{\mathrm{range}(P)\leq r} c_P P,
\end{equation}
where $\mathrm{range}(P)$ denotes the distance between the furthest sites that the Pauli string $P$ acts non-trivially on (e.g., nearest-neighbor interactions have range $1$).

\begin{prop}[Spectral ergodicity is generic]
   \label{thm:generic-ergodic} The set of local Hamiltonians $H$ which violate spectral ergodicity (Definition~\ref{def:spec-erg}), parametrized by tuples of coefficients $(c_P)_{\mathrm{range}(P)\leq r}$ in Eq.~\eqref{Eq:HamDef}, has zero measure for any fixed $r\geq 1$. 
\end{prop}
Proposition~\ref{thm:generic-ergodic} is proven by noting that the $k$-th no-resonance condition is violated only when a certain polynomial in the coefficients $c_P$ vanishes~\cite{Keating2015,Huang2021,Riddell2023,Supple}. By constructing a local Hamiltonian without $k$-th order resonances, we prove that this polynomial is not identically zero, and hence its roots --- coefficients $c_P$ which violate the $k$-th no-resonance condition --- form a zero measure set~\cite{Supple}.

\textit{Energy dispersion is typical}---Reference~\cite{huang2020} proved that a state sampled from the infinite temperature ensemble $\mathcal{E}_0$ [Eq.~\eqref{eq:Haarens}] is energy dispersed. This result was generalized~\cite{huang2024randomproduct} to a finite temperature ensemble similar to $\mathcal{E}_{\beta,\delta}$, but containing only stabilizer product states. We improve on this result to apply to the continuous ensemble~$\mathcal{E}_{\beta,\delta}$ of Eq.~\eqref{eq:ensfintemp}, and provide explicit non-asymptotic bounds~\cite{Supple}.

\begin{prop}[Typical product states are energy dispersed]
   \label{thm:energy-dispersion}
    For any extensive local Hamiltonian and all but an exponentially small fraction of states
    $\psi\sim \mathcal{E}_0$ or $\mathcal{E}_{\beta,\delta}$ ($|\beta|<\beta_\mathrm{c}$), $\psi$ is energy dispersed.
\end{prop}
The proofs of Propositions~\ref{thm:generic-ergodic} and~\ref{thm:energy-dispersion} are found in the SM~\cite{Supple}. These two results, combined with Theorem~\ref{thm:theorem2}, imply our main result, Theorem~\ref{thm:theorem1}.
They also have an independent consequence for the recurrence time to the initial state by Corollary 2 of Ref.~\cite{Riddell2023}:
\begin{corollary}[Double-exponential recurrence time for generic local Hamiltonians]
    For $H$ and $\psi$ as in Theorem~\ref{thm:theorem1}, the average recurrence time of $\psi$ is at least $\exp(\exp(\Omega(n)))$.
\end{corollary}
\noindent This matches double-exponential upper bounds on the first recurrence time recently established in Refs.~\cite{Oszmaniec2024,Kotowski2026}.
\section*{Discussion and outlook}
\vspace{-1em}
We show that generic many-body systems evolving under a time-independent local Hamiltonian achieve an exponential late-time circuit complexity, with sustained growth over an exponential window of time. The Brown-Susskind conjecture additionally requires that such growth is linear in time~\cite{Brown2018}. Proving such a statement for any fixed local Hamiltonian would be challenging because of a complexity-theoretic barrier~\cite{aaronson2016,susskind2018}: exponential-time Hamiltonian evolution lies in $\mathsf{PSPACE}$, while a proof that the resulting states require superpolynomial-size quantum circuits would imply $\mathsf{PSPACE}\not\subseteq\mathsf{BQP}$. Existing works proving the growth of circuit complexity in time-independent Hamiltonians therefore rely on complexity-theoretic assumptions and produce statements about specific models~\cite{aaronson2016,bohdanowicz2017,susskind2018}. Our result does not establish linear growth and hence does not imply a complexity-theoretic separation, but universally applies to generic local Hamiltonians and initial product states. In our technical statements, we leverage ideas from quantum many-body dynamics and thermalization, such as energy dispersion and spectral ergodicity. This suggests the interesting possibility that such properties and techniques can inform open questions in theoretical computer science.

Several directions remain open. A first question is an extension to systems with symmetries, fermionic systems, field theories, etc. It is also important to inquire about the tightness of our bounds: is the double-exponential upper bound on the saturation achieved by some model? It could be fruitful to revisit models of slow thermalization~\cite{Balasubramanian2024}. It would also be valuable to understand the complexity of subsystems near half-system size, where one may expect a transition predicted by the AdS/CFT correspondence~\cite{fan2025,haah2025}, and to ask whether this transition is sensitive to temperature in energy-conserving dynamics. 

Beyond circuit complexity, it would be desirable to find other, more tractable diagnostics of post-thermalization dynamics. Corollary~\ref{cor:cor2} shows that late-time volume-law entanglement is irremovable by any subexponential circuit, but this raises a structural question: is such irremovable entanglement quantitatively different from entanglement generated by polynomial-depth circuits? If so, how to quantify such a difference? Ideally, one would like probes of such a strong form of entanglement that, unlike complexity, have closed-form expressions and hence are tractable numerically or analytically. Quantities based on moments of the temporal ensemble~\cite{Roberts_2017,Pilatowsky_Cameo_2023} or the Krylov complexity~\cite{Balasubramanian2022,lucas2026} have been shown to capture hidden late-time structure, opening new possibilities to characterize post-thermalization dynamics.

\begin{acknowledgments}
    We thank Rathindra Nath Das, Jeongwan Haah, Nicole Yunger Halpern, and Thomas Schuster for feedback on earlier versions of this manuscript. W.~L. acknowledges the financial supports by Samsung Science and Technology Foundation under Project Number SSTF-BA2401-03. This work was supported by Global Partnership Program of Leading Universities in Quantum Science and Technology (RS-2025-08542968) through the National Research Foundation of Korea funded by the Korean government (Ministry of Science and ICT), the NSF CAREER Award DMR-2237244, the Heising-Simons Foundation (Grant No.~2024-4851), the Alfred P. Sloan Foundation through a Sloan Research Fellowship, the Center for Ultracold Atoms (an NSF Physics Frontiers Center; PHY-2317134) and by the U.S. DOE Office of Science, under Award Number DE-SC0021013. 
\end{acknowledgments}
\bibliography{Ref}

@article{Balasubramanian2020,
  author = {Balasubramanian, Vijay and DeCross, Matthew and Kar, Arjun and Parrikar, Onkar},
  title = {Quantum complexity of time evolution with chaotic Hamiltonians},
  journal = {Journal of High Energy Physics},
  volume = {2020},
  number = {1},
  pages = {134},
  year = {2020},
  doi = {10.1007/JHEP01(2020)134},
  eprint = {1905.05765},
  archivePrefix = {arXiv},
  primaryClass = {hep-th}
}

@article{Verstraete2006,
  title = {Matrix product states represent ground states faithfully},
  author = {Verstraete, F. and Cirac, J. I.},
  journal = {Phys. Rev. B},
  volume = {73},
  issue = {9},
  pages = {094423},
  numpages = {8},
  year = {2006},
  month = {Mar},
  publisher = {American Physical Society},
  doi = {10.1103/PhysRevB.73.094423},
  url = {https://link.aps.org/doi/10.1103/PhysRevB.73.094423}
}

@ARTICLE{Harrow2023,
    title    = "Approximate Unitary t-Designs by Short Random Quantum Circuits
              Using {Nearest-Neighbor} and {Long-Range} Gates",
    author   = "Harrow, Aram W and Mehraban, Saeed",
    journal  = "Communications in Mathematical Physics",
    volume   =  401,
    number   =  2,
    pages    = "1531--1626",
    month    =  jul,
    year     =  2023,
    doi      = {10.1007/s00220-023-04675-z},
    url      = {https://doi.org/10.1007/s00220-023-04675-z}
}

@article{Nielsen2006,
  author = {Michael A. Nielsen  and Mark R. Dowling  and Mile Gu  and Andrew C. Doherty },
  title = {Quantum Computation as Geometry},
  journal = {Science},
  volume = {311},
  number = {5764},
  pages = {1133-1135},
  year = {2006},
  doi = {10.1126/science.1121541},
  URL = {https://www.science.org/doi/abs/10.1126/science.1121541},
}

@ARTICLE{Haferkamp2022,
  title    = "Linear growth of quantum circuit complexity",
  author   = "Haferkamp, Jonas and Faist, Philippe and Kothakonda, Naga B T and
              Eisert, Jens and Yunger Halpern, Nicole",
  journal  = "Nature Physics",
  volume   =  18,
  number   =  5,
  pages    = "528--532",
  month    =  may,
  year     =  2022,
  url      = "https://doi.org/10.1038/s41567-022-01539-6",
  doi      = "10.1038/s41567-022-01539-6"
}

@InProceedings{Ji2018,
  author="Ji, Zhengfeng
  and Liu, Yi-Kai
  and Song, Fang",
  title="Pseudorandom Quantum States",
  booktitle="Advances in Cryptology -- CRYPTO 2018",
  year="2018",
  publisher="Springer International Publishing",
  address="Cham",
  pages="126--152",
  isbn="978-3-319-96878-0",
  doi="https://doi.org/10.1007/978-3-319-96878-0_5"
}

@article{Mark2023,
  title = {Benchmarking Quantum Simulators Using Ergodic Quantum Dynamics},
  author = {Mark, Daniel K. and Choi, Joonhee and Shaw, Adam L. and Endres, Manuel and Choi, Soonwon},
  journal = {Phys. Rev. Lett.},
  volume = {131},
  issue = {11},
  pages = {110601},
  numpages = {7},
  year = {2023},
  month = {Sep},
  publisher = {American Physical Society},
  doi = {10.1103/PhysRevLett.131.110601},
  url = {https://link.aps.org/doi/10.1103/PhysRevLett.131.110601}
}

@article{Brand_o_2016,
   title={Local Random Quantum Circuits are Approximate Polynomial-Designs},
   volume={346},
   ISSN={1432-0916},
   url={http://dx.doi.org/10.1007/s00220-016-2706-8},
   DOI={10.1007/s00220-016-2706-8},
   number={2},
   journal={Communications in Mathematical Physics},
   publisher={Springer Science and Business Media LLC},
   author={Brandão, Fernando G. S. L. and Harrow, Aram W. and Horodecki, Michał},
   year={2016},
   month=aug, 
   pages={397–434} 
}

@article{Mark2024,
  title = {Maximum Entropy Principle in Deep Thermalization and in {H}ilbert-Space Ergodicity},
  author = {Mark, Daniel K. and Surace, Federica and Elben, Andreas and Shaw, Adam L. and Choi, Joonhee and Refael, Gil and Endres, Manuel and Choi, Soonwon},
  journal = {Phys. Rev. X},
  volume = {14},
  issue = {4},
  pages = {041051},
  numpages = {49},
  year = {2024},
  month = {Nov},
  publisher = {American Physical Society},
  doi = {10.1103/PhysRevX.14.041051},
  url = {https://link.aps.org/doi/10.1103/PhysRevX.14.041051}
}

@article{Haferkamp2022randomquantum,
  doi = {10.22331/q-2022-09-08-795},
  url = {https://doi.org/10.22331/q-2022-09-08-795},
  title = {Random quantum circuits are approximate unitary {$t$}-designs in depth {$O\left(nt^{5+o(1)}\right)$}},
  author = {Haferkamp, Jonas},
  journal = {{Quantum}},
  issn = {2521-327X},
  publisher = {{Verein zur F{\"{o}}rderung des Open Access Publizierens in den Quantenwissenschaften}},
  volume = {6},
  pages = {795},
  month = sep,
  year = {2022}
}

@misc{Supple,
   title = {See Supplementary Material for further details}, 
}

@misc{susskind2018,
      title={Black Holes and Complexity Classes}, 
      author={Leonard Susskind},
      year={2018},
      eprint={1802.02175},
      archivePrefix={arXiv},
      primaryClass={hep-th},
      url={https://arxiv.org/abs/1802.02175}, 
}

@misc{aaronson2016,
      title={The Complexity of Quantum States and Transformations: From Quantum Money to Black Holes}, 
      author={Scott Aaronson},
      year={2016},
      eprint={1607.05256},
      archivePrefix={arXiv},
      primaryClass={quant-ph},
      url={https://arxiv.org/abs/1607.05256}, 
}

@misc{bohdanowicz2017,
      title={Universal Hamiltonians for Exponentially Long Simulation}, 
      author={Thomas C. Bohdanowicz and Fernando G. S. L. Brandão},
      year={2017},
      eprint={1710.02625},
      archivePrefix={arXiv},
      primaryClass={quant-ph},
      url={https://arxiv.org/abs/1710.02625}, 
}

@article{Ho_2022,
   title={Exact Emergent Quantum State Designs from Quantum Chaotic Dynamics},
   volume={128},
   ISSN={1079-7114},
   doi = {10.1103/PhysRevLett.128.060601},
   url={http://dx.doi.org/10.1103/PhysRevLett.128.060601},
   number={6},
   journal={Physical Review Letters},
   publisher={American Physical Society (APS)},
   author={Ho, Wen Wei and Choi, Soonwon},
   year={2022},
   pages={-},
   month=feb 
}

@inproceedings{chen2024,
      title={Incompressibility and spectral gaps of random circuits}, 
      author={Chi-Fang Chen and Jeongwan Haah and Jonas Haferkamp and Yunchao Liu and Tony Metger and Xinyu Tan},
      booktitle={2025 IEEE 66th Annual Symposium on Foundations of Computer Science (FOCS)},
      pages={1304-1312},
      year={2025},
      doi={10.1109/FOCS63196.2025.00069},
}

@ARTICLE{Pilatowsky-Cameo2025,
  title    = "Quantum thermalization must occur in translation-invariant
              systems at high temperature",
  author   = "Pilatowsky-Cameo, Sa{\'u}l and Choi, Soonwon",
  journal  = "Nature Communications",
  volume   =  17,
  number   =  1,
  pages    = "75",
  month    =  jan,
  year     =  2026,
  url      = "https://www.nature.com/articles/s41467-025-66777-7",
  doi      = "10.1038/s41467-025-66777-7"
}

@article{Oszmaniec2024,
  title = {Saturation and Recurrence of Quantum Complexity in Random Local Quantum Dynamics},
  author = {Oszmaniec, Micha\l{} and Kotowski, Marcin and Horodecki, Micha\l{} and Hunter-Jones, Nicholas},
  journal = {Phys. Rev. X},
  volume = {14},
  issue = {4},
  pages = {041068},
  numpages = {43},
  year = {2024},
  month = {Dec},
  publisher = {American Physical Society},
  doi = {10.1103/PhysRevX.14.041068},
  url = {https://link.aps.org/doi/10.1103/PhysRevX.14.041068}
}

@article{Brandao2021,
  title = {Models of Quantum Complexity Growth},
  author = {Brand\~ao, Fernando G.S.L. and Chemissany, Wissam and Hunter-Jones, Nicholas and Kueng, Richard and Preskill, John},
  journal = {PRX Quantum},
  volume = {2},
  issue = {3},
  pages = {030316},
  numpages = {40},
  year = {2021},
  month = {Jul},
  publisher = {American Physical Society},
  doi = {10.1103/PRXQuantum.2.030316},
  url = {https://link.aps.org/doi/10.1103/PRXQuantum.2.030316}
}

@article{Gu2021,
  doi = {10.22331/q-2021-11-15-577},
  url = {https://doi.org/10.22331/q-2021-11-15-577},
  title = {Fast-forwarding quantum evolution},
  author = {Gu, Shouzhen and Somma, Rolando D. and {\c{S}}ahino{\u{g}}lu, Burak},
  journal = {{Quantum}},
  issn = {2521-327X},
  publisher = {{Verein zur F{\"{o}}rderung des Open Access Publizierens in den Quantenwissenschaften}},
  volume = {5},
  pages = {577},
  month = nov,
  year = {2021}
}

@ARTICLE{Atia2017,
  title    = "Fast-forwarding of {H}amiltonians and exponentially precise
              measurements",
  author   = "Atia, Yosi and Aharonov, Dorit",
  journal  = "Nature Communications",
  volume   =  8,
  number   =  1,
  pages    = "1572",
  month    =  nov,
  year     =  2017,
  url      = "https://doi.org/10.1038/s41467-017-01637-7"
}

@misc{haferkamp2023,
      title={On the moments of random quantum circuits and robust quantum complexity}, 
      author={Jonas Haferkamp},
      year={2023},
      eprint={2303.16944},
      archivePrefix={arXiv},
      primaryClass={quant-ph},
      url={https://arxiv.org/abs/2303.16944}, 
}

@article{haah2025,
      title={Growth and collapse of subsystem complexity under random unitary circuits}, 
      author={Jeongwan Haah and Douglas Stanford},
      journal={SciPost Phys.},
      volume={21},
      pages={016},
      year={2026},
      doi={10.21468/SciPostPhys.21.1.016},
      url={https://scipost.org/10.21468/SciPostPhys.21.1.016}, 
}

@misc{fan2025,
      title={Sharp Transitions for Subsystem Complexity}, 
      author={Yale Fan and Nicholas Hunter-Jones and Andreas Karch and Shivan Mittal},
      year={2025},
      eprint={2510.18832},
      archivePrefix={arXiv},
      primaryClass={hep-th},
      url={https://arxiv.org/abs/2510.18832}, 
}

@article{Susskind2016,
    author = {Susskind, Leonard},
    title = {Computational complexity and black hole horizons},
    journal = {Fortschritte der Physik},
    volume = {64},
    number = {1},
    pages = {24-43},
    doi = {https://doi.org/10.1002/prop.201500092},
    url = {https://onlinelibrary.wiley.com/doi/abs/10.1002/prop.201500092},
    year = {2016}
}

@misc{mao2025,
      title={Random unitaries that conserve energy}, 
      author={Liang Mao and Laura Cui and Thomas Schuster and Hsin-Yuan Huang},
      year={2025},
      eprint={2510.08448},
      archivePrefix={arXiv},
      primaryClass={quant-ph},
      url={https://arxiv.org/abs/2510.08448}, 
}

@INPROCEEDINGS{bakshi2025,
  author={Bakshi, Ainesh and Liu, Allen and Moitra, Ankur and Tang, Ewin},
  booktitle={2024 IEEE 65th Annual Symposium on Foundations of Computer Science (FOCS)}, 
  title={High-Temperature {G}ibbs States are Unentangled and Efficiently Preparable}, 
  year={2024},
  volume={},
  number={},
  pages={1027-1036},
  doi={10.1109/FOCS61266.2024.00068}
}

@book{Bhatia_1997, 
    place={New York, NY}, 
    series={Graduate Texts in Mathematics}, 
    title={Matrix analysis}, 
    volume={169}, 
    publisher={Springer}, 
    author={Bhatia, Rajendra}, 
    year={1997}, 
    collection={Graduate Texts in Mathematics},
    doi={10.1007/978-1-4612-0653-8},
    url={https://doi.org/10.1007/978-1-4612-0653-8},
}

@misc{huang2024randomproduct,
      title={Random product states at high temperature equilibrate exponentially well}, 
      author={Yichen Huang},
      year={2024},
      eprint={2409.08436},
      archivePrefix={arXiv},
      primaryClass={cond-mat.stat-mech},
      url={https://arxiv.org/abs/2409.08436}, 
}

@misc{Huang2021,
      title={Extensive entropy from unitary evolution}, 
      author={Yichen Huang},
      year={2021},
      eprint={2104.02053},
      archivePrefix={arXiv},
      primaryClass={quant-ph}
}

@misc{chen2025,
      title={Quantum {G}ibbs states are locally Markovian}, 
      author={Chi-Fang Chen and Cambyse Rouzé},
      year={2025},
      eprint={2504.02208},
      archivePrefix={arXiv},
      primaryClass={quant-ph},
      url={https://arxiv.org/abs/2504.02208}, 
}

@article{Kliesch2014,
  title = {Locality of Temperature},
  author = {Kliesch, M. and Gogolin, C. and Kastoryano, M. J. and Riera, A. and Eisert, J.},
  journal = {Phys. Rev. X},
  volume = {4},
  issue = {3},
  pages = {031019},
  numpages = {19},
  year = {2014},
  month = {Jul},
  publisher = {American Physical Society},
  doi = {10.1103/PhysRevX.4.031019},
  url = {https://link.aps.org/doi/10.1103/PhysRevX.4.031019}
}

@ARTICLE{Capel2025,
  title    = "From Decay of Correlations to Locality and Stability of the Gibbs
              State",
  author   = "Capel, {\'A}ngela and Moscolari, Massimo and Teufel, Stefan and
              Wessel, Tom",
  journal  = "Communications in Mathematical Physics",
  volume   =  406,
  number   =  2,
  pages    = "43",
  month    =  jan,
  year     =  2025,
  url      = "https://doi.org/10.1007/s00220-024-05198-x",
  doi      = "10.1007/s00220-024-05198-x"
}

@ARTICLE{Brandao2019,
  title    = "Finite Correlation Length Implies Efficient Preparation of
              Quantum Thermal States",
  author   = "Brand{\~a}o, Fernando G S L and Kastoryano, Michael J",
  journal  = "Communications in Mathematical Physics",
  volume   =  365,
  number   =  1,
  pages    = "1--16",
  month    =  jan,
  year     =  2019,
  url      = "https://doi.org/10.1007/s00220-018-3150-8",
  doi      = "10.1007/s00220-018-3150-8"
}

@article{Alhambra2020,
  title = {Time Evolution of Correlation Functions in Quantum Many-Body Systems},
  author = {Alhambra, \'Alvaro M. and Riddell, Jonathon and Garc\'{\i}a-Pintos, Luis Pedro},
  journal = {Phys. Rev. Lett.},
  volume = {124},
  issue = {11},
  pages = {110605},
  numpages = {7},
  year = {2020},
  month = {Mar},
  publisher = {American Physical Society},
  doi = {10.1103/PhysRevLett.124.110605},
  url = {https://link.aps.org/doi/10.1103/PhysRevLett.124.110605}
}

@article{Brandao2019qec,
  title = {Quantum Error Correcting Codes in Eigenstates of Translation-Invariant Spin Chains},
  author = {Brand\~ao, Fernando G. S. L. and Crosson, Elizabeth and \ifmmode \mbox{\c{S}}\else \c{S}\fi{}ahino\ifmmode \breve{g}\else \u{g}\fi{}lu, M. Burak and Bowen, John},
  journal = {Phys. Rev. Lett.},
  volume = {123},
  issue = {11},
  pages = {110502},
  numpages = {6},
  year = {2019},
  month = {Sep},
  publisher = {American Physical Society},
  doi = {10.1103/PhysRevLett.123.110502},
  url = {https://link.aps.org/doi/10.1103/PhysRevLett.123.110502}
}

@article{Brown2018,
  title = {Second law of quantum complexity},
  author = {Brown, Adam R. and Susskind, Leonard},
  journal = {Phys. Rev. D},
  volume = {97},
  issue = {8},
  pages = {086015},
  numpages = {29},
  year = {2018},
  month = {Apr},
  publisher = {American Physical Society},
  doi = {10.1103/PhysRevD.97.086015},
  url = {https://link.aps.org/doi/10.1103/PhysRevD.97.086015}
}

@misc{hunterjones2019,
      title={Unitary designs from statistical mechanics in random quantum circuits}, 
      author={Nicholas Hunter-Jones},
      year={2019},
      eprint={1905.12053},
      archivePrefix={arXiv},
      primaryClass={quant-ph},
      url={https://arxiv.org/abs/1905.12053}, 
}

@ARTICLE{Jian2023,
  title    = "Linear growth of circuit complexity from Brownian dynamics",
  author   = "Jian, Shao-Kai and Bentsen, Gregory and Swingle, Brian",
  journal  = "Journal of High Energy Physics",
  volume   =  2023,
  number   =  8,
  pages    = "190",
  month    =  aug,
  year     =  2023,
  url      = "https://doi.org/10.1007/JHEP08(2023)190",
  doi      = "10.1007/JHEP08(2023)190"
}

@article{Stanford2014,
  title = {Complexity and shock wave geometries},
  author = {Stanford, Douglas and Susskind, Leonard},
  journal = {Phys. Rev. D},
  volume = {90},
  issue = {12},
  pages = {126007},
  numpages = {11},
  year = {2014},
  month = {Dec},
  publisher = {American Physical Society},
  doi = {10.1103/PhysRevD.90.126007},
  url = {https://link.aps.org/doi/10.1103/PhysRevD.90.126007}
}

@article{Brown2016,
  title = {Holographic Complexity Equals Bulk Action?},
  author = {Brown, Adam R. and Roberts, Daniel A. and Susskind, Leonard and Swingle, Brian and Zhao, Ying},
  journal = {Phys. Rev. Lett.},
  volume = {116},
  issue = {19},
  pages = {191301},
  numpages = {5},
  year = {2016},
  month = {May},
  publisher = {American Physical Society},
  doi = {10.1103/PhysRevLett.116.191301},
  url = {https://link.aps.org/doi/10.1103/PhysRevLett.116.191301}
}

@article{Haah2023,
    author = {Haah, Jeongwan and Hastings, Matthew B. and Kothari, Robin and Low, Guang Hao},
    title = {Quantum Algorithm for Simulating Real Time Evolution of Lattice {H}amiltonians},
    journal = {SIAM Journal on Computing},
    volume = {52},
    number = {6},
    pages = {FOCS18-250-FOCS18-284},
    year = {2023},
    doi = {10.1137/18M1231511},
    URL = {https://doi.org/10.1137/18M1231511}
}

@article{Keating2015,
  title = {Spectra and Eigenstates of Spin Chain {H}amiltonians},
  volume = {338},
  ISSN = {1432-0916},
  url = {http://dx.doi.org/10.1007/s00220-015-2366-0},
  DOI = {10.1007/s00220-015-2366-0},
  number = {1},
  journal = {Communications in Mathematical Physics},
  publisher = {Springer Science and Business Media LLC},
  author = {Keating,  J. P. and Linden,  N. and Wells,  H. J.},
  year = {2015},
  month = may,
  pages = {81–102}
}

@ARTICLE{Araki1969,
  title    = "Gibbs states of a one dimensional quantum lattice",
  author   = "Araki, Huzihiro",
  journal  = "Communications in Mathematical Physics",
  volume   =  14,
  number   =  2,
  pages    = "120--157",
  month    =  jun,
  year     =  1969,
  url      = "https://doi.org/10.1007/BF01645134",
  doi      = "10.1007/BF01645134"
}

@article{Bluhm2022,
  doi = {10.22331/q-2022-02-10-650},
  url = {https://doi.org/10.22331/q-2022-02-10-650},
  title = {Exponential decay of mutual information for {G}ibbs states of local {H}amiltonians},
  author = {Bluhm, Andreas and Capel, {\'{A}}ngela and P{\'{e}}rez-Hern{\'{a}}ndez, Antonio},
  journal = {{Quantum}},
  issn = {2521-327X},
  publisher = {{Verein zur F{\"{o}}rderung des Open Access Publizierens in den Quantenwissenschaften}},
  volume = {6},
  pages = {650},
  month = feb,
  year = {2022}
}

@BOOK{Kato1995,
  title     = "Perturbation theory for linear operators",
  author    = "Kato, Tosio",
  publisher = "Springer",
  series    = "Classics in Mathematics",
  edition   =  2,
  month     =  feb,
  year      =  1995,
  address   = "Berlin, Germany",
  url       = "https://doi.org/10.1007/978-3-642-66282-9"
}

@article{Muller2015,
  title = {Thermalization and Canonical Typicality in Translation-Invariant Quantum Lattice Systems},
  volume = {340},
  ISSN = {1432-0916},
  url = {http://dx.doi.org/10.1007/s00220-015-2473-y},
  DOI = {10.1007/s00220-015-2473-y},
  number = {2},
  journal = {Comm. Math. Phys.},
  publisher = {Springer Science and Business Media LLC},
  author = {M\"{u}ller,  Markus P. and Adlam,  Emily and Masanes,  Lluís and Wiebe,  Nathan},
  year = {2015},
  month = sep,
  pages = {499–561}
}

@article{Linden2009,
  title = {Quantum mechanical evolution towards thermal equilibrium},
  author = {Linden, Noah and Popescu, Sandu and Short, Anthony J. and Winter, Andreas},
  journal = {Phys. Rev. E},
  volume = {79},
  issue = {6},
  pages = {061103},
  numpages = {12},
  year = {2009},
  month = {Jun},
  publisher = {American Physical Society},
  doi = {10.1103/PhysRevE.79.061103},
  url = {https://link.aps.org/doi/10.1103/PhysRevE.79.061103}
}

@article{Reimann2008,
    title = {Foundation of {Statistical} {Mechanics} under {Experimentally} {Realistic} {Conditions}},
    volume = {101},
    url = {https://link.aps.org/doi/10.1103/PhysRevLett.101.190403},
    doi = {10.1103/PhysRevLett.101.190403},
    number = {19},
    urldate = {2026-04-24},
    journal = {Physical Review Letters},
    author = {Reimann, Peter},
    month = nov,
    year = {2008},
    pages = {190403},
}

@misc{huang2020,
      title={Instability of localization in translation-invariant systems}, 
      author={Yichen Huang and Aram W. Harrow},
      year={2020},
      eprint={1907.13392},
      archivePrefix={arXiv},
      primaryClass={cond-mat.dis-nn},
      url={https://arxiv.org/abs/1907.13392}, 
}

@article{Biroli2010,
  title = {Effect of Rare Fluctuations on the Thermalization of Isolated Quantum Systems},
  author = {Biroli, Giulio and Kollath, Corinna and L\"auchli, Andreas M.},
  journal = {Phys. Rev. Lett.},
  volume = {105},
  issue = {25},
  pages = {250401},
  numpages = {4},
  year = {2010},
  month = {Dec},
  publisher = {American Physical Society},
  doi = {10.1103/PhysRevLett.105.250401},
  url = {https://link.aps.org/doi/10.1103/PhysRevLett.105.250401}
}

@article{Cotler2022,
  title = {Fluctuations of subsystem entropies at late times},
  author = {Cotler, Jordan and Hunter-Jones, Nicholas and Ranard, Daniel},
  journal = {Phys. Rev. A},
  volume = {105},
  issue = {2},
  pages = {022416},
  numpages = {20},
  year = {2022},
  month = {Feb},
  publisher = {American Physical Society},
  doi = {10.1103/PhysRevA.105.022416},
  url = {https://link.aps.org/doi/10.1103/PhysRevA.105.022416}
}

@article{Brown2016complexity,
  title = {Complexity, action, and black holes},
  author = {Brown, Adam R. and Roberts, Daniel A. and Susskind, Leonard and Swingle, Brian and Zhao, Ying},
  journal = {Phys. Rev. D},
  volume = {93},
  issue = {8},
  pages = {086006},
  numpages = {32},
  year = {2016},
  month = {Apr},
  publisher = {American Physical Society},
  doi = {10.1103/PhysRevD.93.086006},
  url = {https://link.aps.org/doi/10.1103/PhysRevD.93.086006}
}

@book{Jukna2012,
  author    = {Jukna, Stasys},
  title     = {Boolean Function Complexity: Advances and Frontiers},
  publisher = {Springer},
  year      = {2012},
  doi       = {10.1007/978-3-642-24508-4}
}

@article{MARGOLUS1998,
    title = {The maximum speed of dynamical evolution},
    journal = {Physica D: Nonlinear Phenomena},
    volume = {120},
    number = {1},
    pages = {188-195},
    year = {1998},
    note = {Proceedings of the Fourth Workshop on Physics and Consumption},
    issn = {0167-2789},
    doi = {https://doi.org/10.1016/S0167-2789(98)00054-2},
    url = {https://www.sciencedirect.com/science/article/pii/S0167278998000542},
    author = {Norman Margolus and Lev B. Levitin},
}

@article{Riddell2023,
   title={Concentration of quantum equilibration and an estimate of the recurrence time},
   volume={15},
   ISSN={2542-4653},
   url={http://dx.doi.org/10.21468/SciPostPhys.15.4.165},
   DOI={10.21468/scipostphys.15.4.165},
   number={4},
   journal={SciPost Physics},
   publisher={Stichting SciPost},
   author={Riddell, Jonathon and Pagliaroli, Nathan J. and Alhambra, \'Alvaro M.},
   year={2023},
   month=Oct }

@article{Kaneko2020,
  title = {Characterizing complexity of many-body quantum dynamics by higher-order eigenstate thermalization},
  author = {Kaneko, Kazuya and Iyoda, Eiki and Sagawa, Takahiro},
  journal = {Phys. Rev. A},
  volume = {101},
  issue = {4},
  pages = {042126},
  numpages = {23},
  year = {2020},
  month = {Apr},
  publisher = {American Physical Society},
  doi = {10.1103/PhysRevA.101.042126},
  url = {https://link.aps.org/doi/10.1103/PhysRevA.101.042126}
}

@misc{Mori2016,
      title={Weak eigenstate thermalization with large deviation bound}, 
      author={Takashi Mori},
      year={2016},
      eprint={1609.09776},
      archivePrefix={arXiv},
      primaryClass={cond-mat.stat-mech},
      url={https://arxiv.org/abs/1609.09776}, 
}

@misc{cui2025hamiltonian,
      title={Random unitaries from {H}amiltonian dynamics}, 
      author={Laura Cui and Thomas Schuster and Liang Mao and Hsin-Yuan Huang and Fernando Brandao},
      year={2025},
      eprint={2510.08434},
      archivePrefix={arXiv},
      primaryClass={quant-ph},
      url={https://arxiv.org/abs/2510.08434}, 
}

@misc{raza2024,
      title={Online learning of quantum processes}, 
      author={Asad Raza and Matthias C. Caro and Jens Eisert and Sumeet Khatri},
      year={2024},
      eprint={2406.04250},
      archivePrefix={arXiv},
      primaryClass={quant-ph},
      url={https://arxiv.org/abs/2406.04250}, 
}

@article{Balasubramanian2024,
  title = {Glassy Word Problems: Ultraslow Relaxation, {H}ilbert Space Jamming, and Computational Complexity},
  author = {Balasubramanian, Shankar and Gopalakrishnan, Sarang and Khudorozhkov, Alexey and Lake, Ethan},
  journal = {Phys. Rev. X},
  volume = {14},
  issue = {2},
  pages = {021034},
  numpages = {47},
  year = {2024},
  month = {May},
  publisher = {American Physical Society},
  doi = {10.1103/PhysRevX.14.021034},
  url = {https://link.aps.org/doi/10.1103/PhysRevX.14.021034}
}

@ARTICLE{Shende2006,
  author={Shende, V.V. and Bullock, S.S. and Markov, I.L.},
  journal={IEEE Transactions on Computer-Aided Design of Integrated Circuits and Systems}, 
  title={Synthesis of quantum-logic circuits}, 
  year={2006},
  volume={25},
  number={6},
  pages={1000-1010},
  doi = {10.1109/TCAD.2005.855930}
}

@article{Roberts_2017,
   title={Chaos and complexity by design},
   volume={2017},
   ISSN={1029-8479},
   doi={10.1007/JHEP04(2017)121},
   url={http://dx.doi.org/10.1007/JHEP04(2017)121},
   number={4},
   journal={Journal of High Energy Physics},
   publisher={Springer Science and Business Media LLC},
   author={Roberts, Daniel A. and Yoshida, Beni},
   year={2017},
   pages={-},
   month=Apr 
}

@article{Pilatowsky_Cameo_2023,
   title={Complete {H}ilbert-Space Ergodicity in Quantum Dynamics of Generalized {F}ibonacci Drives},
   volume={131},
   ISSN={1079-7114},
   doi = {10.1103/PhysRevLett.131.250401},
   url={http://dx.doi.org/10.1103/physrevlett.131.250401},
   issue={25},
   journal={Physical Review Letters},
   publisher={American Physical Society},
   author={Pilatowsky-Cameo, Saúl and Dag, Ceren B. and Ho, Wen Wei and Choi, Soonwon},
   year={2023},
   pages = {250401},
   numpages = {7},
   month=Dec 
}

@article{Childs2019,
  title = {Nearly Optimal Lattice Simulation by Product Formulas},
  author = {Childs, Andrew M. and Su, Yuan},
  journal = {Phys. Rev. Lett.},
  volume = {123},
  issue = {5},
  pages = {050503},
  numpages = {6},
  year = {2019},
  month = {Aug},
  publisher = {American Physical Society},
  url = {https://link.aps.org/doi/10.1103/PhysRevLett.123.050503},
  doi = {10.1103/PhysRevLett.123.050503}
}

@article{Schon2005,
  title = {Sequential Generation of Entangled Multiqubit States},
  author = {Sch\"on, C. and Solano, E. and Verstraete, F. and Cirac, J. I. and Wolf, M. M.},
  journal = {Phys. Rev. Lett.},
  volume = {95},
  issue = {11},
  pages = {110503},
  numpages = {4},
  year = {2005},
  month = {Sep},
  publisher = {American Physical Society},
  url = {https://link.aps.org/doi/10.1103/PhysRevLett.95.110503},
  doi = {10.1103/PhysRevLett.95.110503}
}

@article{perezgarcia2007,
    author = {Perez-Garcia, D. and Verstraete, F. and Wolf, M. M. and Cirac, J. I.},
    title = {Matrix product state representations},
    year = {2007},
    issue_date = {July 2007},
    publisher = {Rinton Press, Incorporated},
    address = {Paramus, NJ},
    volume = {7},
    number = {5},
    issn = {1533-7146},
    journal = {Quantum Info. Comput.},
    month = jul,
    pages = {401–430},
    numpages = {30},
    url = {https://dl.acm.org/doi/10.5555/2011832.2011833}
}

@article{Iten2016,
  title = {Quantum circuits for isometries},
  author = {Iten, Raban and Colbeck, Roger and Kukuljan, Ivan and Home, Jonathan and Christandl, Matthias},
  journal = {Phys. Rev. A},
  volume = {93},
  issue = {3},
  pages = {032318},
  numpages = {19},
  year = {2016},
  month = {Mar},
  publisher = {American Physical Society},
  url = {https://link.aps.org/doi/10.1103/PhysRevA.93.032318},
  doi = {10.1103/PhysRevA.93.032318}
}

@article{Anandan1990,
  title = {Geometry of quantum evolution},
  author = {Anandan, J. and Aharonov, Y.},
  journal = {Phys. Rev. Lett.},
  volume = {65},
  issue = {14},
  pages = {1697--1700},
  numpages = {0},
  year = {1990},
  month = {Oct},
  publisher = {American Physical Society},
  doi = {10.1103/PhysRevLett.65.1697},
  url = {https://link.aps.org/doi/10.1103/PhysRevLett.65.1697}
}

@article{Taddei2013,
  title = {Quantum Speed Limit for Physical Processes},
  author = {Taddei, M. M. and Escher, B. M. and Davidovich, L. and de Matos Filho, R. L.},
  journal = {Phys. Rev. Lett.},
  volume = {110},
  issue = {5},
  pages = {050402},
  numpages = {5},
  year = {2013},
  month = {Jan},
  publisher = {American Physical Society},
  doi = {10.1103/PhysRevLett.110.050402},
  url = {https://link.aps.org/doi/10.1103/PhysRevLett.110.050402}
}

@misc{mcginley2025,
      title={The Scrooge ensemble in many-body quantum systems}, 
      author={Max McGinley and Thomas Schuster},
      year={2025},
      eprint={2511.17172},
      archivePrefix={arXiv},
      primaryClass={quant-ph},
      url={https://arxiv.org/abs/2511.17172}, 
}

@article{Kotowski2026,
  title = {Tight Bounds on Recurrence Time in Closed Quantum Systems},
  author = {Kotowski, Marcin and Oszmaniec, Micha\l{}},
  journal = {Phys. Rev. Lett.},
  volume = {137},
  issue = {6},
  pages = {060402},
  numpages = {7},
  year = {2026},
  month = {Aug},
  publisher = {American Physical Society},
  doi = {10.1103/crwj-qfw5},
  url = {https://link.aps.org/doi/10.1103/crwj-qfw5}
}

@article{Balasubramanian2022,
  title = {Quantum chaos and the complexity of spread of states},
  author = {Balasubramanian, Vijay and Caputa, Pawel and Magan, Javier M. and Wu, Qingyue},
  journal = {Phys. Rev. D},
  volume = {106},
  issue = {4},
  pages = {046007},
  numpages = {28},
  year = {2022},
  month = {Aug},
  publisher = {American Physical Society},
  doi = {10.1103/PhysRevD.106.046007},
  url = {https://link.aps.org/doi/10.1103/PhysRevD.106.046007}
}

@misc{lucas2026,
      title={Non-perturbative saturation of Krylov complexity, and its implications in quantum gravity}, 
      author={Andrew Lucas and Amit Vikram},
      year={2026},
      eprint={2607.14220},
      archivePrefix={arXiv},
      primaryClass={hep-th},
      url={https://arxiv.org/abs/2607.14220}, 
}

\clearpage

\noindent{\large{\bf Methods}}\\
\noindent{\bf Provable small subsystem complexity}\\

An initial state $\psi$ thermalizes on a region $A$ if two conditions hold: \textit{equilibration} and \textit{temporal-thermal equivalence}. Equilibration means that with high probability, the reduced state $\psi_A(t)$ at late times approximates  $\bar \rho_A$ with a small bounded error, where $\bar \rho$ is the time average of $\psi(t)$. Equilibration provably follows from the non-degenerate gap condition and energy dispersion~\cite{Reimann2008,Linden2009}. Second, we say $\psi$ satisfies the $\varepsilon$-temporal-thermal equivalence on $A$ if $\bar \rho_A$ approximates the Gibbs state in $A$ to $\varepsilon$ trace distance. 

\begin{theorem}[Small subsystem complexity from energy dispersion and temporal-thermal equivalence]
    \label{thm:low-c-from-tte-ed}
    Let $H$ be a local Hamiltonian with non-degenerate spectral gaps, and $\psi$ be energy dispersed  with constant $\gamma$. For any positive constants $\alpha<\varepsilon<1$, if $\psi$ satisfies the $(\varepsilon/3)$-temporal-thermal equivalence at sufficiently high temperature on a contiguous region $A$ with $|A|/n\leq \frac{1}{2}(\gamma-\alpha)$, then, except for an $O(2^{-\alpha n})$ fraction of time $t\geq 0$,
        \begin{equation*}    \mathcal{C}_\varepsilon(\psi_A(t))\leq  \exp(\mathrm{polylog}(|A|)).
    \end{equation*}
\end{theorem}

Theorem~\ref{thm:low-c-from-tte-ed} assumes the temporal-thermal equivalence, which is generically expected, and can be proven in various contexts. In particular, it holds under a certain weak eigenstate thermalization hypothesis, which  can be proven in translation-invariant systems~\cite{Mori2016} for initial states drawn from an ensemble of random product states whose average is the Gibbs state~\cite{huang2020,Pilatowsky-Cameo2025}, such as $\mathcal{E}_0$ or the ensemble of stabilizer product states produced by the algorithm of Ref.~\cite{bakshi2025}. Specifically, we show that the equivalence is valid for $|A|\leq O(\log n)$ in translation-invariant systems for $\mathcal{E}_0$. See the SM~\cite{Supple} for details.

\vspace{1em}
\noindent{\bf Double-exponential temporal delocalization}\\

Here, we provide the proof of the following lemma, used in the proof of Theorem~\ref{thm:theorem1}.

\begin{lemma}[Double-exponential temporal delocalization]
    \label{thm:lemma1}
    For any Hamiltonian having spectral ergodicity, any initial state $\psi$ which is energy dispersed with $\gamma>0$, and any state $\sigma$ supported in $A$, except for an $\exp(-\exp(\alpha n))$ fraction of time with $\alpha\leq\gamma/8$,
    \begin{equation}
        D(\psi_A(t),\sigma)>\varepsilon
    \end{equation}
    for any subsystem $A$ with $|A|>(1-\gamma/8)n$ and $\varepsilon<1-\exp(-\Theta(n))$, where $D(\cdot,\cdot)$ denotes the trace distance. 
\end{lemma}

\begin{proof} For simplicity, let us focus here on the case where $A$ is the full system $|A|=n$ (the general proof is found in the SM~\cite{Supple}).
By Markov's inequality, the probability over time $t\in [0,\infty)$ is bounded, 
\begin{equation}\label{eq:nowhere-small-distance2}        
    \operatornamewithlimits{Pr}_t\bigl[D(\psi(t),\sigma)\leq \varepsilon\bigr]\leq M_k\,{(1-\varepsilon)^{-2k}},
\end{equation}
in terms of the $k$-th temporal moment  $M_k=\lim_{T\to \infty}\frac{1}{T}\int_0^T \dd t \langle\psi(t)|\sigma|\psi(t)\rangle ^k$, so we only need to bound these moments, for which we use our two assumptions. On the one hand, spectral ergodicity implies that upon time averaging, all off-diagonal non-resonating energy terms vanish~\cite{Mark2024},  allowing us to place a bound~\cite{Riddell2023},
\begin{equation}
    M_k\leq k!(\max_\mu p_\mu)^k 
\end{equation}   
in terms of the energy distribution $p_\mu=\abs{\braket{\mu}{\psi}}^2$. On the other hand, the energy dispersion guarantees that the maximum population is exponentially small: $\max_\mu p_\mu \leq 2^{-\gamma n/2}$. Combining both inequalities, the right-hand side of Eq.~\eqref{eq:nowhere-small-distance2} is upper bounded by
\begin{equation*}
    \frac{k! 2^{-\gamma nk/2}}{(1-\varepsilon)^{2k}}\leq \left(\frac{k 2^{-\gamma n/2} }{(1-\varepsilon)^2}\right)^k \leq \exp(-\exp(\alpha n)),
\end{equation*}
where the second inequality is due to Stirling's approximation, and the last one holds by choosing $k=e^{\alpha n}$ for any $\alpha< \gamma/4$ and sufficiently large $n\geq \Omega(\log(1/(1-\varepsilon)))$.\end{proof}

\vspace{1em}
\noindent{\bf Sufficient time for exponential complexity}\\

Here, we place a finite bound on the time $t_\mathrm{typ}$ required to reach the typical circuit complexity stated in Theorem~\ref{thm:theorem1}. To this end, we first make the notion of spectral ergodicity quantitative by introducing the \textit{$k$-th spectral gap $\Delta^{(k)}$}: the minimum difference $\abs{\sum_{i=1}^k E_{\mu_i} -\sum_{i=1}^k E_{\nu_i}}$ over all $k$-tuples of indices $(\mu_1,\dots,\mu_k)$ and $(\nu_1,\dots,\nu_k)$ which are not permutations of each other. The $k$-th no-resonance condition is simply the statement that $\Delta^{(k)}>0$. We prove that $t_\mathrm{typ}$ is at most proportional to $1/\Delta^{(k)}$, for an exponentially large $k$.
\begin{theorem}[Exponential quantum circuit complexity at finite time]
    For $H$ and $\psi$ as in Theorem~\ref{thm:theorem1}, and for a certain $T=\exp(\exp(O(n)))/\Delta^{(k_c)}$ with $k_c=\exp(\Theta(n))$, except for a double-exponentially small time fraction of $[0,T]$,
        \begin{equation*}
\mathcal{C}_\varepsilon(\psi(t))\geq  \exp \Omega(n).
\end{equation*}
\end{theorem}
\noindent This theorem implies that typical states at times of order $\exp(\exp(O(n)))/\Delta^{(k_c)}$ have exponentially large circuit complexity. This time scale gives an upper bound on $t_\mathrm{typ}$. 

While the $k$-th spectral gap is conceptually simple, to the best of our knowledge, no rigorous lower bound on $\Delta^{(k)}$ is known for generic local extensive Hamiltonians, even in the simplest cases of $k=1$ or $k=2$. We find that $\Delta^{(k)}$ decays approximately as $\exp(-\Theta(nk))$ based on a heuristic argument (see the SM~\cite{Supple}). Under this empirical scaling, choosing $k=\exp(\Theta(n))$ gives $1/\Delta^{(k)}=\exp(\exp(\Theta(n)))$, and hence $t_\mathrm{typ}\lesssim \exp(\exp(\Omega(n)))$. This double-exponential scaling is essentially the strongest one that can be obtained from any argument based on counting~\cite{Supple}.

\end{document}


\def\bibsection{\section*{\refname}} 
\newcommand{\E}[0]{\mathop{{}\mathbb{E}}}
\newcommand{\Prob}{\mathop{\mathbb{P}}\limits}

\makeatletter
\newcommand\footnoteref[1]{\protected@xdef\@thefnmark{\ref{#1}}\@footnotemark}
\makeatother

\title{Supplementary Material: Sustained growth of quantum circuit complexity in many-body Hamiltonian dynamics}
\author{Wonjun Lee}
\email{wonjun98@mit.edu}
\affiliation{College of Natural Sciences, Korea Advanced Institute of Science and Technology, Daejeon, 34141, Republic of Korea}
\affiliation{Research Laboratory of Electronics, Massachusetts Institute of Technology, Cambridge, Massachusetts, 02139, United States}
\affiliation{MIT Center for Theoretical Physics -- a Leinweber Institute, Massachusetts Institute of Technology, Cambridge, Massachusetts, 02139, United States
}
\author{Sa\'ul Pilatowsky-Cameo}
\email{saulpila@mit.edu}
\affiliation{MIT Center for Theoretical Physics -- a Leinweber Institute, Massachusetts Institute of Technology, Cambridge, Massachusetts, 02139, United States
}
\author{Soonwon Choi}
\email{soonwon@mit.edu}
\affiliation{MIT Center for Theoretical Physics -- a Leinweber Institute, Massachusetts Institute of Technology, Cambridge, Massachusetts, 02139, United States
}

\newcommand{\maintextref}[1]{{#1}} 

\newcommand{\saul}[1]{\textcolor{blue}{[[Sa\'ul: #1]]}}
\newcommand{\WJL}[1]{{\color{purple} #1}}
\maketitle

\vspace{-1.2em}

In this Supplementary Material, we provide detailed information, complete proofs, and rigorous statements of the results stated in the main text. In Section~\ref{sec:01}, we introduce the notation, the physical setting, our notion of quantum circuit complexity, the class of Hamiltonians, the ensembles of product states, and the finite correlation length for Gibbs states. Sections~\ref{sec:02} and~\ref{sec:03} establish the two main technical properties used throughout the work: spectral ergodicity for generic local Hamiltonians and energy dispersion for typical product states. In Section~\ref{sec:04}, we prove the exponential lower bounds on late-time circuit complexity, Theorems 1 and 3 of the main text. Section~\ref{sec:05} contains the proofs of the four corollaries, namely sustained complexity growth, irremovable volume-law entanglement, generic no fast-forwarding, and the recurrence-time bound. In Section~\ref{sec:06}, we prove the low-complexity result for small subsystems and discuss examples in which small subsystem complexity follows from thermalization. Finally, in Section~\ref{sec:07}, we prove a finite-time result for reaching exponential complexity and analyze the scaling of the relevant spectral gaps.
\vspace{-1.2em}
\tableofcontents
\clearpage
\section{Preliminaries}
\label{sec:01}
\subsection{Setting and basic notations}\label{sec:setting}
\newcommand{\supp}{\mathrm{supp}}

We consider a hypercubic lattice $\Lambda$ in $D$ dimensions, with $n$ vertices. For vertices $i,j\in\Lambda$, $d(i,j)$ denotes the minimum number of edges in a path connecting them. For regions $A,B\subseteq\Lambda$, $d(A,B)$ is defined by $\min_{i\in A,j\in B}d(i,j)$. We allow for periodic boundary conditions, by identifying the edges of the lattice. For a region $A\subseteq\Lambda$, its diameter is defined by $\mathrm{diam}(A)=\max_{i,j\in A} d(i,j)$, and the ball of radius $r$ centered at $j\in\Lambda$ is $B_r(j) = \{i\in\Lambda|d(i,j)\leq r\}$.
At each vertex, we place a qubit, associated with a two-dimensional local Hilbert space $\mathbb{C}^2$. The complete Hilbert space is simply the tensor product of all the local Hilbert spaces. 

For a matrix $A$, we denote by $\|A\|_1$ its trace norm, by $\|A\|_\diamond$ its diamond norm, and by $\|A\|_\mathrm{op}$ its operator norm $\sup_{\left\|\ket{\phi}\right\|=1}\left\|A\ket{\phi}\right\|$. The trace distance between two states $\rho$ and $\sigma$ is $D(\rho,\sigma)=\frac{1}{2}\|\rho-\sigma\|_1$. The fidelity is defined as $F(\sigma,\rho)\coloneqq\mathrm{tr}[(\sqrt{\sigma}\rho \sqrt{\sigma})^{1/2}]^2$. For matrices $A$ and $B$ having the same size, $A\preceq B$ means that $B-A$ is positive semi-definite.

\subsection{Asymptotic notation}
We utilize big-$O$ and related notations. We write $f(n)\leq O(g(n))$ to signify that there exist constants $C>0$ and $n_*>0$ with $f(n)\leq C g(n)$ for all $n\geq n_*$. Similarly $f(n)\leq e^{O(g(n))}$ means $f(n)\leq e^{C g(n)}$ for all $n\geq n_*$. For lower bounds we write $f(n)\geq \Omega(g(n))$ to mean $f(n)\geq  Cg(n)$ for all $n\geq n_*$ and  $f(n)\leq e^{-\Omega(g(n))}$ to mean $f(n)\leq e^{-C g(n)}$ for all $n\geq n_*$.  To combine both upper and lower bounds, we write $f(n)=\Theta(g(n))$, which means $f(n)\leq O(g(n))$ and  $f(n)\geq \Omega(g(n))$ simultaneously. 

\subsection{Quantum circuit complexity}
The circuit complexity of a pure state is the minimum number of basic gates required to construct it, starting from a product state. In this work, we utilize a more general notion of circuit complexity, also applicable to mixed states, based instead on local quantum channels.

Our elementary channels will be two-qubit quantum channels, whose support consists of at most two qubits. 
We define quantum circuit complexity of an $n$-qubit state by asking how many two-qubit channels are required to prepare the state, starting from a product state on $m$ qubits with $m\geq n$. The first $n$ qubits are the system qubits $\mathsf{S}$, while all remaining qubits are treated as ancilla qubits $\mathsf{R}$. We emphasize that our results are valid for an arbitrary choice of the number of ancilla qubits. We denote the set of all two-qubit quantum channels on the $m$-qubit system by
\begin{equation}
    \mathcal{G}=\{\mathcal{N}\,\colon\,|\supp(\mathcal{N})|\leq 2\}.
\end{equation} 

\begin{definition}[$\varepsilon$-robust circuit complexity]\label{def:circuit-complexity} 
We say that an $n$-qubit state $\rho$ has {\it $\varepsilon$-robust relative circuit complexity} $\mathcal{C}_{\varepsilon}(\rho|\sigma)=\mathcal{C}$ with respect to an $m$-qubit reference state $\sigma$, where $m\geq n$, if it can be approximated to trace distance $\varepsilon$ by the application of $\mathcal{C}$ two-qubit channels $\mathcal{N}_i\in\mathcal{G}$ starting from $\sigma$,
\begin{equation}
   \frac{1}{2}\Big\|\rho-\tr_\mathsf{R}(\mathcal{N}_{\mathcal{C}}\circ\mathcal{N}_{\mathcal{C}-1}\circ\cdots \circ\mathcal{N}_1[\sigma])\Big\|_1\leq \varepsilon,
\end{equation}
and $\mathcal{C}$ is the smallest number for which this decomposition holds. For $\varepsilon=0$, we say $\rho$ has {\it exact relative circuit complexity} $\mathcal{C}_0(\rho|\sigma)$ with respect to $\sigma$. We further define {\it $\varepsilon$-robust circuit complexity} $\mathcal{C}_{\varepsilon}(\rho)$ as the minimization of $\mathcal{C}_\varepsilon(\rho|\sigma)$ over the set $\mathcal{M}_m$ of all $m$-qubit product states $\sigma$ and all integers $m\geq n$,
\begin{equation}
    \mathcal{C}_{\varepsilon}(\rho) = \min_{m\geq n}\min_{\sigma\in\mathcal{M}_m}\mathcal{C}_{\varepsilon}(\rho|\sigma).
\end{equation}
\end{definition}
We make a few remarks about the ancilla qubits. First, a circuit consisting of $G$ two-qubit channels can act nontrivially on at most $2G$ qubits, and any ancilla qubits initialized in a product state and left untouched by the circuit have no impact on the system after being traced out. Therefore, it is sufficient to consider at most $2G$ ancilla qubits. Second, we may assume without loss of generality that the ancilla qubits are initialized in a fixed product state. This assumption does not affect either the set of states implementable using $G$ two-qubit channels or the circuit complexities in \Cref{def:circuit-complexity}, since arbitrary two-qubit channels already include the freedom to apply arbitrary single-qubit channels to any ancilla qubits that participate in the circuit. 

The space of states with complexity at most $G$ is highly intricate. Nevertheless, it can be made tractable by discretizing it via a covering net. We prove the following bound on the size of a covering net for the space of channels formed by $G$ two-qubit channels. The result below improves the $(\varepsilon/G)$-dependence of the covering number appearing in Lemma 18 of Ref.~\cite{raza2024}.
\begin{prop}[Lemma 18 of Ref.~\cite{raza2024}]\label{thm:two-qubit-covering-size}
    For any integer $G>0$ and any $\varepsilon>0$, consider the space of channels that can be constructed by concatenating $G$ two-qubit channels $\mathcal{G}_G=\big\{\mathcal{N}_G\circ \mathcal{N}_{G-1}\circ \cdots \circ \mathcal{N}_1\,\colon\, \mathcal{N}_i\in\mathcal{G}, i\in \{1,\dots,G\}\big\}$ acting on $m$ qubits with $m=n+2G$.
    There exists an $\varepsilon$-covering net $\{\mathcal{S}_1,\dots,\mathcal{S}_s\}\subseteq\mathcal{G}_G$ whose size $s$ is upper bounded by
    \begin{equation}
        s \leq ((n+2G)^2(1+512G/\varepsilon)^{256})^G.
    \end{equation}
    Specifically, this means that for any $\mathcal{N}\in \mathcal{G}_G$, we have $\norm{\mathcal{N}-\mathcal{S}_i}_{\diamond}\leq \varepsilon$ for some $i\in\{1,\dots,s\}$.
\end{prop}
\begin{proof}
Note that it is enough to construct an $(\varepsilon/G)$-covering net of the space $\mathcal{G}$ with size $ s \leq m^2 (1+512G/\varepsilon)^{256}$, by the subadditivity of the diamond norm under channel composition~\cite{haah2025}. To that effect, we first construct a covering net of the set of Choi states of the channels in the two-copy space $\mathcal{H}_A\otimes\mathcal{H}_B$. We then use H\"older's inequality to upper bound the diamond norm of a two-qubit channel $\Phi$ by the trace norm of its Choi state $J(\Phi)=\sum_{ij}\ketbra{i}{j}_A\otimes\Phi(\ketbra{i}{j}_B)$ as
\begin{equation}
    \begin{split}
        \|\Phi\|_\diamond 
        &= \max_{\substack{-\sigma\otimes I_B\preceq Y \preceq \sigma\otimes I_B\\ \tr(\sigma)=1}}\tr(Y J(\Phi))\\
        &\leq \max_{\substack{-\sigma\otimes I_B\preceq Y \preceq \sigma\otimes I_B\\ \tr(\sigma)=1}}\|Y\|_\infty \|J(\Phi)\|_1\\
        &\leq \|J(\Phi)\|_1.
    \end{split}
\end{equation}

The Choi state $J(\Phi)$ can be decomposed by Pauli strings $\{\hat{P}_a^{(A)}\otimes \hat{P}_b^{(B)}\}_{a,b}$ acting on $\mathcal{H}_A\otimes\mathcal{H}_B$ that are orthonormal under the Hilbert-Schmidt inner product as
\begin{equation}
    J(\Phi) = \sum_{a,b\in\{I,X,Y,Z\}^{\otimes 2}} c_{a,b} \hat{P}_a^{(A)}\otimes \hat{P}_b^{(B)}.
\end{equation}
Since $\Phi$ is a positive map, we have $J(\Phi)\succeq 0$. In addition, since $\Phi$ is trace-preserving, it follows that $\tr_B(J(\Phi))=I_A$. From these, we can compute $\|J(\Phi)\|_1$ as
\begin{equation}
    \|J(\Phi)\|_1 = \tr(J(\Phi)) = \tr(I_A) = 4.
\end{equation}
This can be used to upper bound the Hilbert-Schmidt norm as 
\begin{equation}
    \|J(\Phi)\|_2^2=\sum_{a,b\in\{I,X,Y,Z\}^{\otimes 2}} c_{a,b}^2 \leq \|J(\Phi)\|_1^2 \leq 4^2.
\end{equation}
Therefore, each coefficient $c_{a,b}$ should be in $[-4,4]$. 

Now, let us construct a hypercubic lattice $\Lambda_l$ with a spacing $l>0$ in the space of Hermitian matrices acting on $\mathcal{H}_A\otimes\mathcal{H}_B$,
\begin{equation}
    \Lambda_l = \left\{\sum_{a,b\in\{I,X,Y,Z\}^{\otimes 2}}(z_{a,b}l) \hat{P}_a^{(A)}\otimes \hat{P}_b^{(B)}\,\Big| \,z_{a,b}\in\mathbb{Z},z_{a,b}l\in[-4,4]\right\}.
\end{equation}
Then, for any two-qubit channel $\Phi$, there exists $J_\Lambda\in\Lambda_l$ such that
\begin{equation}
    \|J(\Phi)-J_\Lambda\|_1 \leq 4\|J(\Phi)-J_\Lambda\|_2 \leq 4\times\frac{l}{2}\times 4^2.
\end{equation}
From this, we can construct the covering net $S_\lambda$ of two-qubit channels as follows:
\begin{equation}
    S_\lambda = \left\{\Pi_\mathcal{G}(\tilde{J})\,\big|\,\tilde{J}\in\Lambda_l\right\}.
\end{equation}
Here, $\Pi_\mathcal{G}$ is defined as
\begin{equation}
    \Pi_\mathcal{G}(\tilde{J})=\operatornamewithlimits{argmin}_{\tilde{\Phi}\in\mathcal{G}}\|\tilde{J}-J(\tilde{\Phi})\|_2.
\end{equation}
Then, the distance between $\Phi$ and $\Phi_\Lambda=\Pi_\mathcal{G}(J_\Lambda)$ is upper bounded by
\begin{equation}
    \begin{split}
        \|\Phi-\Phi_\Lambda\|_\diamond
        &\leq\|J(\Phi)-J(\Phi_\Lambda)\|_1 \\
        &\leq 4\|J(\Phi)-J(\Phi_\Lambda)\|_2 \\
        &\leq 4\|J(\Phi)-J_\Lambda\|_2 + 4\|J_\Lambda-J(\Phi_\Lambda)\|_2\\
        &\leq 8\|J(\Phi)-J_\Lambda\|_2\\
        &\leq 8^2 l.
    \end{split}
\end{equation}
Therefore, by choosing $l$ as $\lambda/8^2$, $S_\lambda$ indeed becomes a $\lambda$-covering net of $\mathcal{G}$.

Finally, let us compute the size of the covering. Since $S_\lambda$ is induced from $\Lambda_l$, its size is equivalent to that of $\Lambda_l$, which is bounded by
\begin{equation}
    |\Lambda_l|\leq (1+8/l)^{4^4} = (1+512/\lambda)^{256}.
\end{equation}
By setting $\lambda=\varepsilon/G$, we get
\begin{equation}
    |S_{\varepsilon/G}|\leq (1+512 G/\varepsilon)^{256}.
\end{equation}

For an $m$-qubit system, there are $\binom{m}{2}\leq m^2$ pairs of qubits. Multiplying the right-hand side by this gives the stated upper bound on the covering size.
\end{proof}

\subsection{Hamiltonian}
We consider a Hamiltonian on a hypercubic periodic lattice of arbitrary dimension. For simplicity, let us expand it in the Pauli basis
\begin{equation}
\label{eq:paliexpansion}
    H=\sum_P c_P P.
\end{equation}
 We allow for periodic boundary conditions where the boundary qubits are identified. For simplicity, we assume $c_I=0$ so that $\tr(H)=0$, and additionally we will demand three physical properties of $H$ \cite{huang2024randomproduct}:

\begin{enumerate}
    \item[(P1)] {\it Locality}. There exists a constant $r$ (called the range of $H$), such that $c_P=0$ if the diameter of $\mathrm{supp}(P)$ is larger than $r$.
    \item[(P2)] {\it Extensivity}. There exists a constant $h_{\min}>0$ such that for any vertex $v$, a ball $B_r(v)$ of radius $r\geq 1$ centered at $v$ contains a term $|c_P| \geq h_{\min}$ which acts within the region $\mathrm{supp}(P)\subseteq B_r(v)$. 
    \item[(P3)] {\it Boundedness}. For all Pauli terms $P$ of $H$, $|c_P|\leq 1$.
\end{enumerate}

\subsection{Ensembles of product states}
\label{sec:Ensemblesofproductstates}
We are interested in the dynamics of generic, physical states, with initially low complexity, and ask how much complexity increases solely due to the unitary evolution generated by $H$. The most interesting case, to which we restrict our attention, is to study initial product states, which have trivial complexity. We define the ensemble of all product states, where each qubit is Haar random, as follows:
\begin{equation}\label{eq:rand-prod-state}
  \textit{{(Random product states ensemble)}}\qquad\qquad \mathcal{E}_0=\Big\{\bigotimes_{i=1}^n \ket{\psi_i} \, \Big | \ket{\psi_i}\sim \mathrm{Haar}(\mathbb C^2)\Big\}.
\end{equation}
Typically, a state $\ket{\psi}\sim \mathcal{E}_0$ will have an infinite effective temperature under $H$, but oftentimes the physical systems of interest have a finite effective temperature. To model such a situation we consider the ensemble of random product states whose energy falls inside a defined energy window. Specifically, for any inverse temperature $\beta$, we consider
\begin{equation}\label{eq:mc-prod-state}
  \textit{{(Product-state microcanonical ensemble)}}\qquad \mathcal{E}_{\beta,\delta}=\Big\{\ket{\psi}\sim\mathcal{E}_0 \, \Big |\, |\expval{H}{\psi}-E(\beta)|<\delta\sqrt{n}\Big\},
\end{equation}
where  $E(\beta)=\tr(\exp(-\beta H) H )/\tr(\exp(-\beta H))$ is the thermal energy. We will require that $\delta \geq 2^{2\mathcal{V}+1}$, where $\mathcal{V}\leq(2r+1)^D$ is the number of vertices inside a ball of radius $r$ in the lattice $\Lambda$. The parameter $\delta$ is otherwise arbitrary apart from this lower bound. 

\subsection{Finite correlation length of Gibbs state}
Throughout the paper, we denote the Gibbs state at inverse temperature $\beta$ by $g_\beta=\exp(-\beta H)/\tr(\exp(-\beta H))$. Later, we will restrict attention to temperatures for which $g_\beta$ has a finite correlation length. This property has been formalized as follows~\cite{Kliesch2014,Brandao2019,Bluhm2022,chen2025}.
\begin{definition}[Uniform clustering]
    Consider a bounded local Hamiltonian $H$. We say that $H$ is {\it uniformly clustering} at inverse temperature $\beta$ if for any $l\geq 0$, any region $X$, any reduced Hamiltonian in $X$, and any subregions $A\subseteq X$ and $C\subseteq X$ such that $d(A,C)\geq l$, we have
    \begin{equation}
        \abs{\mathrm{Cov}_{g_{\beta,X}}(O_A,O_C)}\leq\mathrm{poly}(|A|,|C|)e^{-l/\xi}
    \end{equation}
    for some correlation length $\xi>0$ and any operators $O_A$ and $O_C$ with unit operator norm supported on $A$ and $C$, respectively, where $g_{\beta,X}=\tr_{X^c}(g_\beta)$.
\end{definition}
This condition holds true in 1D at any temperature~\cite{Araki1969,Bluhm2022}, and in arbitrary dimension at high enough temperature, independent of system size~\cite{Kliesch2014,Capel2025}.

With these definitions, our main result is an unconditional statement guaranteeing that with probability exponentially close to unity, an initial product state sampled from the ensembles above has an exponentially large late-time complexity upon evolution under a generic, local Hamiltonian.

In the next two sections, we will introduce two technical properties, one on the spectrum of the Hamiltonian, spectral ergodicity, and the other on the population in the energy eigenbasis of any state, called energy dispersion. We prove that these conditions are generic for local Hamiltonians and product states. Later, we will leverage these properties to prove a late-time exponentially large quantum circuit complexity. 

\section{Spectral ergodicity}
\label{sec:02}
Spectral ergodicity is a condition on the energy levels which allows one to compute infinite time averages and higher moments of states evolving unitarily under the Hamiltonian. We present the definition, and then we show that this property is satisfied generically for local Hamiltonians. 
\begin{definition}[$k$-th no-resonance condition and spectral ergodicity]
    Consider a Hamiltonian $H$ with energy levels $E_1,\dots,E_{2^n}$. For any integer $k>0$, we say $H$ satisfies the $k$-th no-resonance condition \cite{Kaneko2020,Mark2024,Riddell2023} if, for any two $k$-tuples of energy levels $(E_{\mu_1},\dots,E_{\mu_k})$ and $(E_{\nu_1},\dots,E_{\nu_k})$, their sums are equal
    \begin{equation}
        \sum_{i=1}^k E_{\mu_i}=\sum_{i=1}^k E_{\nu_i}
    \end{equation}
     if and only if the indices $(\mu_1,\dots,\mu_k)$ and $(\nu_1,\dots,\nu_k)$ constitute  permutations of each other. Whenever this is satisfied for arbitrary $k$, we say $H$ satisfies {\it spectral ergodicity}.
\end{definition}
\noindent We remark that the $k$-th no-resonance condition implies the $k'$-th no-resonance condition for all $k'\leq k$.

The $k$-th no-resonance condition has been leveraged as a basic assumption in many prior works studying dynamics of many-body systems~\cite{Kaneko2020,Mark2024,Riddell2023}. For $k=2$ this condition simply means the absence of gap degeneracies in the spectrum, which is a basic assumption present in most works that rigorously study quantum equilibration \cite{Reimann2008,Linden2009,Muller2015,Pilatowsky-Cameo2025}.  The notion of spectral ergodicity was leveraged in Ref.~\cite{Riddell2023} (under the name of \textit{generic spectrum}) to prove results on the concentration of quantum equilibration and recurrence times in dynamics.
\subsection{Proposition 1: Generic local Hamiltonians have spectral ergodicity}
Although expected to hold in general, the $k$-th no-resonance conditions have been difficult to prove. Reference~\cite{Huang2021} proved that generic local Hamiltonians have no gap degeneracies ($k=2$). Here, we generalize this result to generic local Hamiltonians for arbitrary $k$. We remark that Ref.~\cite{Riddell2023} proves that generic complex Hermitian matrices have spectral ergodicity. However, this result does not translate immediately to local Hamiltonians, since the set of local Hamiltonians has zero measure within the set of all complex Hermitian matrices.

\begin{prop}[Generic local Hamiltonians have spectral ergodicity; Proposition 1 in the main text]
   \label{thm:generic-ergodic}
Let $\mathcal{P}$ be a set of Pauli strings which contains all single-site $Z_i$ and $X_i$, as well as all nearest-neighbor $Z_iZ_j$ and $X_iZ_j$  strings. Then almost all Hamiltonians that can be constructed from terms $P\in \mathcal{P}$ have spectral ergodicity. Specifically, the set of coefficients $(c_P)_{P\in\mathcal{P}}\in [-1,1]^{|\mathcal{P}|}$ such that the Hamiltonian $H=\sum_{P\in\mathcal{P}} c_P P$ does not have spectral ergodicity has zero measure in $[-1,1]^{|\mathcal{P}|}$.
\end{prop}

We first remark that this proposition follows from the existence of a one-dimensional nearest-neighbor Hamiltonian that has spectral ergodicity. The specific Hamiltonian used in our proof [see Eq.~\eqref{eq:ergodic-ham}] contains single-site Pauli terms $Z_i$ and $X_i$ as well as nearest-neighbor Pauli terms $Z_iZ_j$ and $X_iZ_j$. \Cref{thm:generic-ergodic} therefore applies to any set of Pauli strings containing these terms. In particular, it applies to the set of all Pauli strings whose support has diameter at most $r\geq 1$, as considered in the main text. Operationally, \Cref{thm:generic-ergodic} states that if one were to sample the coefficients $c_P$ of the Hamiltonian $H=\sum_{P\in \mathcal{P}} c_P P$ at random, then with unit probability $H$ has spectral ergodicity. 
Furthermore, this statement remains true if we only consider Hamiltonians which satisfy our extensivity condition (P2), which requires that, in any ball of radius $r$, at least some coefficient $|c_P|\geq h_{\mathrm{min}}$ for a Pauli $P$ acting in the region, since the Hamiltonians that satisfy this condition define a subset of $[-1,1]^{|\mathcal{P}|}$ which has finite, nonzero measure.

\begin{proof}
 We will show that for a fixed $k$, the set of coefficients that generate Hamiltonians which violate the $k$-th no-resonance condition has measure zero. This implies the desired result since a countable union of measure-zero sets also has measure zero. Following Refs.~\cite{Keating2015,Huang2021,Riddell2023}, we begin by showing that one can prove the statement by constructing a single example of a Hamiltonian $H$ which satisfies the $k$-th no-resonance condition, and then explicitly construct such an example. 

Let $\bm c=(c_P)_{P\in\mathcal{P}}$ denote a tuple of coefficients defining the Hamiltonian $H(\bm c)$. To show that the statement reduces to finding a single example, we introduce the following function of $\bm c$:
\begin{equation}
    G_k(\bm c) = \prod_{(\mu_1,\dots,\mu_k)\notin[(\nu_1,\dots,\nu_k)]}\left(\sum_{i=1}^k E_{\mu_i}-\sum_{i=1}^k E_{\nu_i}\right),
\end{equation}
where $\{E_\mu\}_{\mu=1}^{2^n}$ are the energies of $H(\bm c)$ and $[(\nu_1,\dots,\nu_k)]$ is the set of all permutations of $(\nu_1,\dots,\nu_k)$. Importantly, $H(\bm c)$ satisfies the $k$-th no-resonance condition if and only if $G_k(\bm c)\neq 0$. Since $G_k(\bm c)$ is a symmetric polynomial in $\{E_\mu\}_{\mu=1}^{2^n}$ (it is invariant upon exchanging energies), due to the fundamental theorem of symmetric polynomials, $G_k(\bm c)$ is a polynomial in $F_j=\tr(H(\bm c)^j)$, and in particular is a polynomial in $\bm c$. Since the zeros of a multivariate polynomial $[-1,1]^{\abs{\mathcal{P}}}\rightarrow\mathbb{R}$ are of measure zero unless the polynomial is identically zero, it is sufficient to show that there exists a single set $\bm c_0$ of parameters such that the corresponding Hamiltonian $H(\bm c_0)$ satisfies the $k$-th no-resonance condition.

Lemma~\ref{thm:local-ergodic} below constructs a 1D spin-chain Hamiltonian with single-site $X_i$ and $Z_i$ terms and nearest-neighbor $Z_iZ_{i+1}$ and $X_iZ_{i+1}$ interactions, which satisfies the $k$-th no-resonance condition. Moreover, this property is preserved under an overall normalization of the Hamiltonian. Note that constructing the example in 1D is sufficient to prove the result in arbitrary dimension, since an arbitrary lattice can be covered by a single one-dimensional line, which {\it snakes} through the whole space~\cite{Huang2021}. 

\end{proof}

\begin{lemma}[A spin-chain Hamiltonian with nearest-neighbor interactions which satisfies the $k$-th no-resonance condition]\label{thm:local-ergodic}
    For any number of qubits $n$, consider the following Hamiltonian
    \begin{equation}\label{eq:ergodic-ham}
        H_n = Z_1 + \sum_{j=1}^{n-1} [h_j Q_j + J_j(g_j+Q_j)Z_{j+1}],
    \end{equation} 
    with $Q_j=X_j+Z_j$. For any $k$, there exists a choice of parameters $(h_1,g_1,J_1,\cdots,h_{n-1},g_{n-1},J_{n-1})\in\mathbb{R}^{3(n-1)}$ such that $H_n$ satisfies the $k$-th no-resonance condition.
\end{lemma}
\begin{proof}

This proof is inspired by Lemma 7 of Ref.~\cite{Huang2021}. We denote $\{E^{(n)}_\mu\}_{\mu=1}^N$ as the energies of $H_n$. Here, $N=2^n$ is the Hilbert space dimension. For any positive integer $q$, we will call $(m_1,\dots,m_q)\in\mathbb{Z}^q$ a {\it multiplicity vector} if it satisfies $\sum_{\mu=1}^q m_\mu=0$ and $0<\sum_{\mu=1}^q\abs{m_\mu}\leq 2k$. We will show that $H_n$ with appropriately chosen parameters satisfies the $k$-th no-resonance condition. This reduces to showing that
\begin{equation}\label{eq:no-less-than-k-res}
    \sum_{\mu=1}^N m_\mu E_\mu^{(n)}\neq 0
\end{equation}
for all multiplicity vectors $(m_1,\dots,m_N)\in\mathbb{Z}^N$.

Given a multiplicity vector $(m_\mu)_\mu$, moving the terms with negative $m_\nu$ to the right-hand side of Eq.~\eqref{eq:no-less-than-k-res} yields an inequality of the form 
\begin{equation}
    \sum_{m_\mu>0} m_\mu E_\mu^{(n)} \neq \sum_{m_\nu<0} (-m_\nu) E_\nu^{(n)}
\end{equation}
in which all coefficients are positive integers. This is precisely a no-resonance condition for an energy combination of order at most $k$. Conversely, because we require $\sum_{\mu=1}^N\abs{m_\mu}\leq 2k$, every energy combination of order at most $k$ gives the corresponding multiplicity vector. The two notions are therefore equivalent.

Before proceeding, we introduce the following Hamiltonian and its energies $\{\tilde{E}_\mu^{(n)}(x)\}_{\mu=1}^N$:
\begin{equation}
    \tilde{H}_n (x) = H_n + xQ_n.
\end{equation}

{\it Induction hypotheses.} The following two conditions hold:
\begin{enumerate}
    \item[(H1)$_n$] $\,\,$ For all multiplicity vectors $(m_1,\dots,m_N)\in\mathbb{Z}^N$, $\sum_{\mu=1}^N m_\mu E_\mu^{(n)} \neq 0$.
    \item[(H2)$_n$] $\,\,$ There exists an open interval $\mathcal{I}_n\subset\mathbb{R}$ such that $\tilde{E}_\mu^{(n)}(x)$ is a real analytic function in $\mathcal{I}_n$, and for all multiplicity vectors $(m_\mu)_\mu$, $\sum_{\mu=1}^N m_\mu\tilde{E}_\mu^{(n)\prime\prime}(x)$ is not identically zero in $\mathcal{I}_n$. 
\end{enumerate}
Here, the prime denotes the derivative.

For the base case $n=1$, the Hamiltonian is $H_1=Z_1$. Since we have $m_1+m_2=0$, $m_1=-m_2=s$, and from $0<|m_1|+|m_2|\leq 2k$ we have $s\neq 0$ and $|s|\leq k$. For (H1)$_1$, $m_1 E_1^{(1)}+m_2 E_2^{(1)}=m_1-m_2=2s\neq 0$. For (H2)$_1$, we have $\tilde{H}_1(x)=x X_1+(1+x)Z_1$. This gives $\tilde{E}_1^{(1)}(x)=\sqrt{2x^2+2x+1}$ and $\tilde{E}_2^{(1)}(x)=-\sqrt{2x^2+2x+1}$. Thus, $\sum_{\mu=1}^2m_\mu\tilde{E}_\mu^{(1)}=2s\sqrt{2x^2+2x+1}$ is a real analytic function in $\mathcal{I}_1=(-1,1)$. Additionally, we have
\begin{equation}
    \sum_{\mu=1}^2m_\mu\tilde{E}_\mu^{(1)\prime\prime}(x) = \frac{2s}{(2x^2+2x+1)^{3/2}},
\end{equation}
which is not identically zero in $\mathcal{I}_1$.

Now, we perform induction. Let us assume the induction hypotheses (H1)$_n$ and (H2)$_n$. We can write $H_{n+1}$ as
\begin{equation}
    H_{n+1} = H_n + h_n Q_n + J_nQ_nZ_{n+1} + J_ng_n Z_{n+1}
\end{equation} 
and its eigenvalues as
\begin{align}
        E_{\mu,+}^{(n+1)} &= \tilde{E}^{(n)}_\mu(h_n+J_n) + J_n g_n, &\text{and}&&
        E_{\mu,-}^{(n+1)} &= \tilde{E}^{(n)}_\mu(h_n-J_n) - J_n g_n,
\end{align}
where `$+$' and `$-$' represent the up and down states of the $(n+1)$-th qubit in the $z$-direction. The corresponding eigenstates have well-defined values of $Z_{n+1}$.

Next, (H2)$_n$ implies that for every multiplicity vector $(m_\mu)_\mu$, $\sum_{\mu=1}^N m_\mu\tilde{E}_\mu^{(n)}(x)$ and $\sum_{\mu=1}^N m_\mu\tilde{E}_\mu^{(n)\prime}(x)$ are real analytic functions that are not identically zero in $\mathcal{I}_n$. Since they are analytic in $\mathcal{I}_n$, they have discrete zero sets. Furthermore, since the number of different multiplicity vectors is finite, the union of all their zero sets cannot cover $\mathcal{I}_n$. Therefore, there exist $\xi>0$ and $\xi_0\in\mathcal{I}_n$ such that for all $x\in[\xi_0-\xi,\xi_0+\xi]$ and multiplicity vectors $(m_\mu)_\mu$, the following four properties hold: $x\in\mathcal{I}_n$, $\sum_{\mu=1}^N m_\mu\tilde{E}_\mu^{(n)}(x)\neq0$, $\sum_{\mu=1}^N m_\mu\tilde{E}_\mu^{(n)\prime}(x)\neq0$, and $\sum_{\mu=1}^N m_\mu\tilde{E}_\mu^{(n)\prime\prime}(x)\neq0$. We set $h_n$ as $\xi_0$. Additionally, we introduce the following non-zero constants:
\begin{align*}
        \alpha &= \min_{\substack{(m_1,\dots,m_N)\\ |x-h_n|\leq \xi}} \left|\sum_{\mu=1}^N m_\mu \tilde{E}_\mu^{(n)}(x)\right|,&
        \lambda &= \min_{\substack{(m_1,\dots,m_N)\\ |x-h_n|\leq \xi}} \left|\sum_{\mu=1}^N m_\mu \tilde{E}_\mu^{(n)\prime}(x)\right|,&
        \chi &= \min_{\substack{(m_1,\dots,m_N)\\ |x-h_n|\leq \xi}} \left|\sum_{\mu=1}^N m_\mu \tilde{E}_\mu^{(n)\prime\prime}(x)\right|.
\end{align*}

We will show that there exist $J_n,g_n>0$ that make the induction hypotheses hold for $H_{n+1}$. To this aim, we will analyze the simultaneous constraints required by each condition (H1)$_{n+1}$ and (H2)$_{n+1}$, and we will identify a choice for $J_n,g_n>0$ that satisfies these constraints. \\

\noindent {{\bf Claim 1:} (H1)$_{n+1}$ {\it is satisfied for $0<J_n<\min\{1,\xi,\alpha/(2C_1k+C_2k),\lambda/(C_2k)\}$ and $g_n>k(2C_0/J_n+2C_1+C_2)$, for some $C_0,C_1,C_2>0$.\\}}
{\it Proof.} 
Let $(m_{\mu,+},m_{\mu,-})_{\mu}\in\mathbb{Z}^{2N}$ be a multiplicity vector. Let us denote
\begin{align*}
        F&=\sum_{\mu=1}^N [m_{\mu,+} E_{\mu,+}^{(n+1)}+m_{\mu,-} E_{\mu,-}^{(n+1)}]\\
        &=J_ng_n\sum_{\mu=1}^N(m_{\mu,+}-m_{\mu,-})+\frac{1}{2}\sum_{\mu=1}^N(m_{\mu,+}+m_{\mu,-})(\tilde{E}^{(n)}_\mu(h_n+J_n)+\tilde{E}^{(n)}_\mu(h_n-J_n))\\
        &\quad\quad\quad+\frac{1}{2}\sum_{\mu=1}^N(m_{\mu,+}-m_{\mu,-})(\tilde{E}^{(n)}_\mu(h_n+J_n)-\tilde{E}^{(n)}_\mu(h_n-J_n)).
\end{align*}
We will show that, under the assumptions of the claim, there is no multiplicity vector $(m_{\mu,+},m_{\mu,-})_{\mu}$ that gives $F=0$. Since $\tilde{E}_\mu^{(n)}(x)$ is analytic in $(h_n-\xi,h_n+\xi)$, for $J_n<\xi$, we can rewrite $F$ as
\begin{equation}
    F = F_1 + F_2 + F_3 + F_4
\end{equation}
with
\begin{equation}
    \begin{split}
        F_1 &= J_ng_n\sum_{\mu=1}^N(m_{\mu,+}-m_{\mu,-})\\
        F_2 &= \sum_{\mu=1}^N(m_{\mu,+}+m_{\mu,-})\tilde{E}^{(n)}_\mu(h_n)\\
        F_3 &= J_n\sum_{\mu=1}^N(m_{\mu,+}-m_{\mu,-})\tilde{E}_\mu^{(n)\prime}(h_n+\xi_1)\\
        F_4 &= \frac{J_n^2}{2}\sum_{\mu=1}^N(m_{\mu,+}+m_{\mu,-})\tilde{E}_\mu^{(n)\prime\prime}(h_n+\xi_2)
    \end{split}
\end{equation}
for some $0\leq|\xi_1|,|\xi_2|\leq J_n$ due to the mean value theorem and its generalization for the second derivative. Additionally, for $x\in(h_n-\xi,h_n+\xi)$, $\tilde{E}_\mu^{(n)}(x)$, $\tilde{E}_\mu^{(n)\prime}(x)$, and $\tilde{E}_\mu^{(n)\prime\prime}$ are upper bounded by some positive constants $C_0$, $C_1$, and $C_2$, respectively. 

Let us first consider multiplicity vectors giving $|F_1|\neq 0$. In this case, $|F_1|$ is lower bounded by $J_ng_n$. Thus, for $g_n>k(2C_0/J_n+2C_1+J_n C_2)$, we have
\begin{equation}
    |F_2+F_3+F_4| \leq 2C_0k + 2J_nC_1k + J_n^2 C_2k < J_n g_n \leq |F_1|,
\end{equation}
which implies that we have $F\neq 0$ in this case. 

Second, let us consider multiplicity vectors giving $F_1=0$ and $F_2\neq 0$. Since we have
\begin{align}
        \sum_{\mu=1}^N (m_{\mu,+}+m_{\mu,-}) &= 0 &\text{and}&&
        \sum_{\mu=1}^N\abs{m_{\mu,+}+m_{\mu,-}} &\leq \sum_{\mu=1}^N(\abs{m_{\mu,+}}+\abs{m_{\mu,-}}) \leq 2k,
\end{align}
$(m_{\mu,+}+m_{\mu,-})_\mu$ can serve as a multiplicity vector. This implies that $|F_2|$ is lower bounded by $\alpha$. Thus, for $J_n < \min\{1,\alpha/(2C_1k+C_2k)\}$, we have
\begin{equation}
    |F_3+F_4| \leq 2 J_n C_1 k + J_n^2 C_2 k < \alpha \leq |F_2|.
\end{equation}
Again, this means that $F\neq 0$ in this case as well.

Third, let us consider multiplicity vectors giving $F_1=F_2=0$ and $F_3\neq 0$. Since $F_2=0$, we have $m_{\mu,+}+m_{\mu,-}=0$ for all $1\leq\mu\leq N$. Due to this, we can simplify $F_3$ as
\begin{equation}
    F_3 = 2J_n\sum_{\mu=1}^N m_{\mu,+} \tilde{E}_\mu^{(n)\prime}(h_n+\xi_1).
\end{equation}
Importantly, since $\sum_{\mu=1}^N(m_{\mu,+}-m_{\mu,-})=0$ due to $F_1=0$, $(m_{\mu,+})_\mu$ is a multiplicity vector. Thus, $|F_3|$ is lower bounded by $2J_n\lambda$, so for $J_n<\lambda/(C_2k)$, we have
\begin{equation}
    |F_4|\leq J_n^2 C_2 k < 2J_n\lambda \leq |F_3|.
\end{equation}
This implies that $F\neq 0$.

Finally, let us consider multiplicity vectors giving $F_1=F_2=F_3=0$. Again, since $F_2=0$, we have $m_{\mu,+}+m_{\mu,-}=0$ for all $1\leq\mu\leq N$. Additionally, $F_1=0$ makes $(m_{\mu,+})_\mu$ a multiplicity vector. This, together with $F_3=0$, implies that $m_{\mu,+}=0$ for all $1\leq\mu\leq N$. Altogether, we have $m_{\mu,+}=m_{\mu,-}=0$ for all $1\leq\mu\leq N$, contradicting the assumption that $(m_{\mu,+},m_{\mu,-})_\mu$ is a multiplicity vector.

\noindent {{\bf Claim 2:} (H2)$_{n+1}$ {\it is satisfied for $0<J_n<\min\{1,\Delta/4\sqrt{2},6C,\lambda/(2C_2k)\}$ and $g_n>\max\{\frac{10}{3}(\tilde{h}/J_n+4),2k(2C_1+C_2+1280),2560k/\lambda,2560k/(\chi J_n)\}$ with $\tilde{h}=\|\tilde{H}(h_n)\|_\mathrm{op}$ for some constants $C,\Delta>0$. Here, $C_1$ and $C_2$ are the constants from  Claim 1.\\}}
{\it Proof.} 
Let $(m_{\mu,+},m_{\mu,-})_{\mu}\in\mathbb{Z}^{2N}$ be a multiplicity vector. Next, let us write $\tilde{H}_{n+1}(x)$ explicitly,
\begin{equation}
    \tilde{H}_{n+1}(x) = \tilde{H}_n(h_n) + J_nQ_nZ_{n+1} + (J_ng_n+x) Z_{n+1} + xX_{n+1}.
\end{equation}
To relate derivatives of $\tilde{E}_\mu^{(n+1)}(x)$ with those of $\tilde{E}_\mu^{(n)}(x)$, it is helpful to rotate Pauli matrices at the $(n+1)$-th qubit. To this end, we introduce the following reparameterization: $y(x)=x+J_ng_n$. Additionally, let us introduce the following parameters: $\kappa(y) = \sqrt{y^2+(y-J_ng_n)^2}$, $\cos\theta=y/\kappa$, and $\sin\theta=(y-J_ng_n)/\kappa$. With the branch choice $\kappa(0)=J_ng_n$, these are analytic in $\abs{y}<J_ng_n/8$, and in that region, we have $4J_ng_n/5<\abs{\kappa(y)}<6J_ng_n/5$, $\abs{\cos\theta}<1/6$, and $\abs{\sin\theta}<\sqrt{2}$. 

Then, we apply the following similarity transformation at the $(n+1)$-th qubit:
\begin{equation}
    S_{n+1} = \begin{bmatrix}
        u & -v\\ v & u
    \end{bmatrix}
    \quad\text{ with }\quad
    u=\sqrt{\frac{1+\cos\theta}{2}}\quad\text{ and }\quad v=\frac{\sin\theta}{2u}.
\end{equation}
This transformation has unit determinant and gives
\begin{equation}
    S_{n+1}^{-1}\tilde{H}_{n+1}(x)S_{n+1} = \tilde{H}^\mathrm{unp}(y) + \tilde{H}^\mathrm{per}(y)
\end{equation}
with
\begin{equation}
    \tilde{H}^\mathrm{unp}(y) = \tilde{H}_n(h_n)+\kappa Z_{n+1}+J_n\cos\theta Q_n Z_{n+1}\quad\text{and}\quad \tilde{H}^\mathrm{per}(y)=-J_n\sin\theta Q_n X_{n+1}.
\end{equation}
Since we have $\sum_{\mu=1}^N m_\mu\tilde{E}_\mu^{(n)}(x)\neq0$ for all $(m_\mu)_\mu$ and $x\in(h_n-\xi,h_n+\xi)$, it follows that $\tilde{H}_n(h_n)$ has a non-zero gap. Let $\Delta>0$ be the gap of $\tilde{H}_n(h_n)$. 

We proceed with a perturbation theory calculation to get the spectrum of $S_{n+1}^{-1}\tilde{H}_{n+1}(x)S_{n+1}$ for $|y|< J_ng_n/8$, $J_n\leq\Delta/(4\sqrt{2})$, and $g_n\geq \frac{10}{3}(\|\tilde{H}_n(h_n)\|_\mathrm{op}/J_n+4)$. First, let us consider diagonal blocks of $\tilde{H}^\mathrm{unp}(y)$: $\tilde{H}_+^\mathrm{unp}(y)$ and $\tilde{H}_-^\mathrm{unp}(y)$ which correspond to $\langle Z_{n+1}\rangle=1$ or $-1$ sectors, respectively, i.e., $\tilde{H}^\mathrm{unp}_{\pm}(y)=\tilde{H}_n(h_n)\pm(\kappa+J_n\cos(\theta)Q_n)$, so $\tilde{H}^\mathrm{unp}(y)=\tilde{H}_+^\mathrm{unp}(y)\oplus\tilde{H}_-^\mathrm{unp}(y)$. Since $\tilde{H}_n(h_n)$ is nondegenerate with gap $\Delta$, Lemma~\ref{thm:uniform-spec-stab} with $(H^\mathrm{unp},V)=(\tilde{H}_n(h_n),z Q_n)$ for $z\in\mathbb{C}$ implies that $\tilde{H}_n(h_n)+zQ_n$ has uniquely labeled analytic eigenvalues $\tilde{E}_\mu^{(n)}(h_n+z)$ for $|z|<C$ for some positive constant $C<\Delta/(8/\sqrt{2})$. Then, for $J_n<6C$, eigenvalues of $\tilde{H}_\pm^\mathrm{unp}(y)$ are $\tilde{E}^\mathrm{unp}_{\mu,\pm}(y)=\tilde{E}_\mu^{(n)}(h_n\pm J_n\cos\theta)\pm\kappa$. 

Next, we introduce a trial energy $\omega$ to compute the spectrum of the Schur complements, or the effective Hamiltonians, of the sectors $\langle Z_{n+1}\rangle=\pm 1$. To this end, we first need to bound the operator norm of Green's function, $\|(\omega-\tilde{H}^\mathrm{unp}_{\mp}(y))^{-1}\|_\mathrm{op}$, when $\omega$ is near the $\langle Z_{n+1}\rangle=\pm 1$ sector to bound their self-energies, by leveraging the fact that energies of the diagonal blocks are separated by $O(J_ng_n)$, which is taken to be a large number. Specifically, we set $\omega\in D_{\mu,\pm}$ with $D_{\mu,\pm}=\{z\in\mathbb{C}:|z-\tilde{E}^\mathrm{unp}_{\mu,\pm}(y)|\leq\Delta/4\}$. We then use the Neumann series bound: $\|(I-T)^{-1}\|_\mathrm{op}\leq (1-\|T\|_\mathrm{op})^{-1}$ for any bounded operator $T$ with $\|T\|_\mathrm{op}<1$ (see Sec. 4.4 of Ch. 1 of Ref.~\cite{Kato1995}). To use this, we set $K_\mp=\omega-(\tilde{H}_n(h_n)\mp\kappa)$ and $T_\mp=\mp J_n\cos\theta K_\mp^{-1}Q_n$. Then, we have $\|K_\mp^{-1}\|_\mathrm{op}=1/\min_\mu|\omega-(\tilde{E}^{(n)}_\mu(h_n)\mp\kappa)|$. From $\abs{\kappa(y)}>4J_n g_n/5$, $\Delta\leq 2\|\tilde{H}_n(h_n)\|_\mathrm{op}$, and $g_n\geq \frac{10}{3}(\|\tilde{H}_n(h_n)\|_\mathrm{op}/J_n+4)$, the triangle inequality gives
\begin{equation}
    \|K_\mp^{-1}\|_\mathrm{op}^{-1} 
    \geq 2\abs{\kappa}-\max_\mu|\tilde{E}^{(n)}_\mu(h_n)|-\max_\mu|\tilde{E}^{(n)}_\mu(h_n\pm J_n\cos\theta)|-\frac{\Delta}{4}\\
    \geq \frac{17}{20}J_n g_n.
\end{equation}
Furthermore, using $\abs{\cos\theta}<1/6$ and $g_n\geq 10\sqrt{2}/3$, $\|T_\mp\|_\mathrm{op}$ is bounded by
\begin{equation}
    \|T_\mp\|_\mathrm{op}\leq \sqrt{2} J_n \abs{\cos\theta} \|K_\mp^{-1}\|_\mathrm{op} < \frac{10\sqrt{2}}{51 g_n} < \frac{1}{17}.
\end{equation}
Using the Neumann series, this allows us to get
\begin{equation}\label{eq:unp-green-bound}
    \|(\omega-\tilde{H}^\mathrm{unp}_{\mp}(y))^{-1}\|_\mathrm{op}\leq \|(I-T_\mp)^{-1}\|_\mathrm{op}\|K_\mp^{-1}\|_\mathrm{op}\leq \frac{\|K_\mp^{-1}\|_\mathrm{op}}{1-\|T_\mp\|_\mathrm{op}} < \frac{5}{4 J_ng_n}.
\end{equation}
The self-energy term is given by $\Sigma_\pm(\omega,y)=J_n^2\sin^2\theta\, Q_n(\omega-\tilde{H}^\mathrm{unp}_\mp(y))^{-1}Q_n$. Using Eq.~\eqref{eq:unp-green-bound} and $\abs{\sin\theta}<\sqrt{2}$, we have
\begin{equation}\label{eq:self-energy-bound}
    \|\Sigma_\pm(\omega,y)\|_\mathrm{op}<5J_n/g_n.
\end{equation}
Furthermore, for $\omega_1$ and $\omega_2$ in $D_{\mu,\pm}$, 
\begin{equation}
    (\omega_1-\tilde{H}^\mathrm{unp}_\mp(y))^{-1}-(\omega_2-\tilde{H}^\mathrm{unp}_\mp(y))^{-1} = (\omega_2-\omega_1)(\omega_1-\tilde{H}^\mathrm{unp}_\mp(y))^{-1}(\omega_2-\tilde{H}^\mathrm{unp}_\mp(y))^{-1}
\end{equation}
implies
\begin{equation}\label{eq:self-energy-difference-bound}
    \|\Sigma_\pm(\omega_2,y)-\Sigma_\pm(\omega_1,y)\|_\mathrm{op} \leq \frac{25}{4g_n^2}|\omega_2-\omega_1|.
\end{equation}

The effective Hamiltonian for each sector is $\tilde{H}^\mathrm{eff}_\pm(\omega,y)=\tilde{H}^\mathrm{unp}_\pm(y)+\Sigma_\pm(\omega,y)$. We introduce $V_\pm(y)=\pm J_n\cos\theta Q_n$ and $W_\pm(\omega,y)=\pm J_n\cos\theta Q_n+\Sigma_\pm(\omega,y)$. By setting $\eta=\sqrt{2}J_n/6+5J_n/g_n\leq\Delta/8$, their operator norms are both upper bounded by $\eta$. Then, we can use Lemma~\ref{thm:uniform-spec-stab} with $V=V_\pm(y)$ as well as $V=W_\pm(\omega,y)$, together with $\tilde{H}^\mathrm{unp}=\tilde{H}_n(h_n)\pm\kappa$, which is normal. This gives the Lipschitz constant $L\leq 2$, unique nondegenerate eigenvalues in $D_{\mu,\pm}$ for  $\tilde{H}^\mathrm{unp}_\pm(y)$ and $\tilde{H}^\mathrm{eff}_\pm(\omega,y)$. Denote $\{\tilde{E}^\mathrm{eff}_{\mu,\pm}\}_\mu$ as the eigenvalues of $\tilde{H}^\mathrm{eff}_\pm(\omega,y)$. From these with Eq.~\eqref{eq:self-energy-bound}, we can bound
\begin{equation}\label{eq:self-energy-from-unp}
    \abs{\tilde{E}^\mathrm{eff}_{\mu,\pm}(\omega,y)-\tilde{E}^\mathrm{unp}_{\mu,\pm}(y)}\leq L\|W_\pm(\omega,y)-V_\pm(y)\|_\mathrm{op} = L\|\Sigma_\pm(\omega,y)\|_\mathrm{op}\leq 10 J_n/g_n < \Delta/4.
\end{equation}
This shows that $\omega\rightarrow\tilde{E}^\mathrm{eff}_{\mu,\pm}(\omega,y)$ maps $D_{\mu,\pm}$ into itself. Furthermore, the same lemma and Eq.~\eqref{eq:self-energy-difference-bound} give
\begin{equation}\label{eq:self-energy-convergence}
    \abs{\tilde{E}^\mathrm{eff}_{\mu,\pm}(\omega_2,y)-\tilde{E}^\mathrm{eff}_{\mu,\pm}(\omega_1,y)}\leq L \|\Sigma_\pm(\omega_2,y)-\Sigma_\pm(\omega_1,y)\|_\mathrm{op}\leq \frac{25}{2g_n^2}\abs{\omega_2-\omega_1}
\end{equation}
with $25/(2g_n^2)<1$. Thus, repeated substitution of $\omega\rightarrow \tilde{E}^\mathrm{eff}_{\mu,\pm}(\omega,y)$ converges to the unique solution of the self-consistent equation $\tilde{E}^\mathrm{eff}_{\mu,\pm}(\omega,y)=\omega$ uniformly for all $\abs{y}<J_ng_n/8$ due to the contraction mapping theorem. We note that $\tilde{E}^\mathrm{eff}_{\mu,\pm}(\omega,y)$ becomes the energy of $\tilde{H}_{n+1}(y)$ when the self-consistent equation holds. Additionally, $\tilde{E}^\mathrm{eff}_{\mu,\pm}(\omega,y)$ is analytic, and Eq.~\eqref{eq:self-energy-convergence} implies that $\partial_\omega(\omega-\tilde{E}^\mathrm{eff}_{\mu,\pm}(\omega,y))\neq 0$. These together with the implicit function theorem uniformly applied to solutions of the self-consistent equation for $\abs{y}<J_ng_n/8$ imply that the energies of $\tilde{H}_{n+1}(y)$ are unique, nondegenerate, analytic, and furthermore have the same label $(\mu,\pm)$ in $\abs{y}<J_ng_n/8$. We denote its energies by $\{\tilde{E}_{\mu,\pm}^{(n+1)}(y)\}_\mu$.

Now, we define the correction term $R_{\mu,\pm}(y)=\tilde{E}_{\mu,\pm}^{(n+1)}(y)-\tilde{E}^\mathrm{unp}_{\mu,\pm}(y)$. Then, Eq.~\eqref{eq:self-energy-from-unp} gives
\begin{equation}
    \sup_{|y|<J_ng_n/8}\abs{R_{\mu,\pm}(y)}\leq 10 J_n/g_n.
\end{equation}
From this, we have
\begin{equation}
    \begin{split}
        \tilde{E}_{\mu,+}^{(n+1)}(x) &= \tilde{E}_\mu^{(n)}(h_n+J_n\cos\theta)+\kappa+R_{\mu,+}(y)\\
        \tilde{E}_{\mu,-}^{(n+1)}(x) &= \tilde{E}_\mu^{(n)}(h_n-J_n\cos\theta)-\kappa+R_{\mu,-}(y)
    \end{split}
\end{equation}
Differentiating, we obtain
\begin{equation}
    \begin{split}
        \tilde{E}_{\mu,+}^{(n+1)\prime}(x) &= J_n\left(\frac{d}{dy}\cos\theta\right)\tilde{E}_\mu^{(n)\prime}(h_n+J_n\cos\theta)+\kappa'+R_{\mu,+}'(y)\\
        \tilde{E}_{\mu,-}^{(n+1)\prime}(x) &= -J_n \left(\frac{d}{dy}\cos\theta\right) \tilde{E}_\mu^{(n)\prime}(h_n-J_n\cos\theta)-\kappa'+R'_{\mu,-}(y)
    \end{split}
\end{equation}
and
\begin{equation}
    \begin{split}
        \tilde{E}_{\mu,+}^{(n+1)\prime\prime}(x) &= J_n^2\left(\frac{d}{dy}\cos\theta\right)^2\tilde{E}_\mu^{(n)\prime\prime}(h_n+J_n\cos\theta)+\kappa''+J_n\left(\frac{d^2}{dy^2}\cos\theta\right)\tilde{E}_\mu^{(n)\prime}(h_n+J_n\cos\theta)+R''_{\mu,+}(y)\\
        \tilde{E}_{\mu,-}^{(n+1)\prime\prime}(x) &= J_n^2\left(\frac{d}{dy}\cos\theta\right)^2\tilde{E}_\mu^{(n)\prime\prime}(h_n-J_n\cos\theta)-\kappa''-J_n\left(\frac{d^2}{dy^2}\cos\theta\right)\tilde{E}_\mu^{(n)\prime}(h_n-J_n\cos\theta)+R''_{\mu,-}(y).
    \end{split}
\end{equation}
Using Cauchy's estimate with radius $J_ng_n/8$ and the following calculations:
\begin{equation}
    \begin{split}
        \kappa''|_{y=0} &= \frac{1}{J_n g_n}\\
        \frac{d}{dy}\cos\theta|_{y=0} &= \frac{1}{J_ng_n}\\
        \frac{d^2}{dy^2}\cos\theta|_{y=0} &= \frac{2}{J_n^2g_n^2},
    \end{split}
\end{equation}
when $y=0$, \textit{i.e.}, $x=-J_n g_n$, we have $|R'_{\mu,+}(0)|,|R'_{\mu,-}(0)|\leq 80/g^2_n$ and $|R''_{\mu,+}(0)|,|R''_{\mu,-}(0)|\leq 1280/J_ng^3_n$.

Now, let us denote
\begin{equation}
    G(y) = \sum_{\mu=1}^N[m_{\mu,+}\tilde{E}_{\mu,+}^{(n+1)\prime\prime}(x)+m_{\mu,-}\tilde{E}_{\mu,-}^{(n+1)\prime\prime}(x)].
\end{equation}
Then, (H2)$_{n+1}$ holds if there exists $x$ such that $|y(x)|< J_n g_n/8$, which is the convergence condition, and $G(y(x))\neq 0$ for all multiplicity vectors $(m_{\mu,+},m_{\mu,-})_\mu$. Since $\tilde{E}_{\mu,+}^{(n+1)}(x)$ and $\tilde{E}_{\mu,-}^{(n+1)}(x)$ are real analytic functions once they are convergent, if we have $G(y(x))\neq 0$ for some $x_*$ such that $|y(x_*)|<J_n g_n/8$, then there exists a finite open neighborhood $\mathcal{I}_{n+1}$ of $x_*$ such that $G(y(x))$ is not identically zero. Below, we show that we can choose $x_*=-J_ng_n$ independent of the multiplicity vector, i.e., $G(0)\neq 0$ for all multiplicity vectors $(m_{\mu,+},m_{\mu,-})_\mu$. Since $y(x_*)=0$, let us consider 
\begin{equation}
    G(0) = G_1+G_2+G_3+G_4
\end{equation}
with
\begin{equation}
    \begin{split}
        G_1 &= \frac{1}{J_ng_n}\sum_{\mu=1}^N(m_{\mu,+}-m_{\mu,-})\\
        G_2 &= \frac{2}{J_ng_n^2}\sum_{\mu=1}^N(m_{\mu,+}-m_{\mu,-})\tilde{E}_\mu^{(n)\prime}(h_n)\\
        G_3 &= \frac{1}{g_n^2}\sum_{\mu=1}^N(m_{\mu,+}+m_{\mu,-})\tilde{E}_\mu^{(n)\prime\prime}(h_n)\\
        |G_4| &\leq 2560k/J_ng_n^3.
    \end{split}
\end{equation}
To show $G(0)\neq 0$, we will use the same argument as when we proved $F\neq 0$.

First, let us consider multiplicity vectors giving $G_1\neq 0$. In this case, $|G_1|$ is lower bounded by $1/(J_ng_n)$. Thus, for $g_n>\max\{1,2k(2C_1+J_nC_2+1280)\}$, we have
\begin{equation}
    |G_2+G_3+G_4| \leq \frac{4C_1k}{J_ng_n^2}+\frac{2C_2k}{g_n^2}+\frac{2560k}{J_ng_n^3} < \frac{1}{J_ng_n} \leq |G_1|.
\end{equation}
Thus, we have $G(0)\neq 0$ in this case.

Second, let us consider multiplicity vectors giving $G_1=0$ and $G_2\neq 0$. Since $G_1=0$, we have $\sum_{\mu=1}^N(m_{\mu,+}-m_{\mu,-})=0$. Thus, $(m_{\mu,+}-m_{\mu,-})_\mu$ is a multiplicity vector. Then, $|G_2|$ is lower bounded by $2\lambda/(J_ng_n^2)$, and for $J_n<\lambda/(2C_2k)$ and $g_n>2560k/\lambda$, we have
\begin{equation}
    |G_3+G_4| \leq \frac{2C_2k}{g_n^2}+\frac{2560k}{J_ng_n^3} < \frac{2\lambda}{J_ng_n^2} \leq |G_2|.
\end{equation}
Therefore, in this case as well, we have $G(0)\neq 0$.

Third, let us consider multiplicity vectors giving $G_1=G_2=0$ and $G_3\neq 0$. As we discussed before, $(m_{\mu,+}+m_{\mu,-})_\mu$ is a multiplicity vector. Thus, $|G_3|$ is lower bounded by $\chi/g_n^2$. This implies that for $J_ng_n>2560k/\chi$, we have
\begin{equation}
    |G_4| \leq \frac{2560k}{J_ng_n^3} < \frac{\chi}{g_n^2} \leq |G_3|.
\end{equation}
Thus, we again have $G(0)\neq 0$ in this case.

Finally, let us show that it cannot occur that $G_1=G_2=G_3=0$ for any multiplicity vector. When $G_1=0$, we have that $(m_{\mu,+}-m_{\mu,-})_\mu$ is a multiplicity vector, and $G_2=0$ implies that $m_{\mu,+}=m_{\mu,-}$ for all $1\leq\mu\leq N$, but at the same time, $G_3=0$ would imply that $m_{\mu,+}=-m_{\mu,-}$ for all $1\leq\mu\leq N$, which is impossible given that $\sum_{\mu=1}^N(|m_{\mu,+}|+|m_{\mu,-}|)>0$. Therefore, we always have $G(0)\neq 0$.\\

Combining Claim 1 and Claim 2, we have shown that (H1)$_{n+1}$ and (H2)$_{n+1}$ are satisfied for $0<J_n<\min\{1,\xi,6C,\alpha/(2C_1k+C_2k),\lambda/(2C_2k),\Delta/4\sqrt{2}\}$ and $g_n>\max\{\frac{10}{3}(\tilde{h}+4),2620k\max\{C_0,C_1,C_2\}/\min\{J_n,\lambda,\chi J_n\}\}$ for some constants $\xi,C,C_0,C_1,C_2,\Delta,\tilde{h}>0$, which completes the inductive step.
\end{proof}

\begin{lemma}\label{thm:uniform-spec-stab}
    Let $H^\mathrm{unp}$ be normal with eigenvalues $\{E_\mu^\mathrm{unp}\}_\mu$ separated by $\Delta$. For any positive constant $\eta<\Delta/4$ and operator $V$ with $\|V\|_\mathrm{op}\leq \eta$, the perturbed Hamiltonian $H(V)=H^\mathrm{unp}+V$ has a unique non-degenerate analytic eigenvalue $E_\mu(V)$ with $\abs{E_\mu(V)-E_\mu^\mathrm{unp}}<\Delta/4$. Furthermore, for any operators $V_1$ and $V_2$ with $\|V_1\|_\mathrm{op}\leq\|V_2\|_\mathrm{op}\leq \eta$, $E_\mu(\cdot)$ is Lipschitz continuous with constant $L=1/(1-4\eta/\Delta)$:
    \begin{equation}
        \abs{E_\mu(V_1)-E_\mu(V_2)}\leq L\|V_1-V_2\|_\mathrm{op}.
    \end{equation}
\end{lemma}
\begin{proof}
    Let $\Gamma_\mu$ be $\{z\in\mathbb{C}:\abs{z-E^\mathrm{unp}_\mu}=\Delta/4\}$. Since $H^\mathrm{unp}$ is normal, its Green's function is diagonalizable and has
    \begin{equation}
        \max_{z\in\Gamma_\mu}\|(z-H^\mathrm{unp})^{-1}\|_\mathrm{op} = \sup_{z\in\Gamma_\mu}\frac{1}{\min_\nu\abs{z-E^\mathrm{unp}_\nu}}\leq \frac{4}{\Delta}.
    \end{equation}
    For $\|V\|_\mathrm{op}\leq\eta$, this gives
    \begin{equation}
        \max_{z\in\Gamma_\mu}\|V(z-H^\mathrm{unp})^{-1}\|_\mathrm{op}\leq\frac{4\eta}{\Delta}<1.
    \end{equation}
    Therefore, using the Neumann series (see Sec. 4.4 of Ch. 1 of Ref.~\cite{Kato1995}), for $z\in\Gamma_\mu$,
    \begin{equation}
        \|(z-H)^{-1}\|_\mathrm{op}
        \leq\frac{\|(z-H^\mathrm{unp})^{-1}\|_\mathrm{op}}{1-\|V(z-H^\mathrm{unp})^{-1}\|_\mathrm{op}}
        \leq\frac{1}{\Delta/4-\eta},
    \end{equation}
    which means that $z-H$ is invertible, and no eigenvalue of $H$ lies on $\Gamma_\mu$. Then, the total projection $P_\mu(V)$ is well-defined:
    \begin{equation}
        P_\mu(V) = \frac{1}{2\pi i}\oint_{\Gamma_\mu}dz(z-H)^{-1},
    \end{equation}
    whose rank equals the number of eigenvalues of $H$ in $\Gamma_\mu$, counted with their multiplicities (see Sec. 1.4 of Ch. 2 of Ref.~\cite{Kato1995}). The same is true for the interpolation $H(s)=H^\mathrm{unp}+V(s)$ with $V(s)=V_1+s(V_2-V_1)$, $V_1=0$, $V_2=V$, and $s\in[0,1]$. Thus, there is no eigenvalue that crosses $\Gamma_\mu$ along this interpolation. Since $H(0)$ has the unique eigenvalue $E_\mu^\mathrm{unp}$ in $\Gamma_\mu$ without multiplicity, the continuity implies that the eigenvalue of $H=H(1)$ in $\Gamma_\mu$ is non-degenerate and unique. Furthermore, that eigenvalue is given by $\tr(H P_\mu(V))$, which is analytic in the perturbation and can have the same label $\mu$ (see Sec. 2.1 of Ch. 2 of Ref.~\cite{Kato1995}). We denote the eigenvalue as $E_\mu(V)$.

    Now, let us prove the Lipschitz continuity. First, the operator norm of the projector is bounded by (see Sec. 1.5 of Ch. 2 of Ref.~\cite{Kato1995})
    \begin{equation}
        \|P_\mu(V)\|_\mathrm{op} \leq \frac{\Delta}{4}\max_{z\in\Gamma_\mu}\|(z-H)^{-1}\|_\mathrm{op} \leq \frac{1}{1-4\eta/\Delta} = L.
    \end{equation}
    Then, for $\|V_1\|_\mathrm{op}\leq\|V_2\|_\mathrm{op}\leq \eta$, the operator norm of $V(s)$ is bounded by $\|V(s)\|_\mathrm{op}\leq\eta$ due to the convexity of the operator norm. Furthermore, for left and right eigenvectors $v_\mu(s)$ and $u_\mu(s)$ of $H(s)$ associated with $E_\mu(V(s))$ such that $v_\mu^\dag(s)u_\mu(s)=1$, we have 
    \begin{equation}
        \frac{d}{ds}E_\mu(s) = v_\mu^\dag(s) (V_2-V_1) u_\mu(s).
    \end{equation}
    Since the projector is given by $P_\mu(V(s))=u_\mu(s)v^\dag_\mu(s)$, it follows that
    \begin{equation}
        \frac{d}{ds}E_\mu(s) = \tr(P_\mu(V(s))(V_2-V_1)).
    \end{equation}
    These give
    \begin{equation}
        \abs{\frac{d}{ds}E_\mu(s)}\leq \|P_\mu(V(s))\|_\mathrm{op}\|V_2-V_1\|_\mathrm{op}\leq L\|V_2-V_1\|_\mathrm{op}.
    \end{equation}
    Finally, integrating both sides over $s\in[0,1]$ gives the stated Lipschitz continuity.
\end{proof}

\subsection{Temporal moments from spectral ergodicity}
The no-resonance conditions are particularly useful for computing temporal averages and higher moments of states evolving under the unitary dynamics generated by the Hamiltonian~\cite{Mark2024}. The proposition below allows us to bound the $k$-th temporal moment of any initial state, assuming the $k$-th no-resonance condition. 
Below, the temporal average is defined as $\mathbb{E}_t[\cdot]=\lim_{T\rightarrow\infty}\frac{1}{T}\int_0^T(\cdot)dt$, $S_k$ is the symmetric group of order $k$, and $P_\pi$ is the representation of $\pi\in S_k$ in the $k$-copy tensor product Hilbert space. $P_\pi$ permutes the $k$ copies according to $\pi$ and acts trivially within each copy.
\begin{prop}[Bound on temporal moments under time-independent Hamiltonian evolution]\label{thm:no-k-res-temp-avg-symm}
    Let $H$ be an $n$-qubit Hamiltonian satisfying the $k$-th no-resonance condition, with eigenstates $\{\ket{\mu}\}_{\mu=1}^{2^n}$. Then, for any pure state $\ket{\psi}=\sum_{\mu=1}^{2^n}c_\mu\ket{\mu}$, the $k$-th temporal moment of its evolution $\ket{\psi(t)}=e^{-iHt}\ket{\psi}$ can be upper bounded as
    \begin{equation}
    \label{eq:lemmaineq}
        \E_{t}[\dyad{\psi(t)}^{\otimes k}] \preceq  \rho_d^{\otimes k}\sum_{\pi\in S_k} P_\pi
    \end{equation}
    with the energy diagonal state $\rho_d = \sum_{\mu=1}^{2^n}\abs{c_\mu}^2\ketbra{\mu}$.
\end{prop}
\begin{proof}
We begin by expanding the temporal moment in the energy eigenbasis utilizing a standard argument that was first established by Ref.~\cite{Mark2024}. Let us denote by $\bm\mu=(\mu_1,\cdots,\mu_k)$ a $k$-tuple of indices, and the corresponding tensor-product of energy eigenstates $\ket{\bm\mu}\coloneqq \ket{\mu_1}\otimes \cdots \otimes \ket{\mu_k}$, so that
\begin{equation}
    \E_t[\ketbra{\psi(t)}^{\otimes k}] = \sum_{ \bm\mu, \bm \nu}  \biggl( \prod_{i=1}^k c_{\mu_i}c_{\nu_i}^* \biggr)\E_t\biggl[\exp\biggl(-i\sum_{i=1}^k(E_{\mu_i}-E_{\nu_i})t\biggr)\biggr]\ketbra{\bm\mu}{\bm\nu}.
\end{equation}
The $k$-th no-resonance condition implies that the phase $\exp\bigl(-i\sum_{i=1}^k(E_{\mu_i}-E_{\nu_i})t\bigr)$ averages to zero unless $\bm \mu$ is a permutation of $\bm \nu$, in which case it equals exactly 1. Thus, by denoting the set of all permutations of $\bm \mu$ by $[\bm \mu]$ (its {\it orbit} over the action of the symmetric group) and defining $p_{\bm \mu}\coloneqq \prod_{i=1}^k |c_{\mu_i}|^2$, we can write
\begin{equation}
\label{eq:timeaverage}
    \E_t[\ketbra{\psi(t)}^{\otimes k}] = \sum_{ \bm\mu} \sum_{\bm\nu\in[\bm\mu] } p_{\bm \mu} \ketbra{\bm\mu}{\bm\nu}    = \sum_{[\bm \mu]} p_{\bm \mu}\sum_{\bm \eta,\bm\nu \in [\bm \mu]}\ketbra{\bm\eta }{\bm\nu},
\end{equation}
 where the second equality is obtained by expanding the sum over all tuples, $\sum_{\bm \mu}$, as the sum over all possible orbits, $\sum_{[\bm\mu]}$, followed by a sum over each element $\bm \eta \in [\bm \mu]$, where we used that $p_{\bm \mu}$ is invariant under permutations of~$\bm \mu$.

On the other hand, we can similarly expand the right-hand side of Eq.~\eqref{eq:lemmaineq} in the same basis,
\begin{equation}
\rho_d^{\otimes k}\sum_{\pi} P_\pi = \sum_{\bm\mu} p_{\bm\mu} \dyad{\bm\mu} \sum_\pi P_\pi = \sum_{[\bm\mu]} p_{\bm\mu}  \sum_{\bm \eta \in [\bm \mu]}\ket{\bm\eta}\sum_\pi\bra{\pi(\bm\eta)}= \sum_{[\bm\mu]} p_{\bm\mu} \frac{k!}{|[\bm\mu]|} \sum_{\bm\eta, \bm\nu \in [\bm\mu]} \ketbra{\bm\eta}{\bm\nu},\label{eq:diagstateexpansion}
\end{equation}
where the last equality follows from the orbit-stabilizer theorem, which guarantees that applying all possible permutations $\pi\in S_k$ to a fixed representative $\bm\eta\in[\bm \mu]$ yields each orbit element exactly $k!/|[\bm\mu]|$ times, where $|[\bm \mu]|\leq k!$ is the size of the orbit. Comparing Eqs.~\eqref{eq:timeaverage} and~\eqref{eq:diagstateexpansion}, we obtain the desired inequality by noting that the coefficients satisfy $p_{\bm \mu}\,\leq p_{\bm \mu}(k!/|[\bm\mu]|)$  and each term $\sum_{\bm \eta,\bm\nu \in [\bm \mu]}\ketbra{\bm\eta }{\bm\nu}$ is positive semidefinite since it is simply the density matrix of the unnormalized pure state $\sum_{\bm \eta\in [\bm \mu]}\ket{\bm\eta }$. \qedhere
\end{proof}
\Cref{thm:no-k-res-temp-avg-symm} will play an integral part in our bound on circuit complexity, as it will allow us to bound the temporal average of the fidelity between a time-evolved state and an arbitrary target state.

\section{Energy dispersion}\label{sec:03}
In this section, we introduce the notion of energy dispersion of a quantum state. We will be able to prove a bound on the circuit complexity for states that satisfy this property.

\begin{definition}[Effective dimension and energy dispersion]
    The {\it effective dimension}~\cite{Linden2009} of a pure state $\ket{\psi}$ under a Hamiltonian $H$ is defined by the participation ratio $D_\psi^\mathrm{eff}=1/\sum_{\mu=1}^{2^n}\abs{\braket{\psi}{\mu}}^4$. We say that $\ket{\psi}$ is energy dispersed under $H$, with constant $\gamma$, if the effective dimension is exponentially large with coefficient $\gamma$,
    \begin{equation}
        D_\psi^\mathrm{eff}\geq 2^{\gamma n}.
    \end{equation} 
\end{definition}
The effective dimension is a key quantity in studies of quantum equilibration~\cite{Linden2009,Pilatowsky-Cameo2025,huang2024randomproduct}. It is expected that low-entangled states will have an exponentially large effective dimension, and hence they will be energy dispersed. We will prove this expectation for the ensembles of random product states introduced in Sec.~\ref{sec:Ensemblesofproductstates}, following Refs.~\cite{huang2020,huang2024randomproduct}.

\subsection{Proposition 2: Typical product states are energy dispersed}
As recognized by Ref.~\cite{huang2020}, one can very easily show that a random product state, sampled from $\mathcal{E}_0$ [defined in Eq.~\eqref{eq:rand-prod-state}], is energy dispersed, for an arbitrary (even nonlocal) Hamiltonian.
\begin{prop}[Random product states are energy dispersed; Lemma 5 of Ref.~\cite{huang2020}]\label{thm:inf-energy-dispersion}
    Let $H$ be an arbitrary Hamiltonian. Then
    \begin{equation}
        \mathbb{E}_{\psi\sim\mathcal{E}_0}[(D_\psi^\mathrm{eff})^{-1}]\leq \left(\frac{2}{3}\right)^n.
    \end{equation}
   
    Further, for $0<\eta<\log(3/2)$, with probability at least $1-\exp(-\eta n)$, a state $\psi$ sampled from $\mathcal{E}_0$ is energy dispersed with constant $\gamma < (\log(3/2)-\eta)/\log 2$. 
\end{prop}

In what follows, we present a generalization of the above result to the finite-temperature ensembles $\mathcal{E}_{\beta,\delta}$, at high temperature. We follow the idea from Ref.~\cite{huang2024randomproduct}, which proved this in the particular case where the ensembles are discretized to contain only stabilizer product states.

We denote by $\mathcal{V}\leq(2r+1)^D$ the number of vertices inside a ball of radius $r$. Under (P2), we have $\mathcal{V}\geq 2Dr+1$. Let us define
\begin{equation}
\label{eq:betacrit}
    \beta_\mathrm{c}^{-1} = 2^{5D+5+6\mathcal{V}}\mathcal{V}(D+1)! r^{2D}/h_\mathrm{min}^2.
\end{equation}
We note that $\beta_\mathrm{c}$ falls below the inverse temperature threshold required for uniform clustering~\cite{Kliesch2014,Capel2025}, so we will use this property below. We now present a proof that typical product states are energy dispersed for inverse temperatures below $\beta_\mathrm{c}$. 
\begin{prop}[Typical product states are energy dispersed; Proposition 2 in the main text]
   \label{thm:energy-dispersion}
    Let $H$ be a Hamiltonian satisfying (P1)-(P3). For any inverse temperature $|\beta|< \beta_\mathrm{c}$, with probability larger than $1-c\exp(-\eta n)$ for $c=4\,(4^{\mathcal{V}+1}e\mathcal{V}/h_\mathrm{min})^{4\mathcal{V}}$ and $0<\eta<\log(3/2)(1-|\beta|/\beta_\mathrm{c})$, a sampled state $\psi\sim \mathcal{E}_{\beta,\delta}$ [defined in Eq.~\eqref{eq:mc-prod-state}] is energy dispersed with $\gamma < (\log(3/2)(1-|\beta|/\beta_\mathrm{c})-\eta)/\log 2$ for any $\delta\geq 2^{2\mathcal{V}+1}$.
\end{prop}

We will prove \cref{thm:energy-dispersion} via a sequence of lemmas below. For this, let us introduce the following notation. For a given vertex $v\in\Lambda$, we define $H_v$ as the sum of all Pauli terms of $H$ whose supports are contained within a ball of radius $r$ centered at $v$. By extensivity, $H_v$ has at least one Pauli term with an infinity norm lower bounded by $h_\mathrm{min}$. Additionally, we denote $\mathcal{E}_\mathrm{Haar}$ as the ensemble of single qubit Haar random states. Finally, we introduce the following temperature scale for the thermal energy: 
\begin{equation}\label{eq:betath}
    \beta_\mathrm{th}^{-1}=\beta_\mathrm{c}^{-1} h_\mathrm{min}/2^{2\mathcal{V}}.
\end{equation}

We first need the following lemma, showing that the thermal energy at high temperature is upper bounded by an extensive quantity.
\begin{lemma}[Upper bound on thermal energy] \label{thm:therm-energy-bound}
    For $\abs{\beta}<\beta_\mathrm{c}$, $E(\beta)$ is bounded by
    \begin{equation}
        |E(\beta)| \leq n h_\mathrm{min}(|\beta|/4\mathcal{V}\beta_\mathrm{th}).
    \end{equation}
\end{lemma}
\begin{proof}
    We follow the argument in the proof of Theorem 5 of Ref.~\cite{huang2024randomproduct}, by leveraging the exponentially decaying correlation of the Gibbs state $g_\beta$, which is guaranteed for $|\beta|<\beta_\mathrm{c}$ due to Ref.~\cite{Kliesch2014}. In what follows, we take $\beta>0$, as the bound for $\beta<0$ directly follows by an identical argument.

    First of all, we write the derivative of $E(\beta)$ in terms of the variance of the energy:
    \begin{equation}
        \frac{dE(\beta)}{d\beta}=-\operatorname{Var}_{g_\beta}(H)\leq 0.
    \end{equation}
    From this, $E(\beta)$ is not positive and can be naturally bounded by
    \begin{equation}
        E(\beta) = -\int_0^\beta dx \operatorname{Var}_{g_x}(H) \geq -\beta \max_{0\leq x\leq\beta}\left[\operatorname{Var}_{g_x}(H)\right].
    \end{equation}
    Thus, we focus on bounding the variance. To this end, let $\tilde{H}_i$ be obtained from $H_i$ by distributing each Pauli term equally among all balls with radius $r$ that contain its support, so that $H=\sum_i\tilde{H}_i$. Then, we have 
    \begin{equation}
        \mathrm{Var}_{g_x}(H)\leq \sum_{i,j}|\mathrm{Cov}_{g_x}(\tilde{H}_i,\tilde{H}_j)|.
    \end{equation}
    Due to Ref.~\cite{Kliesch2014}, $|\mathrm{Cov}_{g_x}(\tilde{H}_i,\tilde{H}_j)|$ is exponentially decaying in $d(B_i,B_j)$:
    \begin{equation}
        |\operatorname{Cov}_{g_\beta}(\tilde{H}_i,\tilde{H}_j)|\leq \frac{4a\|\tilde{H}_i\|_\infty\|\tilde{H}_j\|_\infty}{\log(3)(1-e^{-1/\xi(\beta)})}e^{-d(B_i,B_j)/\xi(\beta)}
    \end{equation}
    for all $d(B_i,B_j)\geq L_0(\beta,a)$, where $B_i$ is the ball of radius $r$ centered at $i$, and $\xi(\beta)$ is the correlation length given by
    \begin{equation}
        \xi(\beta)^{-1} = \abs{\log(\alpha e^{2|\beta|h_\mathrm{max}}(e^{2|\beta|h_\mathrm{max}}-1))}.
    \end{equation}
    Here, $a$ is the number of boundary edges between $\mathrm{supp}(B_i)$ and its complement, $\alpha$ is the growth constant of the lattice, and $L_0(\beta,a)$ is the microscopic length scale given by
    \begin{equation}
        L_0(\beta,a)=\xi(\beta)\abs{\log(\log(3)(1-e^{-1/\xi(\beta)})/a)}.
    \end{equation}
    Let $h_\mathrm{max}$ be the maximum value of $\|\tilde{H}_i\|_\infty$ over all lattice sites. For microscopic fluctuations in $d(B_i,B_j)< L_0(\beta,a)$, the covariance can be upper bounded by
    \begin{equation}
        |\operatorname{Cov}_{g_\beta}(\tilde{H}_i,\tilde{H}_j)|\leq \abs{\tr(g_\beta \tilde{H}_i \tilde{H}_j)} + \abs{\tr(g_\beta \tilde{H}_i)}\abs{\tr(g_\beta \tilde{H}_j)} \leq 2\|\tilde{H}_i\|_\infty\|\tilde{H}_j\|_\infty\leq 2h_\mathrm{max}^2
    \end{equation}
    using H\"older's inequality. Before we proceed further, we introduce 
    \begin{equation}
        L_0^\mathrm{max}=\max_{0\leq x\leq \beta}L_0(x,a).
    \end{equation}
    Additionally, we denote $S_l$ as the surface of a hypercube of side length $l$. Then, the variance can be upper bounded by
    \begin{equation}
        \begin{split}
            \frac{1}{n}\operatorname{Var}_{g_x}(H) 
            &\leq \frac{1}{n}\sum_{i,j}\abs{\operatorname{Cov}_{g_x}(\tilde{H}_i,\tilde{H}_j)}\\
            &\leq \sum_{d(B_i,B_v)< L_0^\mathrm{max}} 2h_\mathrm{max}^2 + \sum_{d(B_i,B_v)\geq L_0^\mathrm{max}}\frac{4ah_\mathrm{max}^2}{\log(3)(1-e^{-1/\xi(x)})}e^{-d(B_i,B_v)/\xi(x)}\\
            &\leq 2h_\mathrm{max}^2(2L_0^\mathrm{max}+4r)^D + \sum_{l\geq L_0^\mathrm{max}}|S_{2l+2r+1}|\frac{4ah_\mathrm{max}^2}{\log(3)(1-e^{-1/\xi(x)})}e^{-l/\xi(x)}
        \end{split}
    \end{equation}
    for some fixed vertex $v\in\Lambda$. Since $|S_{2l+2r+1}|$ is upper bounded by
    \begin{equation}
        |S_{2l+2r+1}| = (2l+2r+1)^D-(2l+2r-1)^D = 2\sum_{q=0}^{D-1}(2l+2r+1)^q (2l+2r-1)^{D-1-q} \leq 2D(2l+2r+1)^{D-1},
    \end{equation}
    we have
    \begin{equation}
        \frac{1}{n}\operatorname{Var}_{g_x}(H) \leq 2^{D+1}h_\mathrm{max}^2(L_0^\mathrm{max}+2r)^D + \frac{8ah_\mathrm{max}^2D}{\log(3)(1-e^{-1/\xi(x)})}\sum_{l\geq L_0^\mathrm{max}} (2l+2r+1)^{D-1}e^{-l/\xi(x)}.
    \end{equation}
    We can further simplify this by bounding the summation as
    \begin{equation}
        \begin{split}
            \sum_{l\geq L_0^\mathrm{max}} (2l+2r+1)^{D-1}e^{-l/\xi(x)}
            &\leq 2^{D-1}(1+r)^{D-1}\sum_{l=0}^\infty(1+l)^{D-1}e^{-l/\xi(x)} \\
            &= 2^{D-1}(1+r)^{D-1}\frac{A_{D-1}(e^{-1/\xi(x)})}{(1-e^{-1/\xi(x)})^D}\\
            &\leq \frac{2^{D-1}(1+r)^{D-1}(D-1)!}{(1-e^{-1/\xi(x)})^D}
        \end{split}
    \end{equation}
    with the Eulerian polynomials $A_n$. Here, we use an identity due to Euler, $\sum_{j=0}^\infty x^j(1+j)^n=A_n(x)(1-x)^{-(n+1)}$, and the fact that $A_{D-1}(e^{-1/\xi(x)})\leq(D-1)!$. This gives
    \begin{equation}
        \frac{1}{n}\operatorname{Var}_{g_x}(H) \leq 2^{D+1}h_\mathrm{max}^2(L_0^\mathrm{max}+2r)^D + \frac{2^{D+2}(1+r)^{D-1}ah_\mathrm{max}^2 D!}{\log(3)(1-e^{-1/\xi(x)})^{D+1}}.
    \end{equation}
    For a local Hamiltonian on a $D$-dimensional hypercubic lattice, the growth constant is upper bounded by $\alpha\leq 2De$~\cite{Kliesch2014}. From this, we can bound the maximum correlation length as    
    \begin{equation}
        \xi_\mathrm{max}^{-1} 
        = \min_{0\leq x\leq\beta}[\xi(x)^{-1}] = \abs{\log(\alpha e^{2|\beta|h_\mathrm{max}}(e^{2|\beta|h_\mathrm{max}}-1))}\geq -\log(2De\times e^{2|\beta|h_\mathrm{max}}(e^{2|\beta|h_\mathrm{max}}-1))\geq \log 2
    \end{equation}
    for the range of $\beta$ we consider. Additionally, since $a\geq 2$ for all $D\geq 1$, we can also bound $L_0^\mathrm{max}$ as
    \begin{equation}
        L_0^\mathrm{max}=\max_{0\leq x\leq \beta}\left[\xi(x)\abs{\log\left(\frac{a}{\log(3) (1-e^{-1/\xi(x)})}\right)}\right]\leq\xi_\mathrm{max}\log\left(\frac{a}{\log(3) (1-e^{-1/\xi_\mathrm{max}})}\right)\leq \log_2(2a).
    \end{equation}
     Since $a$ is the number of boundary edges of a ball of radius $r$, we have $a\leq 2D(2r+1)^{D-1}$. Altogether, we finally get the stated bound on $E(\beta)$:
    \begin{equation}
        E(\beta)\geq -2 n\beta h_\mathrm{max}^2 \left[ 2^D (\log_2(2a)+2r)^D + 2^{2D+2}(r+1)^{D-1} a D! \right] \geq -n\beta h_\mathrm{max}^2 \times 2^{5D+3}Dr^{2D}D!.
    \end{equation}    
    The last inequality comes from the fact that
    \begin{equation}
        \begin{split}
            (\log_2(2a)+2r)^D 
            &\leq ( 2 + \log_2 D + (D-1)\log_2 (2r+1) + 2r )^D \\
            &\leq ( D+1 + \log_2 D + (D+1)r )^D \\
            &\leq (1+1/D)^D D^{D+1} (1+r)^D\\
            &\leq e D^{D+1} (1+r)^D\\
            &\leq e^{D+1} D\cdot D! (1+r)^D\\
            &\leq 2^{2D+2} D\cdot D! (1+r)^D.
        \end{split}
    \end{equation}
    This together with the fact that $h_\mathrm{max}\leq 4^\mathcal{V}$ gives the stated bound.
\end{proof}

For the lemmas below, we introduce a set $R$ of vertices such that the supports of $\{H_i\}_{i\in R}$ have distances at least $r+1$. The support of each Pauli term in the Hamiltonian intersects with at most one vertex in $R$. We set the number $m$ of vertices in $R$ as 
    \begin{align}\label{eq:m-def}
            m &= \lceil m_* \rceil &\text{with}&&
            m_* &= \frac{|E(\beta)|}{h_\mathrm{min}-w}.
    \end{align}
    Here, $w>0$ is defined as $w=h_\mathrm{min}/x_*$, where $x_*$ is the greater solution of 
    \begin{equation}\label{eq:energy-width-def}
        xe^{-x}=\frac{h_\mathrm{min}}{4^{\mathcal{V}+1}e\mathcal{V}}.
    \end{equation}

    We briefly justify that this choice of $m$ is well-defined and consistent with (P1)-(P3). First, $m_*$ is non-negative since we have $h_\mathrm{min}/w>2$ given that $\mathcal{V}\geq 2Dr+1$ and $h_\mathrm{min}\leq 1$. Additionally, the definitions of $\beta_\mathrm{c}$ and $\beta_\mathrm{th}$ in Eqs.~\eqref{eq:betacrit} and~\eqref{eq:betath} give 
    \begin{equation}
        \beta_\mathrm{th}/\beta_\mathrm{c}=2^{2\mathcal{V}}/h_\mathrm{min}\geq 2^{2\mathcal{V}}\geq (2+3r)^D.
    \end{equation}
    Then, due to the upper bound on the thermal energy established in Lemma~\ref{thm:therm-energy-bound}, for $|\beta|<\beta_\mathrm{c}$, we have $|\beta|<\beta_\mathrm{th}/(2+3r)^D$ and
    \begin{equation}
        m_* \leq \frac{2|E(\beta)|}{h_\mathrm{min}} \leq n(|\beta|/\beta_\mathrm{th}) \leq n/(2+3r)^D.
    \end{equation}
    Since each Pauli term of $H$ has diameter at most $r$, a sphere-packing of $\{\mathrm{supp}(H_i)\}_{i\in R}$ implies that the maximum achievable value of $m$ is strictly greater than $n/(2+3r)^D$. Therefore, we can safely choose $m=\lceil m_*\rceil$.

Using the set $R$, we construct an ensemble $\mathcal{E}_\mathcal{R}$ of product states in $\mathcal{R}=\bigcup_{i\in R}\mathrm{supp}(H_i)$,
    \begin{equation}
        \label{eq:ERdef}\mathcal{E}_\mathcal{R}\coloneqq\{\psi \in\mathcal{E}_\mathrm{Haar}^{\otimes |\mathcal{R}|} : \abs{\expval{H_\mathcal{R}}{\psi} - E(\beta)} < \frac{\zeta}{2}\sqrt{n}+2h_\mathrm{min}\}
    \end{equation}
    with the measure inherited from $\mathcal{E}_\mathrm{Haar}^{\otimes |\mathcal{R}|}$, 
    with $H_\mathcal{R}=\sum_{i\in R}H_i$, and where we will select $\zeta\in [\sqrt{8/3}\, 4^{\mathcal{V}},\delta-2]$ below (this interval is nonempty whenever $\delta\geq 2^{2\mathcal{V}+1}$).

Our next lemma lower bounds the number of states in the ensemble $\mathcal{E}_\mathcal{R}$ at high temperatures.
\begin{lemma}\label{thm:E_R-prob-bound}
    For $|\beta|< \beta_\mathrm{c}$, and the ensemble $\mathcal{E}_\mathcal{R}$ defined in Eq.~\eqref{eq:ERdef}, we have
    \begin{equation}\label{eq:E_R-prob-bound}
        \operatornamewithlimits{Pr}_{\psi\sim\mathcal{E}_\mathrm{Haar}^{\otimes n}}\left[\psi_\mathcal{R}\in\mathcal{E}_\mathcal{R}\right] > \left(1-\frac{4h_\mathrm{min}^2}{\zeta^2}\right)\left(\frac{h_\mathrm{min}}{4^{\mathcal{V}+1}e\mathcal{V}}\right)^{(|\beta|/\beta_\mathrm{th})n+4\mathcal{V}}.
    \end{equation}
\end{lemma}
\begin{proof}
    First of all, if we have $m=0$ in Eq.~\eqref{eq:m-def}, \textit{i.e.}, $E(\beta)=0$, then $H_\mathcal{R}=0$, and $\mathcal{E}_\mathcal{R}$ is identical to $\mathcal{E}_\mathrm{Haar}^{\otimes |\mathcal{R}|}$. Thus, in this case, the stated inequality holds.
    
    Next, let us consider the case when $m>0$. For $w$ defined in Eq.~\eqref{eq:energy-width-def}, we construct two intervals: $I_w^-=[E(\beta)/m-w,E(\beta)/m]$ and $I_w^+=[E(\beta)/m,E(\beta)/m+w]$. For each vertex $i\in R$, let $\mathcal{E}_{i,w}^{\pm}$ be the ensembles of product states in $R_i\coloneq\mathrm{supp}(H_i)$ having $\langle H_i\rangle\in I_w^\pm$ and induced from $\mathcal{E}_\mathrm{Haar}^{\otimes |R_i|}$. These ensembles are nonempty as we have $|E(\beta)/m\pm w|<h_\mathrm{min}$ due to the definition of $m$ in Eq.~\eqref{eq:m-def} and $h_\mathrm{min}/w>2$. The ensemble averages of $\langle H_i\rangle$ over $\mathcal{E}_{i,w}^+$ or $\mathcal{E}_{i,w}^-$ are also in $I_w^+$ or $I_w^-$, respectively. Since changing the sign of each ensemble can only shift the mean energy up to $2w$, there exists an ensemble of states $\mathcal{E}_{R,w}=\bigotimes_{i\in R}\mathcal{E}_{i,w}$ with $\mathcal{E}_{i,w}=\mathcal{E}_{i,w}^+$ or $\mathcal{E}_{i,w}^-$ whose mean energy is in $[E(\beta)-2w,E(\beta)+2w]$. By construction, the energy fluctuation of each $\mathcal{E}_{i,w}$ is upper bounded by
    \begin{equation}
        \operatornamewithlimits{Var}_{\psi\sim\mathcal{E}_{i,w}}(\expval{H_i}{\psi}) \leq w^2.
    \end{equation}
    Additionally, since $\mathcal{E}_{R,w}$ is a product state ensemble, and since for distinct $i,j\in R$, $H_i$ and $H_j$ act on disjoint regions, the random variables $\expval{H_i}{\psi}$ and $\expval{H_j}{\psi}$ are independent. Thus, the total variance is additive:
    \begin{equation}
        \operatornamewithlimits{Var}_{\psi\sim\mathcal{E}_{R,w}}(\expval{H_\mathcal{R}}{\psi}) \leq m w^2.
    \end{equation}    
    Then, due to Markov's inequality, we have
    \begin{equation}\label{eq:eps-ball-Markov}
        \operatornamewithlimits{Pr}_{\psi\sim\mathcal{E}_{R,w}}\left[\abs{\expval{H_\mathcal{R}}{\psi} - E(\beta)}\geq \frac{\zeta}{2}\sqrt{n}+2w\right]\leq \frac{4w^2 m}{\zeta^2 n}.
    \end{equation}
    
    We note that the probability in Eq.~\eqref{eq:E_R-prob-bound} can be obtained by estimating the volume of each $\mathcal{E}_{i,w}$ relative to the volume of the product states in $\mathrm{supp}(H_i)$. We therefore bound this volume below. To this end, we expand $H_i$ into a sum of Pauli terms:
    \begin{equation}
        H_i = \sum_j c_j P_{\alpha_j}\qquad\text{with}\qquad \alpha_j\in\{I,X,Y,Z\}^{|R_i|}.
    \end{equation}    
    Let $\ket{\phi}$ be a product state in $R_i$, $\bigotimes_{j\in R_i}\ket{\phi_j}$. We can parametrize each single qubit state $\ket{\phi_j}$ by a Bloch vector $\vec{r}_j=[1,r_X^{(j)},r_Y^{(j)},r_Z^{(j)}]$. Furthermore, let $\{\vec{r}_{*,j}\}$ be the set of Bloch vectors that corresponds to a product state having $\langle H_i\rangle=E(\beta)/m\pm w/2$, with the sign chosen according to whether $\mathcal{E}_{i,w}=\mathcal{E}_{i,w}^+$ or $\mathcal{E}_{i,w}^-$. We introduce $f(\{\vec{r}_j\})=\expval{H_i}{\phi}$. This can be expanded as
    \begin{equation}
        f(\{\vec{r}_j\}) = \sum_j c_j \prod_{l\in\mathrm{supp}(P_{\alpha_j})} (\vec{r}_l)_{(\alpha_j)_l}.
    \end{equation}
    We can then construct an inequality for the Lipschitz continuity of $f$ as follows: 
    \begin{equation}
        \begin{split}
            |f(\{\vec{r}_i\})-f(\{\vec{r}_{*,i}\})|
            &\leq \sum_j \abs{c_j} \abs{\prod_{l\in\mathrm{supp}(P_{\alpha_j})}(\vec{r}_l)_{(\alpha_j)_l}-\prod_{l\in\mathrm{supp}(P_{\alpha_j})}(\vec{r}_{*,l})_{(\alpha_j)_l}}\\
            &\leq \sum_j\abs{c_j}\sum_{l\in\mathrm{supp}(P_{\alpha_j})}\abs{(\vec{r}_l)_{(\alpha_j)_l}-(\vec{r}_{*,l})_{(\alpha_j)_l}}\\
            &= \sum_j\abs{c_j}\sum_{l\in\mathrm{supp}(P_{\alpha_j})}\abs{\hat{e}_{(\alpha_j)_l}\cdot(\vec{r}_l-\vec{r}_{*,l})}\\
            &\leq  \sum_j\abs{c_j}\sum_{l\in\mathrm{supp}(P_{\alpha_j})}\|\vec{r}_l-\vec{r}_{*,l}\|_2\\
            &\leq 3\cdot 4^{\mathcal{V}-1}\mathcal{V} r_\mathrm{max}
        \end{split}
    \end{equation}
    with $r_\mathrm{max}=\max_{l\in R_i}\|\vec{r}_l-\vec{r}_{*,l}\|_2$. 
    
    Now, for each $l\in R_i$, let $\theta_l$ denote the geodesic angle between $\vec{r}_l$ and $\vec{r}_{*,l}$. Then, we get $r_\mathrm{max}=2\sin(\theta_l/2)$. Thus, if for every $l\in R_i$, $\vec{r}_l$ lies within the spherical cap on the Bloch sphere with pole $\vec{r}_{*,l}$ and opening angle 
    \begin{equation}
        \theta_l\leq\theta_\mathrm{max}\coloneq2\sin^{-1}\left( \frac{w}{3\cdot 4^{\mathcal{V}}\mathcal{V}} \right),
    \end{equation}
    then $r_\mathrm{max}\leq w/(6\cdot 4^{\mathcal{V}-1}\mathcal{V})$. This gives $|f(\{\vec{r}_i\})-f(\{\vec{r}_{*,i}\})|\leq w/2$, and hence $\phi\in\mathcal{E}_{i,w}$. The normalized area of such a spherical cap is 
    \begin{equation}
        \frac{\mathrm{Area}}{4\pi} = \sin^2(\theta_\mathrm{max}/2) = \sin^2\left(\sin^{-1}\left(\frac{w}{3\cdot 4^{\mathcal{V}}\mathcal{V}}\right)\right) = \left(\frac{w}{3\cdot4^{\mathcal{V}}\mathcal{V}}\right)^2. 
    \end{equation}
    Since $|\mathrm{supp}(H_i)|$ is at most $\mathcal{V}$, the probability of having $\langle H_i\rangle$ in either $I_w^+$ or $I_w^-$ is lower bounded by
    \begin{equation}
        \operatornamewithlimits{Pr}_{\psi\sim\mathcal{E}_\mathrm{Haar}^{\otimes |R_i|}}\left[\expval{H_i}{\psi}\in I_w^\pm\right]\geq \left(\frac{w}{3\cdot4^{\mathcal{V}}\mathcal{V}}\right)^{2\mathcal{V}}\geq \left(\frac{w}{4^{\mathcal{V}+1}\mathcal{V}}\right)^{2\mathcal{V}}.
    \end{equation}
    This together with Eq.~\eqref{eq:eps-ball-Markov} gives
    \begin{equation}
        \begin{split}
            &\operatornamewithlimits{Pr}_{\psi\sim\mathcal{E}_\mathrm{Haar}^{\otimes n}}\left[\abs{\expval{H_\mathcal{R}}{\psi} - E(\beta)} < \frac{\zeta}{2}\sqrt{n}+2w \right]\\
            &\geq \operatornamewithlimits{Pr}_{\psi\sim\mathcal{E}_\mathrm{Haar}^{\otimes n}}\left[\psi_\mathcal{R}\in \mathcal{E}_{R,w}\right]\operatornamewithlimits{Pr}_{\psi\sim\mathcal{E}_\mathrm{Haar}^{\otimes n}}\left[\abs{\expval{H_\mathcal{R}}{\psi} - E(\beta)}< \frac{\zeta}{2}\sqrt{n}+2w \,\mid\, \psi_\mathcal{R}\in\mathcal{E}_{R,w}\right]\\
            &>\left(1-\frac{4w^2 m}{\zeta^2 n}\right) \left(\frac{w}{4^{\mathcal{V}+1}\mathcal{V}}\right)^{2m\mathcal{V}}.
        \end{split}
    \end{equation}
    From Eq.~\eqref{eq:energy-width-def}, $\mathcal{V}\geq 2D+1$, and $h_\mathrm{min}\leq 1$, we have $2<h_\mathrm{min}/w\leq 2\log(4^{\mathcal{V}+1}e\mathcal{V}/h_\mathrm{min})$. These together with the definitions of $m$ and $m_*$ in Eq.~\eqref{eq:m-def} give 
    \begin{equation}
        \begin{split}
            \operatornamewithlimits{Pr}_{\psi\sim\mathcal{E}_\mathrm{Haar}^{\otimes n}}\left[\psi_\mathcal{R}\in\mathcal{E}_\mathcal{R}\right] 
            &> \left(1-\frac{4w^2 m}{\zeta^2 n}\right) \left(\frac{w}{4^{\mathcal{V}+1}\mathcal{V}}\right)^{2\mathcal{V}(m_*+1)}\\
            &\geq \left(1-\frac{4h_\mathrm{min}^2}{\zeta^2}\right)\left(\frac{h_\mathrm{min}}{4^{\mathcal{V}+1}e\mathcal{V}}\right)^{4\mathcal{V}|E(\beta)|/h_\mathrm{min}}e^{2\mathcal{V}(1-h_\mathrm{min}/w)}\\
            &\geq \left(1-\frac{4h_\mathrm{min}^2}{\zeta^2}\right)\left(\frac{h_\mathrm{min}}{4^{\mathcal{V}+1}e\mathcal{V}}\right)^{4\mathcal{V}(|E(\beta)|/h_\mathrm{min}+1)}\\
            &\geq \left(1-\frac{4h_\mathrm{min}^2}{\zeta^2}\right)\left(\frac{h_\mathrm{min}}{4^{\mathcal{V}+1}e\mathcal{V}}\right)^{(|\beta|/\beta_\mathrm{th})n+4\mathcal{V}},
        \end{split}
    \end{equation}
    where the last inequality comes from Lemma~\ref{thm:therm-energy-bound}.
\end{proof}

Finally, we present our last lemma, upper bounding the average inverse effective dimension at high temperature. 
We note that \cref{thm:energy-dispersion} follows directly from Lemma~\ref{thm:exp-supp-IPR} and Markov's inequality.
    
\begin{lemma}[Exponential suppression of average inverse effective dimension]\label{thm:exp-supp-IPR}
    Let $H$ be a Hamiltonian satisfying (P1)-(P3). For a pure state $\ket{\psi}$, let $D^\mathrm{eff}_\psi$ be the effective dimension of $\ket{\psi}$ with respect to $H$. Then, it follows that 
    \begin{equation}
        \mathbb{E}_{\psi\sim\mathcal{E}_{\beta,\delta}}\left[(D^\mathrm{eff}_\psi)^{-1}\right]<4\cdot\left(\frac{4^{\mathcal{V}+1}e\mathcal{V}}{h_\mathrm{min}}\right)^{4\mathcal{V}}\left(\frac{2}{3}\right)^{(1-|\beta|/\beta_\mathrm{c})n}   
    \end{equation}
    for any $|\beta|< \beta_\mathrm{c}$ with $ \delta \geq 2^{2\mathcal{V}+1}$.
\end{lemma}
\begin{proof}
    We can derive the stated upper bound on the averaged $1/D_\psi^\mathrm{eff}$ using the set $R$. To this end, let us first upper bound it in terms of the probability of having $\psi\in\mathcal{E}_{\beta,\delta}$ when $\psi$ is sampled from $\mathcal{E}_\mathrm{Haar}^{\otimes n}$. By the definition of $\mathcal{E}_{\beta,\delta}$, $\mathbb{E}_{\psi\sim\mathcal{E}_{\beta,\delta}}[1/D^\mathrm{eff}_\psi]$ is the conditional expectation value of $1/D^\mathrm{eff}_\psi$ given $\psi\in\mathcal{E}_{\beta,\delta}$, which can be written as
    \begin{equation}
        \mathbb{E}_{\psi\sim\mathcal{E}_{\beta,\delta}}\left[\frac{1}{D^\mathrm{eff}_\psi}\right] = \frac{\mathbb{E}_{\psi\sim\mathcal{E}_\mathrm{Haar}^{\otimes n}} \left[\mathbbm{1}_{\psi\in\mathcal{E}_{\beta,\delta}}/{D^\mathrm{eff}_\psi}\right] }{ \operatornamewithlimits{Pr}_{\psi\sim\mathcal{E}_\mathrm{Haar}^{\otimes n}}\left[\psi\in\mathcal{E}_{\beta,\delta}\right]},
    \end{equation}
    where $\mathbbm{1}_{\psi\in\mathcal{E}_{\beta,\delta}}$ is the indicator function that gives one if $\psi\in\mathcal{E}_{\beta,\delta}$ and zero otherwise. Since $\mathbbm{1}_{\psi\in\mathcal{E}_{\beta,\delta}}\leq 1$ for any state $\psi$, we can upper bound the average of $1/D^\mathrm{eff}_\psi$ over $\mathcal{E}_{\beta,\delta}$ as
    \begin{equation}\label{eq:ave-eff-dim-inter-bound}
        \mathbb{E}_{\psi\sim\mathcal{E}_{\beta,\delta}}\left[\frac{1}{D^\mathrm{eff}_\psi} \right]
        \leq \frac{\mathbb{E}_{\psi\sim\mathcal{E}_\mathrm{Haar}^{\otimes n}} \left[{1}/{D^\mathrm{eff}_\psi}\right] }{ \operatornamewithlimits{Pr}_{\psi\sim\mathcal{E}_\mathrm{Haar}^{\otimes n}}\left[\psi\in\mathcal{E}_{\beta,\delta}\right]} \leq \frac{\left({2}/{3}\right)^n }{ \operatornamewithlimits{Pr}_{\psi\sim\mathcal{E}_\mathrm{Haar}^{\otimes n}}\left[\psi\in\mathcal{E}_{\beta,\delta}\right]},
    \end{equation}
    where the last inequality comes from \Cref{thm:inf-energy-dispersion}.
    
    Next, let us lower bound the probability of having $\psi\in\mathcal{E}_{\beta,\delta}$. Then, since we have $h_\mathrm{min}\leq 1$, $n\geq 1$, and $\zeta \leq \delta-2$ [defined below Eq.~\eqref{eq:ERdef}], we can lower bound the probability using the triangle inequality as
    \begin{equation}
        \operatornamewithlimits{Pr}_{\psi\sim\mathcal{E}_{\mathrm{Haar}}^{\otimes n}}[\psi\in\mathcal{E}_{\beta,\delta}]\geq \operatornamewithlimits{Pr}_{\psi\sim\mathcal{E}_{\mathrm{Haar}}^{\otimes n}}\left[\psi_\mathcal{R}\in\mathcal{E}_\mathcal{R} \text{ and } |\expval{H_{\mathcal{R}^c}}{\psi}|<\zeta\sqrt{n}/2\right]
    \end{equation}
    with $H_{\mathcal{R}^c}=H-H_\mathcal{R}$. By the definition of conditional probability, this becomes:
    \begin{equation}\label{eq:prod-of-R-Haar}
        \operatornamewithlimits{Pr}_{\psi\sim\mathcal{E}_{\mathrm{Haar}}^{\otimes n}}[\psi\in\mathcal{E}_{\beta,\delta}]\geq\operatornamewithlimits{Pr}_{\psi\sim\mathcal{E}_{\mathrm{Haar}}^{\otimes n}}\left[\psi_\mathcal{R}\in\mathcal{E}_\mathcal{R}\right] \operatornamewithlimits{Pr}_{\psi\sim\mathcal{E}_{\mathrm{Haar}}^{\otimes n}}\left[ \abs{\expval{H_{\mathcal{R}^c}}{\psi}}<\zeta\sqrt{n}/2 \,\mid\, \psi_{\mathcal{R}}\in \mathcal{E}_\mathcal{R}\right].
    \end{equation}
    Here, $\psi_\mathcal{R}$ is the reduced density matrix of $\psi$ in $\mathcal{R}=\bigcup_{i\in R}\mathrm{supp}(H_i)$, which is a product state. The first term on the right-hand side is bounded by Lemma~\ref{thm:E_R-prob-bound}, so here we focus on bounding the second term on the right-hand side. 
    Because $\psi_{\mathcal{R}}$ is independent of $\psi_{\mathcal{R}^c}$ over the full product Haar ensemble, the latter term is equal to
    \begin{equation}
        \begin{split}
            \operatornamewithlimits{Pr}_{\psi_{\mathcal{R}}\sim \mathcal{E}_\mathcal{R},\psi_{\mathcal{R}^c}\sim\mathcal{E}_{\mathrm{Haar}}^{\otimes|\mathcal{R}^c|}}\left[ \abs{\expval{H_{\mathcal{R}^c}}{\psi}}<\zeta\sqrt{n}/2\right]
            &=\mathbb{E}_{\psi_\mathcal{R}\sim \mathcal{E}_\mathcal{R}} \operatornamewithlimits{Pr}_{\psi_{\mathcal{R}^c}\sim\mathcal{E}_{\mathrm{Haar}}^{\otimes |\mathcal{R}^c|}}\left[ \abs{\expval{H_{\mathcal{R}^c}}{\psi}}<\zeta\sqrt{n}/2\right]\\
            &=\mathbb{E}_{\phi\sim \mathcal{E}_\mathcal{R}} \operatornamewithlimits{Pr}_{\psi\sim\mathcal{E}_{\mathrm{Haar}}^{\otimes n}}\left[ \abs{\expval{H_{\mathcal{R}^c}}{\psi}}<\zeta\sqrt{n}/2 \,\mid\, \psi_\mathcal{R}=\phi\right].    
        \end{split}
    \end{equation}
    Since every Pauli term in $H_{\mathcal{R}^c}$ has nontrivial support on $\mathcal{R}^c$, using Markov's inequality and the fact that $\tr(H_{\mathcal{R}^c})=0$, the above is bounded by
    \begin{equation}
        > 1-
        \frac{4}{\zeta^2 n }\mathbb{E}_{\phi\sim \mathcal{E}_\mathcal{R}} \operatornamewithlimits{Var}_{\psi\sim \mathcal{E}_\mathrm{Haar}^{\otimes n}}\left[ \expval{H_{\mathcal{R}^c}}{\psi} \,\mid\, \psi_{\mathcal{R}}=\phi\right].
    \end{equation}
    Now, let us upper bound the conditional variance on the right-hand side. To this end, we introduce the projection map $\mathrm{proj}_A(\cdot)$ which replaces all Pauli matrices of a Pauli string $P$ outside the region $A$ with identity matrices. We then define the set $S_{\mathcal{R}^c}$ of Pauli strings appearing in $H_{\mathcal{R}^c}$ projected onto $\mathcal{R}^c$: $S_{\mathcal{R}^c}=\{\mathrm{proj}_{\mathcal{R}^c}(P)\,|\,P\in H_{\mathcal{R}^c}\}$. Using this, we can expand $\expval{H_{\mathcal{R}^c}}{\psi}$ as
    \begin{equation}
        \begin{split}
            \expval{H_{\mathcal{R}^c}}{\psi} 
            &= \sum_{P\in H_{\mathcal{R}^c}} c_P \expval{\mathrm{proj}_\mathcal{R}(P)}{\psi_\mathcal{R}} \expval{\mathrm{proj}_{\mathcal{R}^c}(P)}{\psi_{\mathcal{R}^c}}\\
            &= \sum_{Q\in S_{\mathcal{R}^c}}\expval{Q}{\psi_{\mathcal{R}^c}}\sum_{P\in H_{\mathcal{R}^c},\mathrm{proj}_{\mathcal{R}^c}(P)=Q} c_P \expval{\mathrm{proj}_\mathcal{R}(P)}{\psi_\mathcal{R}}.
        \end{split}
    \end{equation}
    For each Pauli string $Q\in S_{\mathcal{R}^c}$, we choose a vertex $v\in\mathrm{supp}(Q)$. Then, any Pauli string $P$ in the Hamiltonian giving $\mathrm{proj}_{\mathcal{R}^c}(P)=Q$ is supported in $B_r(v)$ by the locality of the Hamiltonian. Therefore, the number $N_Q$ of such Pauli strings is upper bounded by $4^\mathcal{V}$. With $\psi_{\mathcal{R}}=\phi$, this then implies
    \begin{equation}
        \begin{split}
            \operatornamewithlimits{Var}_{\psi\sim \mathcal{E}_\mathrm{Haar}^{\otimes n}}\left[ \expval{H_{\mathcal{R}^c}}{\psi} \,\mid\, \psi_{\mathcal{R}}=\phi\right]
            &\leq \sum_{Q\in S_{\mathcal{R}^c}} 3^{-\abs{\mathrm{supp}(Q)}} \left(\sum_{P\in H_{\mathcal{R}^c},\mathrm{proj}_{\mathcal{R}^c}(P)=Q} c_P \expval{\mathrm{proj}_\mathcal{R}(P)}{\psi_\mathcal{R}}\right)^2\\
            &\leq N_Q \sum_{Q\in S_{\mathcal{R}^c}} 3^{-\abs{\mathrm{supp}(Q)}} \sum_{P\in H_{\mathcal{R}^c},\mathrm{proj}_{\mathcal{R}^c}(P)=Q} c_P^2 \expval{\mathrm{proj}_\mathcal{R}(P)}{\psi_\mathcal{R}}^2\\
            &\leq \frac{1}{3}4^\mathcal{V}\sum_{Q\in S_{\mathcal{R}^c}}\sum_{P\in H_{\mathcal{R}^c},\mathrm{proj}_{\mathcal{R}^c}(P)=Q} c_P^2\\
            &\leq \frac{n}{3}4^{2\mathcal{V}}.
        \end{split}
    \end{equation}     
    Here, we use the moments of Pauli matrices for single-qubit Haar-random states in the first line, the Cauchy-Schwarz inequality in the second line, $\abs{\mathrm{supp}(Q)}\geq 1$ for $Q\in S_{\mathcal{R}^c}$ in the third line, and the boundedness of the Hamiltonian in the last line.
    
    Altogether, the probability of sampling a state in $\mathcal{E}_{\beta,\delta}$ from $\mathcal{E}_\mathrm{Haar}^{\otimes n}$ is bounded by
    \begin{equation}
    \label{eq:ineqq}
        \operatornamewithlimits{Pr}_{\psi\sim\mathcal{E}_\mathrm{Haar}^{\otimes n}}\left[\psi\in\mathcal{E}_{\beta,\delta}\right] 
        > \left(1-\frac{4^{2\mathcal{V}+1}}{3\zeta^2}\right)\left(1-\frac{4h_\mathrm{min}^2}{\zeta^2}\right)\left(\frac{h_\mathrm{min}}{4^{\mathcal{V}+1}e\mathcal{V}}\right)^{(|\beta|/\beta_\mathrm{th})n+4\mathcal{V}},
    \end{equation}
    where the first term on the right-hand side of Eq.~\eqref{eq:prod-of-R-Haar} is lower bounded using Lemma~\ref{thm:E_R-prob-bound}. Thus, we see that since $\zeta \geq \sqrt{8/3}\, 4^{\mathcal{V}}$, $h_\mathrm{min}\leq 1$, and $\mathcal{V}\geq 2D+1$, with the definitions of $\beta_\mathrm{c}$ and $\beta_\mathrm{th}$ in Eqs.~\eqref{eq:betacrit} and~\eqref{eq:betath},
    we can finally get the stated upper bound on the averaged $1/D^\mathrm{eff}_\psi$ using Eq.~\eqref{eq:ave-eff-dim-inter-bound} as
    \begin{equation}
        \mathbb{E}_{\psi\sim\mathcal{E}_{\beta,\delta}}\left[\frac{1}{D^\mathrm{eff}_\psi}\right]<4\cdot\left(\frac{4^{\mathcal{V}+1}e\mathcal{V}}{h_\mathrm{min}}\right)^{4\mathcal{V}}\left(\frac{2}{3}\right)^{(1-|\beta|/\beta_\mathrm{c})n}.\qedhere
    \end{equation}
\end{proof}

\subsection{Temporal moments of fidelity}
The key property of energy dispersed states that we will leverage for our bound on circuit complexity is captured in the following lemma.
\begin{lemma}\label{thm:k-th-fidelity-bound} Let $H$ have spectral ergodicity and $\psi$ be energy dispersed with constant $\gamma$. Then for any subsystem $A$ with $a$ qubits, any state $\sigma$ of $A$, and any integer $k>0$, the $k$-th temporal moment of the fidelity between $\sigma$ and $\psi_A(t)$ can be bounded as
    \begin{equation}
        \mathbb{E}_t[F(\psi_A(t),\sigma)^k]\leq 2^{k((1-\gamma/2)n-a)} \frac{(2^{n-a}+k-1)!}{(2^{n-a}-1)!}.
    \end{equation}
\end{lemma}
\begin{proof}
The fidelity can be upper bounded by the trace between $\psi_A(t)$ and $\sigma$ as
\begin{equation}
    F(\psi_A(t),\sigma)\leq r\tr(\psi_A(t)\sigma) = r\expval{\sigma}{\psi(t)},
\end{equation}
where $r$ is the rank of $\psi_A(t)$, which is upper bounded by $2^{n-a}$. Then, the $k$-th moments of the fidelity can be upper bounded by
\begin{equation}
\label{eq:sumofpermswithsigmadiag}
    \mathbb{E}_{t}[F(\psi_A(t),\sigma)^{k}]\leq 2^{(n-a)k}\mathbb{E}_{t}[\expval{\sigma}{\psi(t)}^k]
    \leq 2^{(n-a)k} \sum_{\pi\in S_k} \tr(\sigma^{\otimes k}\rho_d^{\otimes k} P_\pi),
\end{equation}
where the second inequality is \Cref{thm:no-k-res-temp-avg-symm}, $\rho_d$ is the energy diagonal state corresponding to $\ket{\psi}$ and $P_\pi$ is the unitary operator which permutes the $k$ replicas according to the permutation $\pi$. 
Each permutation of $k$ elements can be decomposed into disjoint cycles. The trace term on the right of Eq.~\eqref{eq:sumofpermswithsigmadiag} only depends on the cycle structure of the permutation $\pi\in S_k$. The decomposition into cycles of any permutation is characterized by a list of numbers $k_1,\dots,k_k\geq0$ with $\sum_{l=1}^k l k_l=k$, where $k_l$ is the number of cycles of length $l$ appearing in the decomposition. One can see that $\tr((\sigma \rho_d)^{\otimes k} P_\pi)=\prod_{l=1}^k \tr((\sigma \rho_d)^l)^{k_l}$. To bound $\tr((\sigma \rho_d)^l)$, let $p_\mathrm{max}$ be the maximum eigenvalue of $\rho_d$. Then, it follows that 
\begin{equation}        
    \tr(\left(\sigma \rho_d\right)^l) = \tr(\left(\sigma^{1/2} \rho_d \sigma^{1/2}\right)^l)
    \leq p_\mathrm{max}^l\tr(\sigma^l)
    \leq p_\mathrm{max}^l2^{n-a},
\end{equation}
where we used the positive semidefiniteness of $\sigma$ in the first equality, Weyl's inequality~\cite[Corollary III.2.3]{Bhatia_1997} in the second inequality, and the fact that $\sigma$ only acts on $A$ in the last inequality. Since $\psi$ is energy dispersed with constant $\gamma$, we have $p_\mathrm{max}^2 \leq \tr(\rho_d^2) \leq 2^{-\gamma n}$.

With the above, we can replace the summation over all permutations in Eq.~\eqref{eq:sumofpermswithsigmadiag} by a sum over all possible cycle structures $(k_1,\dots,k_k)$ of the symmetric group:
\begin{equation}\label{eq:psi-sigma-overlap-bound}
    \begin{split}
        \mathbb{E}_{t}[\expval{\sigma}{\psi(t)}^k] &\leq \sum_{(k_1,\dots,k_k)} \frac{k!}{\prod_{l=1}^k l^{k_l}k_l!}\prod_{l=1}^k \tr(\left(\sigma \rho_d\right)^l)^{k_l}\\
        &\leq \sum_{(k_1,\dots,k_k)} \frac{k!}{\prod_{l=1}^k l^{k_l}k_l!}\prod_{l=1}^k p_\mathrm{max}^{l k_l} 2^{(n-a)k_l}\\
        &= p_\mathrm{max}^k \sum_{\pi\in S_k} 2^{(n-a)c(\pi)} \\
        &\leq 2^{((1-\gamma/2)n-a)k}\times 2^{(a-n)k} \frac{(2^{n-a}+k-1)!}{(2^{n-a}-1)!},
    \end{split}
\end{equation}
where we used the energy dispersion, $p_\mathrm{max}\leq 2^{-\gamma n/2}$, and $c(\pi)$ is the number of cycles of $\pi$. Here, we also used the fact that $\sum_{l=1}^k l k_l=k$ and the following identity for Stirling numbers of the first kind ${k \brack l}$:
\begin{equation}
    \sum_{\pi\in S_k} x^{c(\pi)} = \sum_{l=1}^k {k \brack l} x^l = \frac{(x+k-1)!}{(x-1)!}. 
\end{equation}
This together with Eq.~\eqref{eq:sumofpermswithsigmadiag} gives the stated bound.
\end{proof}

\section{Exponential late-time complexity}
\label{sec:04}

In this section we present the full proofs of Theorems 1 and 3 in the main text, restated here in fully rigorous form.

\subsection{Proof of Theorem 1}
\begin{theorem}[Exponential late-time complexity; Theorem 1 in the main text]\label{th:theoremS1}
    Let $|\beta|< \beta_\mathrm{c}$ [defined in Eq.~\eqref{eq:betacrit}], and define $\gamma=\frac{1}{3}(1-|\beta|/\beta_\mathrm{c})$. For almost all\footnote{\label{nt:1}{\it Almost all} means that the set of coefficient tuples $(c_P)_{P\in\mathcal{P}_r}$ which produce a Hamiltonian $H$ [Eq.~\eqref{eq:paliexpansion}] which violates our theorem has zero measure in $[-1,1]^{|\mathcal{P}_r|}$, where $\mathcal{P}_r$ is the set of Pauli strings whose support has diameter at most $r$.} Hamiltonians $H$ which satisfy properties (P1)-(P3), with probability larger than $1-c \exp(-\tfrac{\gamma}{2} n)$ over the choice of the initial state $\psi\sim \mathcal{E}_{\beta,\delta}$ (or $\mathcal{E}_{0}$ at infinite temperature), we have 
\begin{equation}
\label{eq:theorem2}
\operatornamewithlimits{Pr}_{t}\Big[\mathcal{C}_\varepsilon(\psi_A(t))\leq  \exp(\frac{\gamma}{8}n)\Big]\leq \exp(-\exp(\frac{\gamma}{8} n)),
\end{equation}
for any $n\geq n_*\coloneqq\frac{16}{\gamma}\big(7+\log\tfrac{2}{\gamma(1-\varepsilon)}\big)$ and any subsystem $A$ with size $|A|>(1-\gamma/8)n$, where  $c=4\left(\frac{4^{\mathcal{V}+1}e\mathcal{V}}{h_\mathrm{min}}\right)^{4\mathcal{V}}$, $0<\varepsilon<1-\frac{2}{\gamma}\exp(7-\gamma n/16)$, and $\delta\geq 2^{2\mathcal{V}+1}$.
\end{theorem}
\begin{proof} Due to \Cref{thm:generic-ergodic}, almost all Hamiltonians $H$ which satisfy properties (P1)-(P3) have spectral ergodicity, and due to Theorem~\ref{th:theoremS2} below, if the initial state is energy dispersed with constant $\gamma$, \textit{i.e.}, $D_\psi^\mathrm{eff}\geq 2^{\gamma n}$, we obtain Eq.~\eqref{eq:theorem2} as desired. A state sampled from $\mathcal{E}_{\beta,\delta}$ (or $\mathcal{E}_{0}$) is indeed energy dispersed, with probability larger than $1-c\exp(-\gamma n /2)$, due to \Cref{thm:energy-dispersion} (or \Cref{thm:inf-energy-dispersion}) applied to  $\eta=\gamma/2\ < \gamma \times 3\log(3/2)=\log(3/2)(1-|\beta|/\beta_\mathrm{c})$, where we used $(3\log(3/2)-1/2)/\log 2\approx 1.03>1$ to guarantee that $\gamma<(\log(3/2)(1-|\beta|/\beta_\mathrm{c})-\eta)/\log 2$ as required by \Cref{thm:energy-dispersion}.
\end{proof}
\subsection{Proof of Theorem 3}

\begin{theorem}[Exponential complexity from spectral ergodicity and energy dispersion; Theorem 3 in the main text]
   \label{th:theoremS2}
    Let $H$ be a Hamiltonian having spectral ergodicity, and $\psi$ be energy dispersed with constant $\gamma>0$ under $H$. Then, for any subsystem $A$ with size $|A|> (1-\gamma/8)n$, and any $0<\varepsilon<1-\frac{2}{\gamma}\exp(7-\gamma n/16)$,
    \begin{equation}
        \operatornamewithlimits{Pr}_{t}\Big[\mathcal{C}_\varepsilon(\psi_A(t))\leq  \exp(\frac{\gamma}{8}n)\Big]\leq \exp(-\exp({\frac{\gamma}{8}\,n}))
    \end{equation}
    for system size $n\geq n_*\coloneq \frac{16}{\gamma}\big(7+\log\tfrac{2}{\gamma(1-\varepsilon)}\big)$.
\end{theorem}
\noindent Before providing the proof, we make a few remarks. First, Theorem~\ref{th:theoremS2} does not require any extra assumptions on $\psi$ or $H$. This statement holds even for non-product states and nonlocal Hamiltonians, as long as the conditions of spectral ergodicity and energy dispersion are met. Second, by optimizing the exponents, the lower bound on the subsystem size can be improved to $\abs{A}>(1-\gamma/4)n$. This bound can be further reduced to $\abs{A}>(1-\gamma_\mathrm{max}/2)n$ using the maximum probability $2^{-\gamma_\mathrm{max}n}$ of the energy distribution $\abs{\braket{\mu}{\psi}}^2$ with eigenstates $\{\ket{\mu}\}_\mu$ of $H$. These two points also apply to Lemma~\ref{th:temp-deloc} below. Third, the following proof also has implications for recent circuit complexity lower bounds for the Scrooge ensemble associated with the diagonal state $\rho_d$ [Corollary 7 of Ref.~\cite{mcginley2025}]. Under energy dispersion, the $k$-th moments of the Scrooge ensemble approximate $\rho_d^{\otimes k}\sum_{\pi\in S_k}P_\pi$ in Eq.~\eqref{eq:lemmaineq} up to $O(k^22^{-\gamma_\mathrm{max} n/2})$ relative error [Theorem 1 of Ref.~\cite{mcginley2025}]. Any ensemble having such moments exactly allows one to apply the proof below to bound the exponential constant-error robust complexity of states in the ensemble.
This also works up to constant relative error, as such an error introduces a constant multiplicative factor in the upper bound of Eq.~\eqref{eq:generalinequalitynolog} below.

\begin{proof}
    The following argument holds for an arbitrary choice of an $m$-qubit reference product state $\phi$ with $m\geq n$. Let $S_G=\{\sigma\,|\,\mathcal{C}_0(\sigma|\phi)\leq G\}$ be the set of all states with exact relative circuit complexity at most $G> 0$ with respect to $\phi$ [introduced in \Cref{def:circuit-complexity}]. Let $\lambda=\exp(8-n)$ and define $\{\sigma_1,\dots,\sigma_s\}\subseteq S_G$ to be a $\lambda$-covering net in trace distance. By Proposition~\ref{thm:two-qubit-covering-size} and the monotonicity of the trace distance under the partial trace of the ancilla qubits, we can select $s\leq((n+2G)^2(1+512G/\lambda)^{256})^G$. Then, taking $D(\cdot,\cdot)$ to denote the trace distance,
    \begin{align}\label{eq:complexity-from-covering-size}
        \operatornamewithlimits{Pr}_{t}[\mathcal{C}_\varepsilon(\psi_A(t)|\phi)\leq G] &=\operatornamewithlimits{Pr}_{t}[\exists \,\sigma\in S_G \,|\, D(\psi_A(t),\sigma)\leq \varepsilon ] & \text{(definition)}
        \\&\leq\operatornamewithlimits{Pr}_{t}[\exists \,i\in \{1,\dots,s\} \,|\, {D}(\psi_A(t),\sigma_i)\leq \varepsilon+\lambda ] & \text{($\lambda$-net)}
        \\&\leq\sum_{i=1}^s\operatornamewithlimits{Pr}_{t}[{D}(\psi_A(t),\sigma_i)\leq \varepsilon+\lambda ]. & \text{(union bound)}
    \end{align} 
A similar argument appears in Ref.~\cite{haah2025}; here, we improve the maximum approximation error to be exponentially close to unity. Since $\frac{1}{2}(1+\varepsilon)\geq \varepsilon+\lambda $ follows from our condition $n\geq n_*$, we can use Lemma~\ref{th:temp-deloc} with $\frac{1}{2}(1+\varepsilon)$ tolerance (this choice is allowed for the chosen range of $\varepsilon$):
\begin{equation}
    \operatornamewithlimits{Pr}_{t}[{D}(\psi_A(t),\sigma_i)\leq \varepsilon+\lambda ]\leq \operatornamewithlimits{Pr}_{t}\left[{D}(\psi_A(t),\sigma_i)\leq \frac{1}{2}(1+\varepsilon)\right] \leq  \exp(-\frac{(1-\varepsilon)^2}{8\times 2^{(1-\gamma/2)n-a}})
\end{equation}
with $a=|A|$. Consequently, we have
\begin{equation}
\label{eq:generalinequalitynolog}
   \operatornamewithlimits{Pr}_{t}[\mathcal{C}_\varepsilon(\psi_A(t)|\phi)\leq G] \leq  ((n+2G)^2(1+512G/\lambda)^{256})^G  \exp(-\frac{(1-\varepsilon)^2}{8\times 2^{(1-\gamma/2)n-a}}).
\end{equation}

We now select $G=\lfloor\exp(\gamma n/8)\rfloor+\lceil n/2\rceil\leq 2\exp(\gamma n/8)$. Taking the logarithm of Eq.~\eqref{eq:generalinequalitynolog}, we obtain 
the following chain of inequalities
\begin{align}\label{eq:generalinequalitylog}
    \log\operatornamewithlimits{Pr}_{t}[\mathcal{C}_\varepsilon(\psi_A(t)|\phi)\leq G] &\leq  G \log(\left(n+2G\right)^2\Big(1+\frac{512 G}{\lambda}\Big)^{256})-\frac{(1-\varepsilon)^2}{8\cdot 2^{(1-\gamma/2)n-a}}
    \\& \leq  2\exp(\frac{\gamma n}{8})\left[2n + 256\left(8 + \frac{\gamma n}{4} + \log \frac{1}{\lambda}\right)\right]  -\frac{(1-\varepsilon)^2}{8\cdot 2^{(1-\gamma/2)n-a}} 
    \\& \leq   \exp(\frac{\gamma n}{8})\times 644 n -\frac{(1-\varepsilon)^2}{8}\exp(\frac{\gamma n}{8} 3\log 2) 
    \\&\leq -\exp(\frac{\gamma n}{8}),\label{eq:generalinequalityloglast}
\end{align}
where we used  $\log(1+z)<1+\log z$, $(z>1)$ and $\log(n+2G)\leq n$ in the second inequality; $\lambda=\exp(8-n)$, $\gamma\leq 1$, and $a>(1-\gamma/8)n$ in the third inequality; and $n \geq n_*$ in the last one. 

Finally, we replace $\mathcal{C}_\varepsilon(\psi_A(t)|\phi)$ by $\mathcal{C}_{\varepsilon}(\psi_A(t))$, which is the minimum value of $\mathcal{C}_\varepsilon(\psi_A(t)|\phi)$ over all product states $\phi$. Let $\phi_{*,t}$ be an optimal product state giving $\mathcal{C}_{\varepsilon}(\psi_A(t))=\mathcal{C}_\varepsilon(\psi_A(t)|\phi_{*,t})$. As discussed below \Cref{def:circuit-complexity}, the value of $\mathcal{C}_\varepsilon(\psi_A(t)|\phi)$ is independent of the choice of product state on the ancilla qubits, since the single-qubit channels on the relevant ancilla qubits that participate in the subsequent circuit can be absorbed into the circuit. Thus, it remains only to compare the reduced state $\phi_{*,t,\mathsf{S}}$ on the system qubits $\mathsf{S}$ with $\phi_\mathsf{S}$.

Since Eq.~\eqref{eq:generalinequalityloglast} applies to an arbitrary choice of the product state $\phi$, we may take $\phi$ to be a fixed pure product state on $(n+2G)$ qubits independent of $t$. Importantly, $\phi_{*,t,\mathsf{S}}$ can be prepared from $\phi_\mathsf{S}$ using at most $\lceil n/2\rceil$ two-qubit channels. Therefore, the circuit achieving $\mathcal{C}_\varepsilon(\psi_A(t)|\phi_{*,t})$ can also be used to produce $\psi_A(t)$ from $\phi$, after mapping $\phi_\mathsf{S}\rightarrow\phi_{*,t,\mathsf{S}}$, up to single-qubit channels on the relevant ancilla qubits of $\phi_{*,t}$. Consequently, we have 
$\mathcal{C}_\varepsilon(\psi_A(t)|\phi)\leq \mathcal{C}_\varepsilon(\psi_A(t))+\lceil n/2\rceil$. Since this bound holds independent of $t$, we get
\begin{equation}
    \log\operatornamewithlimits{Pr}_{t}[\mathcal{C}_\varepsilon(\psi_A(t))\leq \exp(\gamma n/8)]\leq -\exp(\frac{\gamma n}{8}).\qedhere
\end{equation}

\end{proof}

\subsection{Proof of Lemma 1}
\begin{lemma}[Double-exponential temporal delocalization; Lemma 1 in the main text]
   \label{th:temp-deloc}
    Let $H$ be a Hamiltonian having spectral ergodicity, and $\psi$ be energy dispersed with constant $\gamma>0$ under $H$. Then, for any subsystem $A$ with size $|A|> (1-\gamma/8)n$, and any $0<\varepsilon<1-\frac{1}{\gamma}\exp(7-\gamma n/16)$,
    \begin{equation}
        \operatornamewithlimits{Pr}_{t}\Big[D(\psi_A(t),\sigma)\leq\varepsilon\Big]\leq \exp(-\frac{(1-\varepsilon)^2}{2\times 2^{(1-\gamma/2)n-a}})\leq \exp(-\exp({\frac{\gamma}{8}\,n}))
    \end{equation}
    for any system size $n\geq \frac{16}{\gamma}\big(7+\log\tfrac{1}{\gamma(1-\varepsilon)}\big)$.
\end{lemma}
\begin{proof}
We denote the fidelity by $F(\cdot,\cdot)$. We first bound the stated probability by the temporal moments of the fidelity using $D\geq 1-\sqrt{F}$ and Markov's inequality:
\begin{equation}
    \operatornamewithlimits{Pr}_{t}[D(\psi_A(t),\sigma)\leq \varepsilon]
    \leq\operatornamewithlimits{Pr}_{t}[F(\psi_A(t),\sigma)\geq (1-\varepsilon)^2]
    \leq\frac{\mathbb{E}_{t}[F(\psi_A(t),\sigma)^k]}{(1-\varepsilon)^{2k}}
\end{equation}
for any positive integer $k$. In Lemma~\ref{thm:k-th-fidelity-bound}, we bound the $k$-th temporal moments as 
\begin{equation}\label{eq:k-th-fidelity-bound}
    \mathbb{E}_t[F(\psi_A(t),\sigma)^k]\leq  2^{k((1-\gamma/2)n-a)} \frac{(2^{n-a}+k-1)!}{(2^{n-a}-1)!},
\end{equation}
where $a=|A|$. This gives
\begin{equation}\label{eq:generalbound}
  \operatornamewithlimits{Pr}_{t}[D(\psi_A(t),\sigma)\leq \varepsilon]  \leq \left(\frac{ 2^{(1-\gamma/2)n-a}}{(1-\varepsilon)^2}\right)^k \frac{(2^{n-a}+k-1)!}{(2^{n-a}-1)!}.
\end{equation}

We are free to select the parameter $k$ in Eq.~\eqref{eq:generalbound} to obtain the desired bound. To select $k$ we study the minimum value of the function $g(k)$,
\begin{equation}
    g(k)=\left(\frac{1}{xy}\right)^k\frac{(x+k-1)!}{(x-1)!},
   \qquad\text{with}\qquad x = 2^{n-a}>0\qquad\text{and}\qquad y = \frac{(1-\varepsilon)^2}{2^{(2-\gamma/2)n-2a}}> 1.
\end{equation}
 Since we have ${g(k+1)}/{g(k)} = {(x+k)}/{(xy)}$ 
it follows that $g(k)$ decreases when $x+k<xy$ and increases when $x+k>xy$. Thus, the minimum of $g(k)$ is achieved either for $k=\lceil x(y-1)\rceil$ or $k=\lfloor x(y-1)\rfloor$. Let $k_*$ be the integer giving the minimum $g(k)$. Then, note that
\begin{align}
g(k_*) e^{x(y-1)}&= g(k_*) \sum_{k=0}^{\infty} \frac{\bigl(x(y-1)\bigr)^k}{k!} 
\leq \sum_{k=0}^{\infty} g(k)\,\frac{\bigl(x(y-1)\bigr)^k}{k!} = \sum_{k=0}^{\infty} \frac{(x+k-1)!}{(x-1)!\,k!}
   \left(1 - \frac{1}{y}\right)^{\!k} \nonumber= \left(\frac{1}{y}\right)^{-x} = y^x,
\end{align}
where we used the negative binomial series $\sum_{k=0}^{\infty} \binom{x+k-1}{k} \tau^k = (1-\tau)^{-x}$, valid for $|\tau|<1$, evaluated at $\tau = 1 - 1/y$. Thus, we find that
\begin{equation}\label{eq:boundforgstar}
    g(k_*)\leq y^x e^{-(y-1)x}\leq e^{-y x/2},
\end{equation}
where the last inequality follows from $y\geq 6$ (one can check that in fact $y\gg6$ using $a > (1-\gamma/8)n$ and $n\geq \frac{16}{\gamma}\log\tfrac{\exp(7)}{\gamma(1-\varepsilon)}$).
Substituting the inequality in Eq.~\eqref{eq:boundforgstar} into Eq.~\eqref{eq:generalbound} with $k=k_*$ yields
\begin{equation*}
   \operatornamewithlimits{Pr}_{t}[D(\psi_A(t),\sigma)\leq \varepsilon] \leq  \exp(-\frac{(1-\varepsilon)^2}{2\times 2^{(1-\gamma/2)n-a}})\leq\exp(-\exp(\frac{\gamma}{8}\, n)).\qedhere
\end{equation*}
\end{proof}

\section{Corollaries}
\label{sec:05}
In this section, we provide full proofs for the corollaries of Theorems 1 and 2 stated in the main text. 

\subsection{Corollary 1: Sustained growth}

Here, we prove Corollary 1 in the main text, on the sustained growth of circuit complexity. The idea is to use the Hamiltonian simulation bound developed in Ref.~\cite{Haah2023}: the time evolution operator $e^{-iH\Delta t}$ of a bounded local Hamiltonian can be approximately implemented by a circuit with $n\Delta t\,\mathrm{polylog}(n\Delta t/\eta)$ gates up to $\eta$ error in the operator norm [defined in Sec.~\ref{sec:setting}]. This bound implies that the complexity cannot increase too rapidly over any fixed time interval. Since the complexity is exponentially large at late times by \Cref{th:theoremS1}, while it is initially vanishing, one may expect that there are exponentially many intermediate time intervals in which the complexity grows before saturation. However, this argument does not prove the sustained growth of circuit complexity with a fixed approximation tolerance $\varepsilon$. Hamiltonian simulation over each time interval introduces an additional error, so the Hamiltonian simulation bound can only give an ordering between robust complexities with different tolerances. Importantly, robust circuit complexity may change drastically as the tolerance varies. 

To manage this difficulty, we introduce the averaged circuit complexity $\bar{\mathcal{C}}_{\varepsilon,\Delta\varepsilon}=\frac{1}{2\Delta\varepsilon}\int_{\varepsilon-\Delta\varepsilon}^{\varepsilon+\Delta\varepsilon}d\varepsilon'\mathcal{C}_{\varepsilon'}$. We note that this integral is well-defined as $\mathcal{C}_{\varepsilon}$ is a monotonic function in $\varepsilon$. This averaging stabilizes the complexity under a sufficiently small change in the tolerance, which allows us to show a sustained growth in time. 

\begin{corollary}[Sustained growth of complexity over exponentially long time; Corollary 1 in the main text]\label{thm:cor1}
    Under the conditions of Theorem~\ref{th:theoremS1}, for any $\Delta\varepsilon$ and $\varepsilon$ such that $\exp(-O(n))<\Delta\varepsilon<1/4$ and $\Delta\varepsilon(1+1/2^{n+1})\leq\varepsilon<1-\Delta\varepsilon-\frac{2}{\gamma}\exp(7-\gamma n/16)$, before $\bar{\mathcal{C}}_{\varepsilon,\Delta\varepsilon}(\psi(t))$ saturates at an exponentially large value, its growth is sustained over an exponentially long time, in the sense that for any constant $\Delta t>0$, there are exponentially many times $t_1<t_2<\cdots <t_M$ separated by at least $\Delta t\leq t_{m+1}-t_m$ such that
    \begin{equation}
         \bar{\mathcal{C}}_{\varepsilon,\Delta\varepsilon}(\psi(t_{m+1}))> \bar{\mathcal{C}}_{\varepsilon,\Delta\varepsilon}(\psi(t_{m}))
    \end{equation}
    for each $m=1,\ldots,M-1$.
\end{corollary}
\begin{proof}
    
    By \Cref{th:theoremS1}, with probability larger than $1-c\exp(-\gamma n/2)$ over the choice of $\psi\sim\mathcal{E}_{\beta,\delta}$, for sufficiently large $n$, there exists a late time $t_+$ such that $\mathcal{C}_{\varepsilon+\Delta\varepsilon}(\psi(t_+))>\exp(\gamma n/8)$. This implies that for all $x\in[0,\varepsilon+\Delta\varepsilon]$, $\mathcal{C}_x(\psi(t_+))>\exp(\gamma n/8)$, and thus, $\bar{\mathcal{C}}_{\varepsilon,\Delta\varepsilon}(\psi(t_+))>\exp(\gamma n/8)$. Let us fix such a choice of $\psi$, and let $t_+$ be the earliest time at which this inequality holds. 

    Let $J_+$ be the integer such that $J_+\Delta t<t_+\leq(J_+ +1)\Delta t$. By definition, for $x\in[0,1]$, $\mathcal{C}_x(\psi(t))$ is upper bounded by the gate count of any circuit preparing $\psi(t)$ up to $x$ error. For $H$ satisfying (P1)-(P3), its time evolution operator $e^{-iH\Delta t}$ can be approximated by a circuit with $G_{\Delta t,\eta}=n\Delta t\,\mathrm{polylog}(n\Delta t/\eta)$ gates up to $\eta$ error in the operator norm~\cite{Haah2023}. We set $\eta$ as $\Delta\varepsilon/2^{n+1}\geq\exp(-O(n))$, which gives $G_{\Delta t,\eta}=\mathrm{poly}(n)$. The Hamiltonian simulation bound gives
    \begin{equation}
        \mathcal{C}_x(\psi((j+1)\Delta t))\leq \mathcal{C}_{x-\eta}(\psi(j\Delta t))+G_{\Delta t,\eta}
    \end{equation}
    for $x\in[\eta,1]$ and all $j=0,\ldots,J_+$. We now integrate both sides over $x\in[\varepsilon-\Delta\varepsilon,\varepsilon+\Delta\varepsilon]$ to match the error tolerances of the complexities. This gives
    \begin{equation}
        \begin{split}
            \bar{\mathcal{C}}_{\varepsilon,\Delta\varepsilon}(\psi((j+1)\Delta t))
            &\leq \frac{1}{2\Delta\varepsilon}\int_{\varepsilon-\Delta\varepsilon}^{\varepsilon+\Delta\varepsilon}dx\,\mathcal{C}_{x-\eta}(\psi(j\Delta t))+G_{\Delta t,\eta}\\
            &= \bar{\mathcal{C}}_{\varepsilon,\Delta\varepsilon}(\psi(j\Delta t))+\frac{1}{2\Delta\varepsilon}\left[\int_{\varepsilon-\Delta\varepsilon-\eta}^{\varepsilon-\Delta\varepsilon}-\int_{\varepsilon+\Delta\varepsilon-\eta}^{\varepsilon+\Delta\varepsilon}\right]dx\,\mathcal{C}_x(\psi(j\Delta t))+G_{\Delta t,\eta}\\
            &\leq \bar{\mathcal{C}}_{\varepsilon,\Delta\varepsilon}(\psi(j\Delta t)) + \frac{\eta}{2\Delta\varepsilon}\mathcal{C}_0(\psi(j\Delta t))+G_{\Delta t,\eta},
        \end{split}
    \end{equation}
    where we use the monotonicity of the complexity in the third line. Since any $n$-qubit state can be exactly prepared by $2^{n+1}$ arbitrary two-qubit gates~\cite{Shende2006}, $\mathcal{C}_0(\psi(j\Delta t))$ is at most $2^{n+1}$, giving
    \begin{equation}\label{eq:unit-time-c-bound}
        \bar{\mathcal{C}}_{\varepsilon,\Delta\varepsilon}(\psi((j+1)\Delta t))\leq  \bar{\mathcal{C}}_{\varepsilon,\Delta\varepsilon}(\psi(j\Delta t)) + G_{\Delta t,\eta} + 1.
    \end{equation}
    Therefore, for each constant time step, the averaged complexity can grow at most polynomially.

    Next, we construct a sequence $\{l_m\}$ of times from $l_1=0$ by recursively defining
    \begin{equation}
        l_{m+1} = \operatornamewithlimits{argmin}_{l_m<j\leq J_+}\{\bar{\mathcal{C}}_{\varepsilon,\Delta\varepsilon}(\psi(j\Delta t))>\bar{\mathcal{C}}_{\varepsilon,\Delta\varepsilon}(\psi(l_m\Delta t))\}.
    \end{equation}
    Denote by $M$ the length of this sequence. By definition, we have $\bar{\mathcal{C}}_{\varepsilon,\Delta\varepsilon}(\psi((l_{m+1}-1)\Delta t))\leq \bar{\mathcal{C}}_{\varepsilon,\Delta\varepsilon}(\psi(l_m\Delta t))$. This, together with Eq.~\eqref{eq:unit-time-c-bound}, implies
    \begin{equation}
        \bar{\mathcal{C}}_{\varepsilon,\Delta\varepsilon}(\psi(l_{m+1}\Delta t))\leq \bar{\mathcal{C}}_{\varepsilon,\Delta\varepsilon}(\psi((l_{m+1}-1)\Delta t))+G_{\Delta t,\eta}+1\leq \bar{\mathcal{C}}_{\varepsilon,\Delta\varepsilon}(\psi(l_m\Delta t))+G_{\Delta t,\eta}+1
    \end{equation}
    for each $m=1,\ldots,M-1$. Since, by Eq.~\eqref{eq:unit-time-c-bound} again, 
    \begin{equation}
        \bar{\mathcal{C}}_{\varepsilon,\Delta\varepsilon}(\psi(J_+\Delta t))>\exp(\gamma n/8)-G_{\Delta t,\eta}-1,
    \end{equation}
    and using the definition of $\{l_m\}$, we get
    \begin{equation}
        \bar{\mathcal{C}}_{\varepsilon,\Delta\varepsilon}(\psi(l_1\Delta t))+(M-1)(G_{\Delta t,\eta}+1)\geq \bar{\mathcal{C}}_{\varepsilon,\Delta\varepsilon}(\psi(l_M\Delta t))\geq \bar{\mathcal{C}}_{\varepsilon,\Delta\varepsilon}(\psi(J_+\Delta t))>\exp(\gamma n/8)-G_{\Delta t,\eta}-1.
    \end{equation}
    Using $\bar{\mathcal{C}}_{\varepsilon,\Delta\varepsilon}(\psi(l_1\Delta t))=0$, this with $G_{\Delta t,\eta}=\mathrm{poly}(n)$ implies
    \begin{equation}
        M>\exp(\gamma n/8)/(G_{\Delta t,\eta}+1)\geq \exp(\Omega(n)).
    \end{equation}
    Finally, setting $t_m=l_m\Delta t$ gives $t_{m+1}-t_m\geq\Delta t$, $t_M<t_+$, and $\bar{\mathcal{C}}_{\varepsilon,\Delta\varepsilon}(\psi(t_{m+1}))>\bar{\mathcal{C}}_{\varepsilon,\Delta\varepsilon}(\psi(t_m))$ for each $m=1,\ldots,M-1$. This completes the proof.
\end{proof}

\subsection{Corollary 2: Irremovable entanglement}
Consider a partition $\{A_m\}_{m=1}^M$ of the lattice into $M$ regions, with average region size $\bar a=\tfrac{1}{M}\sum_m |A_m|=\frac{n}{M}$, and define the average subsystem entropy $\bar S(\psi)=\frac{1}{M}\sum_{m=1}^M S(\psi_{A_m})$. 

\begin{corollary}[Irremovable volume-law entanglement; Corollary 2 in the main text] 
    Let $\eta\in (0,1/4)$ and $b\geq 4$. Assume that $\max_m |A_m|\leq O(\bar a)$. Under the conditions of Theorem~\ref{th:theoremS1}, at typical late times, the average subsystem entropy of $\psi(t)$ follows a volume law, \textit{i.e.}, $\bar S(t)\geq \Omega(\bar a)$. Furthermore, the state cannot be mapped to a state $\phi$ satisfying $\bar S(\phi)\leq O(\bar a^{1-b\eta})$ by any quantum circuit with $\exp(O(n^{1-\eta}))$ two-qubit gates for $ \Omega(n^\eta)\leq M\leq O(n^{1-3/b})$.
\end{corollary}
\begin{proof}
We show this corollary in two steps. First, we show that any low-entangled state can be prepared by a low-complexity circuit. Second, we suppose for contradiction that a late-time state can be disentangled by a low-complexity circuit, from which we can prepare the late-time state easily by first preparing the resulting disentangled state and subsequently applying the inverse of the disentangling unitary. This contradicts Theorem~\ref{th:theoremS1}, which shows that the late-time state has exponential circuit complexity.

Before proceeding, we first construct a coarse-grained partition. Define $L=\lfloor C_L n^{\eta}\rfloor \in\mathbb{N}$ for some constant $C_L>0$ such that $L\leq M$ and the cumulative sizes $l_j=\sum_{i=1}^j |A_i|$. For each $j=1,\ldots,L-1$, let $r_j$ be the smallest index such that $l_{r_j}\geq nj/L$, and set $r_0=0$, $r_L=M$. Then, define the coarse-grained regions $I_j=\bigcup_{i=r_{j-1}+1}^{r_j} A_i$. By construction, we have $|I_j|\leq O(n/L+\bar{a})\leq O(n^{1-\eta})$ for all $j=1,\ldots,L-1$. By subadditivity, we have
\begin{equation}\label{eq:coarse-grained-en-subadd}
    \sum_{j=1}^L S(\phi_{I_j}) \leq \sum_{i=1}^M S(\phi_{A_i}) = M\bar{S}(\phi).
\end{equation}

Let us begin with the second part. Let $t_*$ be a typical late time for which Theorem~\ref{th:theoremS1} gives $\mathcal{C}_\varepsilon(\psi(t_*))>\exp(\gamma n/8)$. We argue by contradiction. Suppose that there is a circuit $U$ made of at most $\exp(O(n^{1-\eta}))$ two-qubit gates such that $U\ket{\psi(t_*)}=\ket{\phi}$ and $\bar S(\phi)\leq C\bar a^{1-b\eta}$, where $C$ is a constant independent of $n$.

Order the coarse-grained regions as $I_1,\ldots,I_L$ and regard them as the sites of a one-dimensional matrix product state after blocking. The Hilbert space dimension of each blocked site is at most $d\leq 2^{C_d n^{1-\eta}}$ for some constant $C_d>0$. We denote $d_\mathrm{max}$ as the maximum Hilbert space dimension among the blocks.

For each $j=1,\ldots,L-1$, let $B_j=\bigcup_{m=1}^j I_m$. By subadditivity and Eq.~\eqref{eq:coarse-grained-en-subadd},
\begin{equation}
    S(\phi_{B_j})
    \leq \sum_{m=1}^j S(\phi_{I_m})
    \leq M\bar S(\phi)
    \leq C M\bar a^{1-b\eta}.
\end{equation}
Now write the Schmidt decomposition of $\phi$ across the cut $B_j|B_j^c$ as $\ket{\phi}=\sum_i s_i\ket{i}_{B_j}\ket{i}_{B_j^c}$, with the Schmidt values ordered in non-increasing order. If $p=\sum_{i=1}^{\chi_{\mathrm c}}s_i^2$, then the normalized truncation to the largest $\chi_{\mathrm c}$ Schmidt values has error $\sqrt{1-p}$. Since $s_i^2\leq 1/\chi_{\mathrm c}$ for $i>\chi_{\mathrm c}$, we have
$S(\phi_{B_j})\geq (1-p)\log\chi_{\mathrm c}$. Hence, the error at this cut is at most $\varepsilon/L$ provided $\log\chi_{\mathrm c}\geq L^2S(\phi_{B_j})/\varepsilon^2$.
Thus, it is enough to choose $\chi_{\mathrm c}$ so that
\begin{equation}
    \log\chi_{\mathrm c}
    \geq
    \frac{C}{\varepsilon^2}M L^2\bar a^{1-b\eta}.
\end{equation}
Using $L=\Theta(n^\eta)$, $\bar a=n/M$, $\Omega(n^{3/b})\leq\bar{a}\leq O(n^{1-\eta})$, and $\eta<1/4$, we have $ML^2\bar{a}^{1-b\eta}\leq O(n^{1-\eta})$. Therefore, we may choose
\begin{equation}
    \chi_{\mathrm c}\leq \exp(O(n^{1-\eta})).
\end{equation}

By standard MPS truncation bounds (see Lemma~1 of Ref.~\cite{Verstraete2006}), truncating all consecutive cuts in this way gives a matrix product state $\tilde\phi$ with maximum bond dimension $\chi_{\mathrm c}$ and total error at most $\varepsilon$. By the usual sequential generation of MPS, such a state can be prepared using $O(Ld_\mathrm{max}\chi_{\mathrm c}^2)$ two-qubit
gates~\cite{Schon2005,Shende2006,perezgarcia2007,Iten2016}. Since $L=\Theta(n^\eta)$, $d_\mathrm{max}\leq 2^{C_d n^{1-\eta}}$, and $\chi_{\mathrm c}\leq \exp(O(n^{1-\eta}))$, this preparation uses $\exp(O(n^{1-\eta}))$ two-qubit gates.

We can therefore prepare $\psi(t_*)$ up to error $\varepsilon$ by first preparing $\tilde\phi$ and then applying $U^{-1}$. The total number of two-qubit gates is still $\exp(O(n^{1-\eta}))$. Since $\eta>0$, this is subexponential in $n$, contradicting $\mathcal{C}_\varepsilon(\psi(t_*))>\exp(\gamma n/8)$ for sufficiently large $n$. Therefore, no such circuit $U$ exists.

The remaining part is to show that $\bar S(t)\geq \Omega(\bar a)$. In fact, by the argument above with trivial $U$, we have already shown that $\bar S(t)\geq \Omega(\bar a^{1-b\eta})$. Since the argument is uniform over fixed choices of $\eta\in(0,1/4)$ and $b\geq 4$, we may choose them so that $b\eta$ is an arbitrarily small fixed constant, subject to the corresponding conditions on $M$, but we can remove the $-b\eta$ in the exponent by modifying techniques from Ref.~\cite{Huang2021}, as shown in \Cref{thm:extensive-entanglement2} below.
\end{proof}

\begin{prop}[Extensive average entanglement from spectral ergodicity and energy dispersion]
Let $H$ be a Hamiltonian having spectral ergodicity, and let $\psi$ be
energy dispersed with constant $\gamma>0$ under $H$. Assume that the partition
$\{A_m\}_{m=1}^M$ with $M\leq n$ introduced above satisfies $a_{\max}\coloneqq \max_m |A_m|\leq \gamma n/4$.
Then, 
\begin{equation}
    \label{eq:probovert312}
    \operatornamewithlimits{Pr}_t\left[\bar S(\psi(t))
        \leq \bar a \frac{\gamma  \log 2}{2}
    \right]
    \leq 2^{-\gamma n/8}
\end{equation}
for all $n\geq \frac{16}{\gamma}(7+\log_2\frac{2}{\gamma})$.
\label{thm:extensive-entanglement2}
\end{prop}

\begin{proof} We follow the proof of Theorem 1 in Ref.~\cite{Huang2021}. Let $\rho_d$ be the diagonal state of $\psi$ in the energy eigenbasis,
and let $\rho_{d,A_m}$ be its reduced state on $A_m$. Fix
$\alpha=\gamma/8$. Applying \Cref{thm:relaxation} below to each region and
using the union bound gives
\begin{equation}
    \operatornamewithlimits{Pr}_t\left[\forall m\colon \left\|\psi_{A_m}(t)-\rho_{d,A_m}\right\|_1<2^{-\alpha n+1} \right]>1-\frac{M}{4}2^{-\gamma n+2a_{\max}+2\alpha n}.    
\end{equation}

On this event, the trace distance between
$\psi_{A_m}(t)$ and $\rho_{d,A_m}$ is at most $2^{-\alpha n}$ for every
$m$. By the continuity bound for the von Neumann entropy,
\begin{equation}
    |\bar{S}(\psi(t))-\bar S(\rho_d)|\leq \frac{1}{M}\sum_{m=1}^M
    \left|S(\psi_{A_m}(t))-S(\rho_{d,A_m})\right|
    \leq
    2^{-\alpha n}(\bar a+\alpha n+2)\log 2 ,    
\end{equation}
so it is enough to lower bound the average entropy of the diagonal state.
We have
\begin{equation}
   \bar{S}(\rho_d)=\frac{1}{M}\sum_{m=1}^M S(\rho_{d,A_m})\geq \frac{S(\rho_d)}{M}\geq \frac{ \gamma n \log2}{M}=\gamma \bar a\log2,
\end{equation}
where the first inequality is the subadditivity of entropy and the second inequality follows from comparing the von Neumann
entropy with the second R\'enyi entropy and the energy dispersion. Thus, on the same high-probability
event we have
$\bar S(\psi(t))
    \geq
    \left(\gamma\bar a
    -2^{-\alpha n}(\bar a+\alpha n+2)\right)\log 2 .$
For $\alpha=\gamma/8$ and $a_{\max}\leq \gamma n/4$, the second term is at most
$\gamma\bar a/2$ if $n\geq \frac{16}{\gamma}(7+\log_2\frac{2}{\gamma})$, giving Eq.~\eqref{eq:probovert312}.
\end{proof}

\subsection{Corollary 3: Generic no fast-forwarding}

We first introduce the following notion of fast-forwarding. This notion of fast-forwarding allows arbitrary classical preprocessing in the construction of the optimal circuit and arbitrary two-qubit gates.
\begin{definition}[Fast-forwarding]
    A Hamiltonian $H$ can be \textit{$\varepsilon$-approximate fast-forwarded} with complexity $G\geq 0$, if for any time $t$ and any initial state $\phi$, there exists a quantum circuit with $G$ two-qubit gates which prepares $\phi(t)$ to trace distance $\varepsilon$ starting from $\phi$.
\end{definition}

\begin{corollary}[Generic no fast-forwarding; Corollary 3 in the main text]
    For almost all\footnoteref{nt:1} Hamiltonians $H$ which satisfy properties (P1)-(P3) and any positive constant $\varepsilon<1-\exp(-O(n))$, $H$ cannot be $\varepsilon$-approximate fast-forwarded with complexity $G=\mathrm{poly}(n)$. Furthermore, Hamiltonians allowing such fast-forwarding, including all-commuting, 1D finite-range quadratic fermionic (without boundary terms), or number-conserving quadratic bosonic (with particle-number cutoff) Hamiltonians that satisfy (P1)-(P3) must form a measure zero set in the ensemble of all Hamiltonians satisfying (P1)-(P3).
\end{corollary}
\begin{proof}
    By Theorem~\ref{th:theoremS1}, there exist an initial state $\psi$ and a time $t_*>0$ such that $\mathcal{C}_\varepsilon(\psi(t_*))>\exp(\gamma n/8)$. Therefore, there cannot exist an algorithm that can prepare $\psi(t)$ using a polynomial-size circuit up to $\varepsilon$ error in the trace distance, implying that $H$ cannot be $\varepsilon$-approximate fast-forwarded with complexity $G=\mathrm{poly}(n)$.

    Commuting local Hamiltonians, efficiently block-diagonalizable Hamiltonians, 1D finite-range quadratic fermionic Hamiltonians, and number-conserving quadratic bosonic Hamiltonians are known to be fast-forwardable~\cite{Atia2017,Gu2021}. For commuting local Hamiltonians, the evolution factorizes into a product of local rotations. For 1D finite-range quadratic fermions, the Jordan-Wigner transformation maps the system to a local qubit Hamiltonian whose dynamics remains Gaussian and can be compiled using polynomially many two-qubit gates once the single-particle phase rotation angles are worked out. For number-conserving quadratic bosons, the analogous statement holds after a qubit encoding with a suitable occupation cutoff. Therefore, these Hamiltonians form a measure zero set in the ensemble of all Hamiltonians satisfying (P1)-(P3) due to Theorem~\ref{th:theoremS1}.
\end{proof}

\subsection{Corollary 4: Bound on recurrence time}
As a consequence of the spectral ergodicity established in \Cref{thm:generic-ergodic} and the energy dispersion established in \Cref{thm:inf-energy-dispersion,thm:energy-dispersion}, we present a bound on the recurrence time, which follows from the main result of Ref.~\cite{Riddell2023}. 

We follow the notion of recurrence introduced in Ref.~\cite{Riddell2023}. Let $\Delta$ denote the recurrence duration. We say that an $(\varepsilon,\Delta)$-recurrence occurs on an interval $\mathcal{I}=[t_*,t_*+\Delta]$ if 
\begin{equation}
    1-F(\psi(t),\psi)\leq \varepsilon
\end{equation}
for all $t\in\mathcal{I}$. For local Hamiltonians, which have $\mathrm{Var}_\psi(H)\leq O(n^2)$, the quantum speed limit~\cite{Anandan1990,Taddei2013} suggests a physically natural choice of the time scale $\Delta=\Theta(\sqrt{\varepsilon}/n)$.

Let $t_\Delta^k$ denote the initial time of the interval associated with the $k$-th $(\varepsilon,\Delta)$-recurrence. We define the average $(\varepsilon,\Delta)$-recurrence time by 
\begin{equation}
    T_\mathrm{rec}(\varepsilon,\Delta) \coloneq \lim_{k\rightarrow\infty} \frac{t_\Delta^k}{k}.
\end{equation}
We note that, depending on the choice of $\Delta$, there may be no time interval satisfying the recurrence condition. In this case, we say that $(\varepsilon,\Delta)$-recurrence does not occur and set $T_\mathrm{rec}(\varepsilon,\Delta)\coloneq\infty$. Then, the following corollary provides a lower bound on $T_\mathrm{rec}(\varepsilon,\Delta)$. 

\begin{corollary}[Double-exponential recurrence for generic local Hamiltonians; Corollary 4 in the main text]
    Under the conditions of Theorem~\ref{th:theoremS1}, for any constant $0<\varepsilon<1$, and $\Delta\geq 1/\mathrm{poly}(n)$, the average recurrence time $T_\mathrm{rec}(\varepsilon,\Delta)$ of $\psi$ is at least $\exp(\exp(\Omega(n)))$. 
\end{corollary}
\begin{proof}
    We leverage the following result: 
    \begin{lemma}[Corollary 2 of Ref.~\cite{Riddell2023}]
        Let $H$ be a Hamiltonian having spectral ergodicity, and $\psi$ be an initial state with effective dimension $D_\psi^\mathrm{eff}$. Then,
        \begin{equation}
            T_\mathrm{rec}(\varepsilon,\Delta)\geq \frac{\Delta}{2e}\exp(\frac{1-\varepsilon}{e}D_\psi^\mathrm{eff}).
        \end{equation}
    \end{lemma}
    \noindent As discussed in the proof of \Cref{th:theoremS1}, $\psi$ is energy dispersed with $\gamma>0$, and thus $D_\psi^\mathrm{eff}$ is exponentially large. Therefore, the average recurrence time $T_\mathrm{rec}(\varepsilon,\Delta)$ is lower bounded by $\Delta\exp(\exp(\Omega(n)))$.
\end{proof}

\section{Low complexity of small regions}
\label{sec:06}

In the previous sections, we focused on the circuit complexity of global states. In this regime, \Cref{th:theoremS1} shows that whenever the subsystem is larger than a certain constant fraction of the whole system, the circuit complexity of the late-time reduced density matrix is exponentially large. On the other hand, for a smaller region, this result breaks down. For example, if the subsystem has a constant size, then its reduced state can be trivially prepared by a circuit with a constant number of local channels. Furthermore, for a growing, but still small region $A$ with $|A|=\omega(1)$, one expects relaxation within $A$ to occur on a relatively short timescale in generic thermalizing systems. Consequently, the circuit complexity of the subsystem is also expected to saturate at a relatively small value. Here, we make this intuition precise for thermalizing systems whose local late-time states are described by thermal ensembles. Before proceeding, we briefly review  quantum thermalization.

\begin{definition}[Temporal-thermal equivalence]
    Consider a Hamiltonian $H$ and an initial state $\psi$. Let $\rho_d$ be the energy diagonal state of $\psi$ with respect to $H$. For $0<\varepsilon<1$, we say that $\psi$ satisfies {\it $\varepsilon$-approximate temporal-thermal equivalence} under $H$ on a region $A$ at inverse temperature $\beta$ if we have
    \begin{equation}
        \frac{1}{2}\left\|\tr_{A^c}(\rho_d - g_\beta)\right\|_1 \leq \varepsilon.
    \end{equation} 
\end{definition}

We say that an initial state thermalizes on a region $A$ under a Hamiltonian if two conditions hold: 1) its reduced state in $A$ relaxes to that of its energy diagonal state, and 2) it satisfies temporal-thermal equivalence on $A$. The first condition is guaranteed when the initial state is sufficiently energy dispersed as originally recognized by  Refs.~\cite{Reimann2008,Linden2009} and shown in Lemma~\ref{thm:relaxation} below. Moreover, at sufficiently high temperature, the Hamiltonian satisfies the uniform clustering property, and its Gibbs state becomes efficiently preparable, as shown in Lemma~\ref{thm:local-gibbs-recovory-complexity} below using algorithms developed in Refs.~\cite{Brandao2019,chen2025}.

Combining these ingredients, we show that assuming temporal-thermal equivalence of an initial state,  uniform clustering of the Hamiltonian (which provably occurs at high temperature), and energy dispersion of the initial state, the late-time complexity of the reduced state in $A$ is small.

\begin{lemma}[Late-time relaxation of small subsystem \cite{Linden2009}]\label{thm:relaxation}
    Consider a Hamiltonian $H$ for $n$ qubits with nondegenerate spectral gaps. Let $\psi$ be an initial state that is energy dispersed with constant $\gamma$. Then, for any subsystem $A$ and $\varepsilon>0$, we have
    \begin{equation}
        \operatornamewithlimits{Pr}_t[\left\|\tr_{A^c}(\psi(t)-\rho_d)\right\|_1\geq \varepsilon]\leq \varepsilon^{-2} 2^{-\gamma n+2|A|},
    \end{equation}
    where $\rho_d$ is the energy diagonal state of $\psi$.
\end{lemma}
\begin{proof}
    We generalize the proofs of Theorem 1 of Ref.~\cite{Cotler2022} and Theorem 3 of Ref.~\cite{fan2025} to an energy dispersed initial state under time-independent Hamiltonian dynamics. Let $\rho_{d,A}$ be the subsystem of the diagonal density matrix in $A$, i.e., $\rho_{d,A}=\tr_{A^c}(\rho_d)$. Then, for any $\varepsilon> 0$, the probability of $\psi_A(t)$ being $\varepsilon$-far from $\rho_{d,A}$ is upper bounded by
    \begin{equation}
        \begin{split}
            \operatornamewithlimits{Pr}_t[\left\|\psi_A(t)-\rho_{d,A}\right\|_1\geq \varepsilon]
            &\leq \operatornamewithlimits{Pr}_t\left[2^{|A|}\left\|\psi_A(t)-\rho_{d,A}\right\|_2^2\geq \varepsilon^2\right]\\
            &=\operatornamewithlimits{Pr}_t\left[\tr([\psi_A(t)]^2) + \tr(\rho_{d,A}^2) - 2 \tr(\psi_A(t) \rho_{d,A}) \geq \varepsilon^2 2^{-|A|}\right]\\
            &\leq \frac{2^{|A|}}{\varepsilon^2}\mathbb{E}_t[\tr([\psi_A(t)]^2) + \tr(\rho_{d,A}^2) - 2 \tr(\psi_A(t) \rho_{d,A})],
        \end{split}
    \end{equation}
    where we used Markov's inequality at the last line. Since $H$ has nondegenerate spectrum, we have $\mathbb{E}_t[\psi(t)]=\rho_d$, and the upper bound becomes
    \begin{equation}
        = \frac{2^{|A|}}{\varepsilon^2} \left(\mathbb{E}_t[\tr([\psi_A(t)]^2)] - \tr(\rho^2_{d,A})\right).
    \end{equation}
    Next, let us compute $\mathbb{E}_t[\tr([\psi_A(t)]^2)]$. We can expand and bound it as
    \begin{equation}
        \begin{split}
            \mathbb{E}_t[\tr([\psi_A(t)]^2)] 
            &= \tr(\rho_{d,A}^2) + \sum_{\mu_1\neq\mu_2}p_{\mu_1}p_{\mu_2}\tr(\tr_A(\ketbra{\mu_1})\tr_A(\ketbra{\mu_2}))\\
            &\leq \tr(\rho_{d,A}^2) + \tr(\rho_{d,A^c}^2),
        \end{split}
    \end{equation}
    where we used the fact that $H$ has nondegenerate spectral gaps. Here, $p_\mu$ is the energy distribution defined by $\abs{\braket{\mu}{\psi}}^2$. 
    Therefore, we have
    \begin{equation}
        \operatornamewithlimits{Pr}_t[\left\|\psi_A(t)-\rho_{d,A}\right\|_1\geq \varepsilon] \leq \frac{2^{|A|}}{\varepsilon^2} \tr(\rho^2_{d,A^c}).
    \end{equation}
    Since $\psi$ is energy dispersed with constant $\gamma$, it follows that
    \begin{equation}
        \tr(\rho^2_{d,A^c})\leq 2^{-\gamma n+|A|}.
    \end{equation}
    These give the stated bound.
\end{proof}

To prove the efficient implementation of the thermal state, we leverage the following two results:
\begin{lemma}[Corollary 3.3 of Ref.~\cite{chen2025}]\label{thm:global-gibbs-recovory-complexity}
    Under the assumption of uniform clustering for $H$ at inverse temperature $\beta$, there exists a channel with (dissipative) gate complexity $e^{O(\log^D(n/\varepsilon))}$ outputting a state $\rho$ such that $\|g_\beta-\rho\|_1\leq\varepsilon$.
\end{lemma}
\begin{lemma}[Theorem 3.2 of Ref.~\cite{chen2025}; Theorem 5 of Ref.~\cite{Brandao2019}]\label{thm:local-indistinguishability}
    Consider a Hamiltonian $H$ on a $D$-dimensional lattice and an inverse temperature $\beta$. Suppose the pair $(H,\beta)$ is uniformly clustering with correlation length $\xi$. Then, the pair satisfies local indistinguishability: For any $ABC=X\subseteq\Lambda$, with the distance $d(A,C)=l$, we have that
    \begin{equation}
        \|\tr_{BC}(g_\beta^X)-\tr_B(g_\beta^{AB})\|_1\leq e^{c'\beta}|\partial_B C|\left(\mathrm{poly}(|A|,l^D)e^{-l/2\xi}+e^{-l/c}\right)
    \end{equation}
    for some universal constant $c$ and a constant $c'>0$ which depends on $\beta$ and the locality of $H$, and where $\partial_B C$ is the boundary of $C$ with $B$. Here, $g_\beta^X$ is the Gibbs state of the reduced Hamiltonian in $X$ with the inverse temperature $\beta$.
\end{lemma}
\begin{lemma}[Efficient implementation of reduced state of Gibbs state]\label{thm:local-gibbs-recovory-complexity}
    Let $A$ be a contiguous region in a $D$-dimensional cube of length $s$ on a $D$-dimensional lattice. Under the assumption of uniform clustering for $H$ at inverse temperature $\beta$, there exists a channel with (dissipative) gate complexity $e^{O(\log^D(|A|/\varepsilon))}$ outputting a state $\rho_A$ such that $\left\|\tr_{A^c}(g_\beta)-\rho_A\right\|_1\leq\varepsilon$.
\end{lemma}
\begin{proof}
    We consider a buffer region $B$ that encloses $A$ such that $AB$ is a cube of length $s+2l$, and $d(A,C)\geq l$ with the complement $C$ of $AB$. Let $\partial_B C$ be the boundary of $C$ with $B$. Then, $|\partial_B C|$ is given by $O((s+2l)^{D-1})$, and using Lemma~\ref{thm:local-indistinguishability}, we have
    \begin{equation}
        \begin{split}
            \|\tr_{A^c}(g_\beta)-\tr_B(g_\beta^{AB})\|_1
            &\leq e^{c'\beta}O((s+2l)^{D-1})\left(\mathrm{poly}(s^D,l^D)e^{-l/2\xi}+e^{-l/c}\right)\\
            &\leq O\left(\mathrm{poly}(s^D,l^D)e^{-l/\xi'}\right)
        \end{split}
    \end{equation}
    with $\xi'=\max(2\xi,c)$. We can choose $l$ such that this error is upper bounded by $\varepsilon/2$ as follows:
    \begin{equation}
        l = m \log(|A|/\varepsilon)
    \end{equation}
    for some constant $m>0$. Since $g_\beta^{AB}$ is the Gibbs state of the reduced Hamiltonian of $H$ in $AB$, it can be directly implemented by the local Lindbladian evolution in Ref.~\cite{chen2025} with the gate complexity of $e^{O(\log^D(|A|/\varepsilon))}$ according to Lemma~\ref{thm:global-gibbs-recovory-complexity} up to $\varepsilon/2$ error. Thus, $\tr_{A^c}(g_\beta)$ can be approximately implemented using the same number of gates up to $\varepsilon$ error in the trace distance.
\end{proof}
We remark that the circuit construction of this lemma utilizes a local Lindbladian given by the sum of single-qubit Pauli jump operators~\cite{chen2025}. Therefore, the gates are geometrically local. 

\subsection{Proof of Theorem 2}
\begin{theorem}[Thermalizing systems must have small complexity for small subsystems; Theorem 2 in the main text]
    \label{thm:low-c-from-therm}
    Let $H$ be a Hamiltonian, $\psi$ be an initial state, and $A$ be a contiguous region. Assume that $\psi$ has a sufficiently low inverse temperature $\beta$ that ensures uniform clustering for $H$. For any positive constant $\varepsilon<1$, suppose $\psi$ thermalizes under $H$ on $A$ with probability larger than $P_t$ and error $2\varepsilon/3$ in trace distance. Then, we have
    \begin{equation}
        \operatornamewithlimits{Pr}_{t}\Big[\mathcal{C}_\varepsilon(\psi_A(t))\leq  \exp(O(\log^D(|A|/\varepsilon)))\Big]>P_t.
    \end{equation}
\end{theorem}
\begin{proof}
    By the definition of thermalization, $\psi_A(t)$ can be well approximated by the reduced density matrix $g_{\beta,A}$ of the Gibbs state $g_\beta$ in $A$:
    \begin{equation}
        \operatornamewithlimits{Pr}_{t}\Big[D(\psi_A(t),g_{\beta,A})\leq \frac{2\varepsilon}{3}\Big]>P_t.
    \end{equation}
    At the same time, by Lemma~\ref{thm:local-gibbs-recovory-complexity} and uniform clustering for $H$, $g_{\beta,A}$ can be approximated by the reduced density matrix $\tilde{g}_{\beta,A}$ of $\tilde{g}_\beta$ in $A$ which can be implemented using $e^{O(\log^D(|A|/\varepsilon))}$ local two-qubit channels up to $\varepsilon/3$ error in the trace distance. Therefore, we have
    \begin{equation}
        \operatornamewithlimits{Pr}_t\left[D(\psi_A(t),\tilde{g}_{\beta,A})\leq \varepsilon \right] > P_t,
    \end{equation}
    which implies the stated bound on the circuit complexity.
\end{proof}

\subsection{Proof of Theorem 4}
\begin{theorem}[Small subsystem complexity from energy dispersion and temporal-thermal equivalence; Theorem 4 in the main text]
    \label{thm:theorem3}
    Let $H$ be a Hamiltonian with non-degenerate spectral gaps that is uniformly clustering at inverse temperature $\beta$. Let $\psi$ be an energy dispersed state with constant $\gamma$. For any positive constant $\varepsilon<1$, assume that $\psi$ satisfies $(\varepsilon/3)$-approximate temporal-thermal equivalence under $H$ on a contiguous region $A$. Then, we have
    \begin{equation}
        \operatornamewithlimits{Pr}_{t}\Big[\mathcal{C}_\varepsilon(\psi_A(t))\leq  \exp(O(\log^D(|A|/\varepsilon)))\Big]>1- 9\varepsilon^{-2}2^{-\gamma n+2|A|}.
    \end{equation}
\end{theorem}
\begin{proof}
    Due to Lemma~\ref{thm:relaxation} and the energy dispersion of $\psi$, $\psi_A(t)$ can be approximated by the reduced density matrix $\rho_{d,A}$ of the energy diagonal state $\rho_d$ in $A$ up to $\varepsilon/3$ error in trace distance:
    \begin{equation}
        \operatornamewithlimits{Pr}_t\left[D(\psi_A(t),\rho_{d,A})< \frac{\varepsilon}{3}\right]>1- \frac{9}{\varepsilon^2} 2^{-\gamma n+2|A|}.
    \end{equation}
    Next, the temporal-thermal equivalence implies that $\rho_{d,A}$ can be approximated by the reduced density matrix $g_{\beta,A}$ of $g_\beta$ in $A$ as
    \begin{equation}
        D(\rho_{d,A},g_{\beta,A}) \leq \frac{\varepsilon}{3}. 
    \end{equation}
    These imply that $\psi$ thermalizes under $H$ on $A$ with error $2\varepsilon/3$ and probability larger than 
    \begin{equation}
        P_t=1- \frac{9}{\varepsilon^2} 2^{-\gamma n+2|A|}.
    \end{equation}
    Then, Theorem~\ref{thm:low-c-from-therm} gives the stated bound on the circuit complexity.
\end{proof}

\subsection{Examples with low subsystem complexity}
In Theorem~\ref{thm:theorem3}, we show that late-time reduced states have small circuit complexity, provided that the initial state satisfies temporal-thermal equivalence and is energy dispersed. We now describe situations in which these conditions can be rigorously proven.

First, temporal-thermal equivalence under generic local Hamiltonians is expected to hold for typical initial product states. Indeed, it has been rigorously established in several settings. Specifically, it has been proven for generic translation-invariant systems at sufficiently high temperature~\cite{huang2020,Pilatowsky-Cameo2025}, a consequence of the weak eigenstate thermalization hypothesis (ETH)~\cite{Biroli2010,Mori2016,Brandao2019qec,Alhambra2020}. We introduce a version of the weak ETH as follows.

\begin{definition}[Weak eigenstate thermalization hypothesis]
    Consider a Hamiltonian $H$ at inverse temperature $\beta$ with eigenstates $\{\ket{\mu}\}_\mu$. Let $\{p_\mu\}_\mu$ be a probability distribution over eigenstates. For $\varepsilon,\eta>0$, we say that $H$ satisfies the {\it $(\varepsilon,\eta)$-weak eigenstate thermalization hypothesis} on a region $A$ with respect to $\{p_\mu\}_\mu$ if for any Hermitian observable $O_A$ supported in $A$ with $\|O_A\|_\mathrm{op}\leq 1$, its eigenstates satisfy
    \begin{equation}\label{eq:weak-ETH}
        \operatornamewithlimits{Pr}_{\mu\sim p_\mu}\left[\abs{\expval{O_A}{\mu}-\tr(O_A g_\beta)}\geq \varepsilon\right]\leq \eta.
    \end{equation}
\end{definition}

\begin{lemma}[Temporal-thermal equivalence from weak ETH]\label{thm:Gibbs-ED}
    Consider a Hamiltonian $H$ with inverse temperature $\beta$, eigenstates $\{\ket{\mu}\}_\mu$, and an initial state $\psi$. For $\varepsilon,\eta>0$, assume that $H$ satisfies the $(\varepsilon,\eta)$-weak ETH on a region $A$ with respect to $\{\abs{\braket{\mu}{\psi}}^2\}_\mu$. Then, $\psi$ satisfies $(\eta+\varepsilon/2)$-approximate temporal-thermal equivalence under $H$ on $A$.
\end{lemma}
\begin{proof}
    We denote $p_\mu$ as the energy distribution $\abs{\braket{\mu}{\psi}}^2$ and $\rho_d$ as the energy diagonal state of $\psi$. Due to the $(\varepsilon,\eta)$-weak ETH, we have
    \begin{equation}
        \operatornamewithlimits{Pr}_{\mu\sim p_\mu}\left[\abs{\expval{O_A}{\mu}-\tr(O_A g_\beta)}\geq \varepsilon\right]\leq \eta.
    \end{equation}
    This can then be used to obtain
    \begin{equation}
        \abs{\tr(O_A(\rho_d-g_\beta))} = \abs{\mathbb{E}_{\mu\sim p_\mu}\left[\expval{O_A}{\mu}-\tr(O_A g_\beta)\right]} \leq 2\eta+\varepsilon.
    \end{equation}
    Here, we use the definition of the energy diagonal state, $\rho_d=\sum_\mu p_\mu \ketbra{\mu}$. Since this is valid for all observables supported in $A$ having unit operator norm, by choosing $O_A$ as the optimal distinguisher, we have
    \begin{equation}
        \left\|\tr_{A^c}(\rho_d-g_\beta)\right\|_1 \leq 2\eta+\varepsilon.
    \end{equation}
\end{proof}

With this lemma, the late-time complexity of a reduced state of a state satisfying the weak ETH is guaranteed to be small as follows.

\begin{corollary}[Small late-time subsystem complexity from weak ETH]
    Let $H$ be a uniformly clustering Hamiltonian at inverse temperature $\beta$ with non-degenerate spectral gaps and eigenstates $\{\ket{\mu}\}_\mu$. Let $\psi$ be an $n$-qubit initial state that is energy dispersed with constant $\gamma$. Assume that $H$ satisfies the $(\frac{2\varepsilon}{9},\frac{2\varepsilon}{9})$-weak ETH on a region $A$ with respect to $\{\abs{\braket{\mu}{\psi}}^2\}_\mu$ for some positive constant $\varepsilon<1$. Then, we have
    \begin{equation}
        \operatornamewithlimits{Pr}_{t}\Big[\mathcal{C}_\varepsilon(\psi_A(t))\leq  e^{O(\log^D(|A|/\varepsilon))}\Big]>1- O(\varepsilon^{-2}2^{-\gamma n+2|A|}).
    \end{equation}
\end{corollary}
\begin{proof}    
    Under the $(\frac{2\varepsilon}{9},\frac{2\varepsilon}{9})$-weak ETH assumption, Lemma~\ref{thm:Gibbs-ED} implies that $\psi$ satisfies $(\varepsilon/3)$-approximate temporal-thermal equivalence under $H$ on $A$. Then, this together with Theorem~\ref{thm:theorem3} gives the stated bound.
\end{proof}

 For random product states at infinite temperature, as we discussed in Sec.~\ref{sec:03}, Ref.~\cite{huang2020} proved that the energy dispersion condition typically holds. The same reference also shows the temporal-thermal equivalence for translation-invariant systems, so we obtain the following result on complexity.

\begin{corollary}[Random product states have small late-time subsystem complexity]
\label{thm:rand-prod-states-small-complexity}
    Let $H$ be an $n$-qubit local extensive translation-invariant Hamiltonian with non-degenerate spectral gaps. For $\alpha>0$, consider a $D$-dimensional cubic region $A$ with size $|A|\leq(\frac{1}{2}-\alpha)\log_2 n$. Then, for positive constants $\nu<\alpha<1/2$, $\eta<\log(3/2)$, and $\varepsilon<1$, with probability at least $1-\kappa-\exp(-\eta n)$ over the choice of $\psi\sim\mathcal{E}_0$ [defined in Eq.~\eqref{eq:rand-prod-state}],  
    \begin{equation}
        \operatornamewithlimits{Pr}_{t}\Big[\mathcal{C}_\varepsilon(\psi_A(t))=0\Big]>1- 9\varepsilon^{-2}2^{-\gamma n}
    \end{equation}
    with $\kappa\leq O(n^{-\nu})$ and $\gamma < (\log(3/2)-\eta)/\log 2$.
\end{corollary}
\begin{proof}
    We follow an argument used in the proof of Theorem 2 of Ref.~\cite{huang2020}. Let $O_A$ denote a traceless observable supported on $A$ with $\|O_A\|_\mathrm{op}\leq 1$ and define its translation average by $\bar{O}_A=\frac{1}{n}\sum_{{\bf x}\in\Lambda}\mathbb{T}^{\bf x}O_A\mathbb{T}^{-\bf x}$. Since $A$ is a $D$-dimensional cube, we have 
    \begin{equation}
        \frac{1}{2^n}\tr(\bar{O}_A^\dag\bar{O}_A)
        = \frac{1}{2^n n}\sum_{{\bf x}\in\Lambda}\tr(O_A^\dag\mathbb{T}^{\bf x}O_A\mathbb{T}^{-\bf x})
        \leq \frac{2^D|A|}{n}.
    \end{equation}
    Here, we used the fact that the overlap vanishes whenever the supports of $O_A$ and $\mathbb{T}^{\bf x}O_A\mathbb{T}^{-\bf x}$ are disjoint.
    Since eigenstates are translation-invariant, we have $\expval{O_A}{\mu}=\expval{\bar{O}_A}{\mu}$ for any eigenstate $\ket{\mu}$ and the following bound~\cite{Biroli2010,Keating2015}:
    \begin{equation}
        \frac{1}{2^n}\sum_\mu\abs{\expval{O_A}{\mu}}^2 \leq \frac{1}{2^n}\tr(\bar{O}_A^\dag\bar{O}_A)\leq\frac{2^D|A|}{n}.
    \end{equation}
    Since $\tr(O_Ag_0)=\tr(O_A)=0$, where $g_0=I/2^n$ is the infinite-temperature state, Lemma~\ref{thm:weak-ETH-tracenorm} gives the following temporal-thermal equivalence:
    \begin{equation}
        \operatornamewithlimits{Pr}_{\psi\sim\mathcal{E}_0}\left[\frac{1}{2}\left\|\tr_{A^c}(\rho_d-g_0)\right\|_1 \geq \frac{\varepsilon}{3}\right] \leq 3\times 2^{|A|+1}\varepsilon^{-1}\sqrt{\frac{2^D|A|}{n}}.
    \end{equation}
    Then, Theorem~\ref{thm:theorem3} implies that with probability at least $1-3\times 2^{|A|+1}\varepsilon^{-1}\sqrt{2^D|A|/n}$ over the choice of $\psi\sim\mathcal{E}_0$, if $\psi$ is energy dispersed with $\gamma$, we have
    \begin{equation}
        \operatornamewithlimits{Pr}_{t}\Big[\mathcal{C}_\varepsilon(\psi_A(t))=0\Big]>1- 9\varepsilon^{-2}2^{-\gamma n+2\abs{A}}.
    \end{equation}
    Here, the resulting upper bound on the circuit complexity comes from the fact that the infinite temperature state is a product state, and any product state has zero circuit complexity by definition.
    
    Finally, we fix $\alpha>\nu>0$ and set $|A|\leq(\frac{1}{2}-\alpha)\log_2 n$. Then, we get $\kappa=3\times 2^{|A|+1}\varepsilon^{-1}\sqrt{2^D|A|/n}\leq O(n^{-\nu})$. Additionally, due to Proposition~\ref{thm:inf-energy-dispersion}, $\psi$ sampled from $\mathcal{E}_0$ is energy dispersed with $\gamma$ with probability higher than $1-\exp(-\eta n)$. Combining these gives the stated bound and the claimed success probability.
\end{proof}

Reference~\cite{Pilatowsky-Cameo2025} extends these results below infinite temperature, but for a different kind of ensemble compared to the ones discussed in Sec.~\ref{sec:Ensemblesofproductstates}. Specifically, consider the stabilizer product state ensemble $\mathcal{E}_\beta^\mathrm{(stab)}$ generated by the high-temperature Gibbs sampling algorithm of Ref.~\cite{bakshi2025}. By construction, this ensemble satisfies $\mathbb{E}_{\psi\sim\mathcal{E}_\beta^\mathrm{(stab)}}[\ketbra{\psi}]=g_\beta$, for inverse temperatures below the  threshold:
\begin{equation}\label{eq:betasep}
    \beta_\mathrm{s}=1/(100\mathcal{V}D).
\end{equation}
As recognized by Ref.~\cite{huang2024randomproduct}, a state sampled from this ensemble is energy dispersed with high probability when $|\beta|< (\beta_\mathrm{s}/2)(2/3)^\mathcal{V}$, a condition which is satisfied for $|\beta|<\beta_\mathrm{c}$ [defined in Eq.~\eqref{eq:betacrit}].

\begin{lemma}[Stabilizer product states are energy dispersed \cite{huang2024randomproduct}]\label{thm:stab-energy-dispersion}
    Let $H$ be a Hamiltonian satisfying (P1)-(P3) with inverse temperature $|\beta|<\beta_\mathrm{c}$. Then, for $0<\eta<\eta_\mathrm{max}\coloneq\log(3/2)-\log(1+\frac{|\beta|}{\beta_\mathrm{s}}\left(\frac{3}{2}\right)^\mathcal{V})$, with probability at least $1-e^{-\eta n}$, a state sampled from $\mathcal{E}_\beta^\mathrm{(stab)}$ is energy dispersed with $\gamma<(\eta_\mathrm{max}-\eta)/\log 2$.
\end{lemma}
\begin{proof}
    We explicitly bound the energy dispersion, as compared with Lemma 4 of Ref.~\cite{huang2024randomproduct} which only provides asymptotic statements. 
   
    Below, we restrict to $\beta\geq 0$. The corresponding statement for $\beta<0$ follows immediately by replacing $H$ with $-H$. Let us fix an output configuration $\chi$ of Algorithm 4.14 of Ref.~\cite{bakshi2025}, and let $\mathcal{E}_\chi$ denote the stabilizer product state ensemble associated with $\chi$. Then, we have
    \begin{equation}
        \mathbb{E}[\sigma(\chi)] \propto g_\beta,
    \end{equation}
    where $\sigma(\chi)$ is 
    \begin{equation}
        \sigma(\chi) = I^{\otimes\alpha}\otimes\bigotimes_{\{c_i,X_i\}\in\chi}(I+c_iX_i)
    \end{equation}
    for some integer $\alpha\geq 0$. Here, $X_i$ is a non-identity Hermitian Pauli monomial, and the supports of $\{X_i\}$ are disjoint. Let $t_i$ denote the degree of $X_i$. Then, Theorem 4.18 of Ref.~\cite{bakshi2025} gives $|c_i|\leq(\beta/\beta_\mathrm{s})^{t_i}$. Since $H$ satisfies (P1), it is an at most $\mathcal{V}$-local Hamiltonian, so $|X_i|$ is upper bounded by $\mathcal{V}t_i$. For each $i$, we can decompose $I+c_iX_i$ as 
    \begin{equation}
        I+c_iX_i = (1-|c_i|)I + |c_i|(I+\frac{c_i}{|c_i|}X_i).
    \end{equation}
    Following Lemma 4.19 of Ref.~\cite{bakshi2025}, the operator $I+\frac{c_i}{|c_i|}X_i$, up to normalization, can be represented by a convex sum of pure stabilizer states. Thus, with probability $1-|c_i|$, the algorithm samples a uniformly random stabilizer state, whose average gives $I$, while with probability $|c_i|$, it samples a stabilizer state from a distribution biased toward $I+\frac{c_i}{|c_i|}X_i$. The uniformly random stabilizer state contributes a factor of $(2/3)$ per site to the upper bound on $(D_\psi^\mathrm{eff})^{-1}$, which is the smallest possible contribution.
    
    Now, let us consider the worst case $\alpha = 0$, which ultimately yields the largest upper bound on $(D_\psi^\mathrm{eff})^{-1}$. Then, for a state $\psi$ sampled from $\mathcal{E}_\chi$ with $\alpha=0$, $(D_\psi^\mathrm{eff})^{-1}$ is upper bounded by
    \begin{equation}
        \mathbb{E}_{\psi\sim\mathcal{E}_\chi}[(D_\psi^\mathrm{eff})^{-1}]\leq \prod_i\left(\left(\frac{2}{3}\right)^{|X_i|}(1-\abs{c_i})+\abs{c_i}\right)=\left(\frac{2}{3}\right)^n\prod_i\left(1-\abs{c_i}+\abs{c_i}\left(\frac{3}{2}\right)^{|X_i|}\right).
    \end{equation}
    When $\beta/\beta_\mathrm{s}<(2/3)^\mathcal{V}$, this implies that
    \begin{equation}
        \mathbb{E}_{\psi\sim\mathcal{E}_\chi}[(D_\psi^\mathrm{eff})^{-1}]\leq \left(\frac{2}{3}\right)^n\prod_i\left(1+\left(\frac{\beta}{\beta_\mathrm{s}}\right)^{t_i}\left(\frac{3}{2}\right)^{\mathcal{V}t_i}\right)\leq \left(\frac{2}{3}\right)^n\left(1+\frac{\beta}{\beta_\mathrm{s}}\left(\frac{3}{2}\right)^\mathcal{V}\right)^n.
    \end{equation}
    Since this bound is independent of the configuration $\chi$, it applies uniformly to all possible configurations. This together with Markov's inequality gives
    \begin{equation*}
        \operatornamewithlimits{Pr}_{\psi\sim\mathcal{E}_\beta^\mathrm{(stab)}} \left[D_\psi^\mathrm{eff}< 2^{\gamma n}\right]< e^{-[\log(3/2)-\log(1+\frac{\beta}{\beta_\mathrm{s}}\left(\frac{3}{2}\right)^\mathcal{V})-\gamma\log 2]n}.\qedhere
    \end{equation*}
\end{proof}

Since stabilizer product states are energy dispersed with high probability due to Lemma~\ref{thm:stab-energy-dispersion}, and since temporal-thermal equivalence holds for translation-invariant systems at sufficiently high temperature~\cite{Pilatowsky-Cameo2025}, Theorem~\ref{thm:theorem3} implies that, under translation-invariant Hamiltonian dynamics, small subsystems have small circuit complexity at late times, as follows.

\begin{corollary}[Stabilizer product states have small late-time subsystem complexity]
\label{thm:small-complexity-stab-prod-states}
    Let $H$ be a bounded local translation-invariant Hamiltonian for $n$ qubits in a $D$-dimensional lattice with non-degenerate spectral gaps. Assume that the inverse temperature is upper bounded as $|\beta|<\beta_\mathrm{c}$ [defined in Eq.~\eqref{eq:betacrit}]. For $\alpha>0$, consider a subsystem $A$ with size $|A|\leq(\frac{1}{D+1}-\alpha)\log_2 n$. Then, for positive constants $\nu<\alpha<1/(D+1)$, $\eta<\log(3/2)-\log(1+\frac{|\beta|}{\beta_\mathrm{s}}\left(\frac{3}{2}\right)^\mathcal{V})$ [defined in Eq.~\eqref{eq:betasep}], and $\varepsilon<1$, with probability at least $1-\kappa-\exp(-\eta n)$ over the choice of $\psi\sim\mathcal{E}_\beta^\mathrm{(stab)}$, 
    \begin{equation}
        \operatornamewithlimits{Pr}_{t}\Big[\mathcal{C}_\varepsilon(\psi_A(t))\leq  e^{O(\log^D(|A|/\varepsilon))}\Big]>1- 9\varepsilon^{-2}2^{-\gamma n}
    \end{equation}
    with $\kappa\leq O(n^{-\nu})$ and $\gamma < (\log(3/2)-\log(1+\frac{|\beta|}{\beta_\mathrm{s}}\left(\frac{3}{2}\right)^\mathcal{V})-\eta)/\log 2$.
\end{corollary}
\begin{proof}
    Let $q_\mu$ denote the Gibbs distribution $\expval{g_\beta}{\mu}$ in the energy basis. By Refs.~\cite{Mori2016,Brandao2019qec,Alhambra2020,Pilatowsky-Cameo2025}, for any region $A$ satisfying $|A|\leq o(n^{1/(D+1)})$, $\lambda>0$, and $\zeta=e^{-\Omega(n^{1/(D+1)}\lambda)}$, the Hamiltonian $H$ satisfies the $(\lambda,\zeta^2)$-weak ETH on $A$ with respect to $\{q_\mu\}_\mu$. That is, for every observable $O_A$ supported on $A$ with $\|O_A\|_\mathrm{op}\leq 1$,
    \begin{equation}
        \mathrm{Pr}_{\mu\sim q_\mu}\left[\abs{\expval{O_A}{\mu}-\tr(O_A g_\beta)}\geq \lambda\right]\leq \zeta^2.
    \end{equation}
    This implies that
    \begin{equation}
        \sum_\mu q_\mu (\expval{O_A}{\mu}-\tr(O_A g_\beta))^2 \leq 4\zeta^2+\lambda^2.
    \end{equation}    
    Lemma~\ref{thm:weak-ETH-tracenorm} below gives the following temporal-thermal equivalence:
    \begin{equation}
        \operatornamewithlimits{Pr}_{\psi\sim\mathcal{E}_\beta^\mathrm{(stab)}}\left[\frac{1}{2}\left\|\tr_{A^c}(\rho_d-g_\beta)\right\|_1 \geq \frac{\varepsilon}{3}\right] \leq 3\times 2^{|A|+1}\varepsilon^{-1}(2\zeta+\lambda).
    \end{equation}
    
    We can now apply Theorem~\ref{thm:theorem3}. It implies that, with probability at least $1-3\times 2^{|A|+1}\varepsilon^{-1}(2\zeta+\lambda)$, if $\psi\sim\mathcal{E}_\beta^\mathrm{(stab)}$ is energy dispersed with $\gamma$, then it satisfies
    \begin{equation}
        \operatornamewithlimits{Pr}_{t}\Big[\mathcal{C}_\varepsilon(\psi_A(t))\leq  e^{O(\log^D(|A|/\varepsilon))}\Big]>1- 9\varepsilon^{-2}2^{-\gamma n+2\abs{A}}.
    \end{equation}
    
    Finally, we choose $\alpha>\nu>0$, $|A|\leq(\frac{1}{D+1}-\alpha)\log_2 n$, and $\lambda=n^{-1/(D+1)+\alpha-\nu}$. With these choices, we have $\kappa=3\times 2^{|A|+1}\varepsilon^{-1}(2\zeta+\lambda)\leq O(n^{-\nu})$. Moreover, Lemma~\ref{thm:stab-energy-dispersion} implies that $\psi\sim\mathcal{E}_\beta^\mathrm{(stab)}$ is energy dispersed with $\gamma$ with probability at least $1-\exp(-\eta n)$. Altogether, we get the statement of the corollary.
\end{proof}

\begin{lemma}[Temporal-thermal equivalence from a variant of weak ETH]\label{thm:weak-ETH-tracenorm}
    Consider a Hamiltonian $H$ with inverse temperature $\beta$ and eigenstates $\{\ket{\mu}\}_\mu$. Additionally, consider an ensemble $\mathcal{E}$ of product states whose first moment $\mathbb{E}_{\psi\sim\mathcal{E}}[\ketbra{\psi}]$ equals the Gibbs state $g_\beta$. We denote $q_\mu$ as the Gibbs distribution $\expval{g_\beta}{\mu}$. For $\varepsilon,\eta>0$, if $H$ satisfies 
    \begin{equation}\label{eq:weak-ETH-tracenorm}
        \sum_\mu q_\mu (\expval{O_A}{\mu}-\tr(O_A g_\beta))^2 \leq \eta^2
    \end{equation}
    for any observable $O_A$ supported in $A$ with $\|O_A\|_\mathrm{op}\leq 1$, then $\psi\sim\mathcal{E}$ satisfies temporal-thermal equivalence under $H$ on $A$, up to $\varepsilon$ error in trace distance, with probability at least $1-\varepsilon^{-1}\eta 2^{|A|+1}$.
\end{lemma}
\begin{proof}
    The weak ETH type assumption in Eq.~\eqref{eq:weak-ETH-tracenorm} can be used to bound the mean trace norm of $\tr_{A^c}(\ketbra{\mu}-g_\beta)$ as
    \begin{equation}
        \begin{split}
            \sum_\mu q_\mu\left\|\tr_{A^c}(\ketbra{\mu}-g_\beta)\right\|_1 
            &\leq 2^{|A|/2}\sum_\mu q_\mu\|\tr_{A^c}(\ketbra{\mu}-g_\beta)\|_2 \\
            &\leq 2^{|A|/2}\left(\sum_\mu q_\mu\|\tr_{A^c}(\ketbra{\mu}-g_\beta)\|_2^2\right)^{1/2}\\
            &= \left(\sum_\mu q_\mu\sum_{k\in\{I,X,Y,Z\}^{|A|}}\tr(P_k(\ketbra{\mu}-g_\beta))^2\right)^{1/2}\\
            &\leq 2^{|A|}\eta.
        \end{split}
    \end{equation}
    We then use the following fact from the proof of Proposition S4 of Ref.~\cite{Pilatowsky-Cameo2025}:
    \begin{equation}
        \mathbb{E}_{\psi\sim\mathcal{E}}\left[\left\|\tr_{A^c}(\rho_d-g_\beta)\right\|_1 \right] \leq \sum_\mu q_\mu \left\|\tr_{A^c}(\ketbra{\mu}-g_\beta)\right\|_1.
    \end{equation}
    This, combined with Markov's inequality, gives
    \begin{equation*}
\operatornamewithlimits{Pr}_{\psi\sim\mathcal{E}}\left[\frac{1}{2}\left\|\tr_{A^c}(\rho_d-g_\beta)\right\|_1 \geq \varepsilon\right] \leq 2^{|A|+1}\varepsilon^{-1}\eta. \qedhere
    \end{equation*}
\end{proof}
We remark that \Cref{thm:rand-prod-states-small-complexity,thm:small-complexity-stab-prod-states} prove a small late-time complexity for regions with $\abs{A}\sim \log(n)$. We expect, however, that a small late-time complexity should hold for extensive regions $\abs{A}\sim\eta n$ as long as the constant $\eta$ is below a certain fraction smaller than $1/2$. Proving this expectation is an important future direction of our work.

\section{Finite time}
\label{sec:07}
In Theorem~\ref{th:theoremS2}, we show that late-time states generically have exponential circuit complexity. In this section, we investigate the timescale sufficient to reach such exponential circuit complexity. To this end, we introduce the \textit{$k$-th spectral gap}:
\begin{equation}
    \Delta^{(k)} = \min_{\{m_1,\cdots,m_N\}} \abs{\sum_{\mu=1}^N m_\mu E_\mu},
\end{equation}
where the minimum runs over all $\{m_1,\cdots,m_N\}\in\mathbb{Z}^N$ which are multiplicity vectors, i.e., such that $\sum_{\mu=1}^N m_\mu = 0$ and $0<\sum_{\mu=1}^N\abs{m_\mu}\leq 2k$ [introduced in the proof of Lemma~\ref{thm:local-ergodic}]. In the next subsection, we show that an exponential circuit complexity is guaranteed to arise after time $(1/\Delta^{(k_c)})\exp(\exp(O(n)))$,  with $k_c=\exp(O(n))$.

We remark that a double-exponential scaling is the strongest bound that can be obtained using any counting argument. With a covering net argument as used in Theorem~\ref{th:theoremS2}, one cannot prove that the trajectory $\psi(t)$ fails to be covered by a net with $s$ elements unless $\psi(t)$ explores more than $\mathcal{O}(s)$ distinguishable states. When the gate-count cutoff $G$ is exponentially large, $s$ becomes double-exponentially large. The  timescale for resolving such states is bounded by the quantum speed limit, which implies that the time required to evolve to a constant-overlap state is at least $1/\mathrm{poly}(n)$~\cite{Anandan1990,MARGOLUS1998,Taddei2013}. Consequently, this covering net argument can certify that most of the trajectory lies outside the low-complexity net only on double-exponentially long time intervals. In the proofs of complexity growth on random circuits~\cite{Haferkamp2022,haferkamp2023,haah2025}, this limitation is bypassed since the dynamics generate double-exponentially many distinguishable trajectories at late times, a feature which is not available when considering dynamics under a single fixed Hamiltonian.

\subsection{Proof of Theorem 5}
\begin{theorem}[Exponential quantum circuit complexity at finite time; Theorem 5 in the main text]
    Let $H$ have spectral ergodicity and $\psi$ be energy dispersed with constant $\gamma>0$ under $H$. Let  $\Delta^{(k)}$ be the $k$-th spectral gap of $H$,  $0<\varepsilon<1$, and
    \begin{equation}
        T\geq\frac{2}{\Delta^{(k_c)}}\exp(\exp(\frac{3\log 2}{4}\gamma n)),
    \end{equation}
    with $k_c=\left\lceil\exp(\frac{\log 2}{2}\gamma n)\right\rceil$. Then
    \begin{equation}
        \operatornamewithlimits{Pr}_{t\sim[0,T]}\Big[\mathcal{C}_\varepsilon(\psi(t))\leq  \exp(\frac{\gamma}{8}n)\Big]\leq \exp(-\exp({\frac{\gamma}{8}\,n}))
    \end{equation}
    for  $n\geq n_*\coloneq\frac{16}{\gamma}\big(7+\log\tfrac{2}{\gamma(1-\varepsilon)}\big)$.
\end{theorem}
\begin{proof}
Our starting point is Eq.~\eqref{eq:complexity-from-covering-size} with $|A|=n$, where the infinite time average is replaced by the finite time average over $\mathcal{I}=[0,T]$:
\begin{equation}
    \begin{split}
        \operatornamewithlimits{Pr}_{t\sim\mathcal{I}}[\mathcal{C}_\varepsilon(\psi(t)|\phi)\leq G]
        &\leq ((n+2G)^2(1+512G/\lambda)^{256})^G\operatornamewithlimits{Pr}_{t\sim\mathcal{I}}\left[D(\psi(t),\sigma_i)\leq\frac{1}{2}(1+\varepsilon)\right]\\
        &\leq 4^k((n+2G)^2(1+512G/\lambda)^{256})^G\frac{\mathbb{E}_{t\in\mathcal{I}}[\expval{\sigma}{\psi(t)}^k]}{(1-\varepsilon)^{2k}}
    \end{split}
\end{equation}
where we use $\frac{1}{2}(1+\varepsilon)\geq \varepsilon+\lambda$ in the first line and Markov's inequality in the second line. Let us rewrite $\mathbb{E}_{t\in\mathcal{I}}[\expval{\sigma}{\psi(t)}^k]$ into
\begin{equation}
    \mathbb{E}_{t\in\mathcal{I}}[\expval{\sigma}{\psi(t)}^k] = \underbrace{\left(\mathbb{E}_{t\in\mathcal{I}}[\expval{\sigma}{\psi(t)}^k]-\mathbb{E}_t[\expval{\sigma}{\psi(t)}^k]\right)}_{\coloneq f(T)}+\mathbb{E}_t[\expval{\sigma}{\psi(t)}^k].
\end{equation}
We present an upper bound on $\mathbb{E}_t[\expval{\sigma}{\psi(t)}^k]$ in Lemma~\ref{thm:k-th-fidelity-bound}:
\begin{equation}\label{eq:inf-time-avg-fidelity-bound}
    \mathbb{E}_t[\expval{\sigma}{\psi(t)}^k]\leq 2^{-\gamma nk/2}k!.
\end{equation}
We then build an analogous bound for the off-resonance terms $f(T)$. Following the notation in the proof of \Cref{thm:no-k-res-temp-avg-symm}, $f(T)$ is upper bounded by
\begin{equation}
    \begin{split}
        f(T)&=\frac{1}{T}\int_0^T dt\sum_{\bm \nu}\sum_{{\bm\mu}\notin [{\bm\nu}]}e^{-i\sum_{j=1}^k (E_{\mu_j}-E_{\nu_j})t}\bra{\bm\nu}\sigma^{\otimes k}\ketbra{\bm\mu}(\ketbra{\psi})^{\otimes k}\ket{\bm\nu}\\
        &\leq \sum_{\bm \nu}\sum_{{\bm\mu}\notin [{\bm\nu}]}\abs{\frac{2\sin(T\sum_{j=1}^k(E_{\mu_j}-E_{\nu_j})/2)}{T\sum_{j=1}^k (E_{\mu_j}-E_{\nu_j})}} \abs{\bra{\bm\nu}\sigma^{\otimes k}\ketbra{\bm\mu}(\ketbra{\psi})^{\otimes k}\ket{\bm\nu}}\\
        &\leq \frac{2}{T\Delta^{(k)}} \sum_{{\bm\mu},{\bm\nu}}\abs{\bra{\bm\nu}\sigma^{\otimes k}\ketbra{\bm\mu}(\ketbra{\psi})^{\otimes k}\ket{\bm\nu}}\\
        &\leq \frac{2}{T\Delta^{(k)}} \left(\sum_{{\bm\mu},{\bm\nu}}\abs{\bra{\bm\nu}\sigma^{\otimes k}\ket{\bm\mu}}^2 \right)^{1/2}\left(\sum_{{\bm\mu},{\bm\nu}}\abs{\bra{\bm\mu}(\ketbra{\psi})^{\otimes k}\ket{\bm\nu}}^2\right)^{1/2},
    \end{split}
\end{equation}
where the first equality holds since for $\bm \mu\in [\bm \nu]$ with the orbit $[\bm \nu]$ of $\bm\nu$ over the action of the symmetric group, the phases already vanish regardless of $t$; the second line is a triangle inequality, and we use the spectral ergodicity when solving the integral; the third line holds by the definition of the $k$-th spectral gap and $|\sin(x)|\leq 1$; and the fourth line is the Cauchy-Schwarz inequality. Since the terms in the parentheses in the last line are the $k$-th powers of the purities of $\sigma$ and $\psi$, for any positive integer $k$, it follows that
\begin{equation}
    f(T) \leq \frac{2}{T\Delta^{(k)}}.
\end{equation}

Now, let us use the condition on $T$:
\begin{equation}
    T\geq \frac{2}{\Delta^{(k_c)}}\exp(\exp(\frac{3\log 2}{4}\gamma n)).
\end{equation} 
Additionally, let us consider the choice of $k_*$ in the proof of Lemma~\ref{th:temp-deloc}. We find that $k_*$ is upper bounded by
\begin{equation}
    k_*\leq\frac{1}{4}\exp(\frac{\log 2}{2}\gamma n)\leq k_c.
\end{equation}
By the definition of the $k$-th spectral gap, it is nonincreasing as $k$ increases. Therefore, we have $\Delta^{(k_c)}\leq\Delta^{(k_*)}$ and
\begin{equation}
    f(T)\leq\frac{2}{T\Delta^{(k_*)}}\leq \exp(-\exp(\frac{3\log 2}{4}\gamma n))\leq \exp(-\frac{\gamma n \log 2}{8} 2^{\gamma n/2})\leq 2^{-\gamma nk_*/2}\leq 2^{-\gamma nk_*/2} (k_*)!.
\end{equation} 
Then, Eq.~\eqref{eq:inf-time-avg-fidelity-bound} with $k=k_*$ implies that
\begin{equation}
    \mathbb{E}_{t\in\mathcal{I}}[F(\psi(t),\sigma_i)^{k_*}] \leq 2\times 2^{-k_*\gamma n/2}(k_*!).
\end{equation}
This upper bound is twice the upper bound of the $k$-th moments of the fidelity established in Eq.~\eqref{eq:k-th-fidelity-bound} with $k=k_*$. Following the same argument as in the proof of Theorem~\ref{th:theoremS2}, this extra factor of two contributes only an additive $\log2$ correction to Eq.~\eqref{eq:generalinequalitylog} and does not affect the subsequent upper bounds. Furthermore, we can adopt $\mathcal{C}_\varepsilon(\psi(t)|\phi)\rightarrow\mathcal{C}_\varepsilon(\psi(t))$ by taking $G\rightarrow G-\lceil n/2\rceil$ as discussed in the same proof. Therefore, we have 
\begin{equation}
    \operatornamewithlimits{Pr}_{t\sim\mathcal{I}}[\mathcal{C}_\varepsilon(\psi(t))\leq\exp(\gamma n/8)]\leq \exp(-\exp(\gamma n/8)).\qedhere
\end{equation}
\end{proof}

\subsection{Scaling of the $k$-th spectral gap}
We expect that the typical inverse scale of $\Delta^{(k)}$ is given by $1/\Delta^{(k)}\sim\exp(\exp(O(n)))$ for $k=\exp(O(n))$ based on the following heuristic argument~\cite{Pilatowsky-Cameo2025}. Suppose the single energy difference $\hat{\Delta}=E_\mu-E_{\mu'}$ with $\mu\neq\mu'$ follows a uniform distribution $\mathrm{Unif}[-W,W]$, where $W$ is upper bounded by $O(n)$. Define a random variable $\hat{X}$ as $\hat{X}=\sum_{j=1}^k\Delta_j$, where each $\Delta_j$ is an independent sample of $\hat{\Delta}$. Since there are $K=2^{kn}$ distinct ordered $k$-tuples of energies with repetition, for $k=\exp(O(n))$, we heuristically model $\Delta^{(k)}$, up to subleading factors, as the minimum absolute value among $K^2$ independent samples of $\hat{X}$. Since $\operatorname{Var}(\hat{\Delta})=W^2/3$, for large $k$, by the central limit theorem, $\hat{X}$ follows a normal distribution with zero mean and $kW^2/3$ variance. Therefore, the density of $\hat{X}$ at zero is approximately $\sqrt{3}/(W\sqrt{2\pi k})$. Consequently, for small $x>0$, using independence among the $K^2$ candidate differences, we have
\begin{equation}
    \operatornamewithlimits{Pr}(\Delta^{(k)}>x)\approx \left(1-\frac{\sqrt{6}}{W\sqrt{\pi k}}x\right)^{K^2}\approx \exp(-\frac{\sqrt{6}K^2}{W\sqrt{\pi k}}x).
\end{equation}
The characteristic decay scale of this exponential is given by the inverse of the exponent, $x\sim W\sqrt{k}K^{-2}$, which is double-exponentially small for $k=\exp(O(n))$. Therefore, the tail integral formula
\begin{equation}
    \mathbb{E}[\Delta^{(k)}]=\int_0^\infty dx\operatornamewithlimits{Pr}(\Delta^{(k)}>x)
\end{equation}
is dominated by this small-$x$ regime. Thus, we have
\begin{equation}
    \mathbb{E}[\Delta^{(k)}]\sim W\sqrt{k}2^{-2kn}.
\end{equation}
Consequently, for $k=\exp(O(n))$, the typical inverse scale of $\Delta^{(k)}$ is $1/\Delta^{(k)}\sim\exp(\exp(O(n)))$.

\bibliography{Ref}